\documentclass[%
 reprint,
superscriptaddress,
nofootinbib,
bibnotes,
 amsmath,amssymb,
pra,
]{revtex4-1}

\usepackage{graphicx}
\usepackage{dcolumn}
\usepackage{braket}
\usepackage{textcomp, gensymb}
\usepackage{bm}
\usepackage{amsmath}
\usepackage{float}
\usepackage{mathtools}
\usepackage{color}
\usepackage[dvipsnames]{xcolor}
\usepackage[utf8]{inputenc}
\usepackage{csquotes}
\usepackage{thmtools}
\declaretheorem[]{definition}
\declaretheorem[]{theorem}
\declaretheorem[]{property}

\newtheorem{lemma}{Lemma}
\newtheorem{corollary}{Corollary}[theorem]

\makeatletter
\DeclareRobustCommand{\qed}{%
  \ifmmode 
  \else \leavevmode\unskip\penalty9999 \hbox{}\nobreak\hfill
  \fi
  \quad\hbox{\qedsymbol}}
\newcommand{\openbox}{\leavevmode
  \hbox to.77778em{%
  \hfil\vrule
  \vbox to.675em{\hrule width.6em\vfil\hrule}%
  \vrule\hfil}}
\newcommand{\qedsymbol}{\openbox}
\newenvironment{proof}[1][\proofname]{\par
  \normalfont
  \topsep6\p@\@plus6\p@ \trivlist
  \item[\hskip\labelsep\itshape
    #1.]\ignorespaces
}{%
  \qed\endtrivlist
}
\newcommand{\proofname}{Proof}
\makeatother

\renewcommand{\thesubsection}{\thesection.\Roman{subsection}}

\newcommand{\commute}[2]{\{#1,#2\}_{\textrm{G}}}
\newcommand{\commuteO}[2]{\{#1,#2\}_{\textrm{O}}}
\newcommand{\commuteS}[2]{\{#1,#2\}_{\textrm{S}}}

\newcommand{\pool}{\mathcal{P}}    
\newcommand{\subpool}{\mathcal{S}} 
\newcommand{\nset}{\mathcal{N}}    
\newcommand{\layer}{\mathcal{A}}   

\newcommand{\support}{\mathrm{Supp}}   

\newcommand{\pauli}{\mathcal{P}}
\newcommand{\paulix}{\mathcal{X}}
\newcommand{\pauliy}{\mathcal{Y}}
\newcommand{\pauliz}{\mathcal{Z}}

\newcommand{\damp}{\mathcal{F}}
\newcommand{\dephase}{\mathcal{C}}
\newcommand{\depolarize}{\mathcal{D}}

\newcommand{\damplayer}{\bar{\damp}}
\newcommand{\dephaselayer}{\bar{\dephase}}
\newcommand{\depolarizelayer}{\bar{\depolarize}}

\newcommand{\residual}{\mathcal{R}}

\newcommand{\tmax}{t_{\text{max}}}
\newcommand{\mmax}{m_{s}}
\newcommand{\nmax}{n_{\text{max}}}

\newcommand{\bound}{\mathcal{E}}    
\newcommand{\accuracy}{\varepsilon}    
\newcommand{\Lvalue}{\mathcal{L}}   
\newcommand{\Lmax}{\ell} 

\newcommand{\dt}{\tau}    

\newcommand{\HF}{\text{HF}}

\DeclareMathOperator*{\argmin}{arg\,min}
\newcommand{\approptoinn}[2]{\mathrel{\vcenter{
  \offinterlineskip\halign{\hfil$##$\cr
    #1\propto\cr\noalign{\kern2pt}#1\sim\cr\noalign{\kern-2pt}}}}}

\newcommand{\appropto}{\mathpalette\approptoinn\relax}
\newcommand{\bigO}[1]{\mathcal O\left(#1\right)}

\newcommand{\bigOmega}[1]{\Omega\left(#1\right)}

\newcommand{\bigTheta}[1]{\Theta\left(#1\right)}

\newcommand\namebond[4][5pt]{\chemmove{\path(#2)--(#3)node[midway,sloped,yshift=#1]{#4};}}
\newcommand\arcbetweennodes[3]{%
\pgfmathanglebetweenpoints{\pgfpointanchor{#1}{center}}{\pgfpointanchor{#2}{center}}%
\let#3\pgfmathresult}
\newcommand\arclabel[6][stealth-stealth,shorten <=1pt,shorten >=1pt]{%
\chemmove{%
\arcbetweennodes{#4}{#3}\anglestart \arcbetweennodes{#4}{#5}\angleend
\draw[#1]([shift=(\anglestart:#2)]#4)arc(\anglestart:\angleend:#2);
\pgfmathparse{(\anglestart+\angleend)/2}\let\anglestart\pgfmathresult
\node[shift=(\anglestart:#2+1pt)#4,anchor=\anglestart+180,rotate=\anglestart+90,inner sep=0pt,
outer sep=0pt]at(#4){#6};}}

\usepackage{url}

\usepackage{breakurl}
\usepackage[breaklinks]{hyperref}
\hypersetup{pdfborder=0 0 0}

\usepackage[normalem]{ulem}

\usepackage{physics}
\usepackage[version=4]{mhchem}
\usepackage{chemfig}
\usepackage{dsfont}
\usepackage{tikz}
\usetikzlibrary{patterns, positioning, intersections}
\usepackage{tikz-3dplot}
\usepackage{qcircuit}
\usepackage{algorithm}
\usepackage{algorithmicx}
\usepackage[noend]{algpseudocode}
\usepackage{slashed}
\usepackage{natbib}
\usepackage{empheq}
\usepackage{cases}
\usepackage[all]{hypcap}
\usepackage{siunitx}
\usepackage[capitalize]{cleveref}
\usepackage{afterpage}
\usepackage{placeins}
\usepackage{listings}
\usepackage{pgfplots}
\usepackage{tabularx}

\usepgfplotslibrary{fillbetween}

\lstdefinestyle{inline}{basicstyle=\tt}

\DeclareSIUnit{\Hartree}{Ha}

\definecolor{myorange}{HTML}{FE6100}
\definecolor{myblue}{HTML}{64A5FF}
\definecolor{mypink}{HTML}{FF1972}
\definecolor{mygreen}{HTML}{62D09D}
\definecolor{myviolet}{HTML}{4F428E}
\definecolor{myyellow}{HTML}{FFBF00}

\renewcommand{\thesubsection}{\Roman{section} \Alph{subsection}}

\makeatletter
\def\p@subsection{}
\makeatother
\makeatletter
\def\p@subsubsection{}
\makeatother

\allowdisplaybreaks

\begin{document}

\title{Layering and subpool exploration for adaptive Variational 
Quantum Eigensolvers: \\ Reducing circuit depth, runtime, and 
susceptibility to noise}

\author{Christopher K. Long}
\affiliation{ Hitachi  Cambridge  Laboratory,  J.  J.  Thomson  Ave.,  
Cambridge,  CB3  0HE,  United  Kingdom
}
\affiliation{ Cavendish Laboratory,  Department of Physics,  University  of  
Cambridge,  Cambridge,  CB3  0HE,  United  Kingdom
}
\author{Kieran Dalton}
\affiliation{ Hitachi  Cambridge  Laboratory,  J.  J.  Thomson  Ave.,  
Cambridge,  CB3  0HE,  United  Kingdom
}
\affiliation{ Cavendish Laboratory,  Department of Physics,  University  of  
Cambridge,  Cambridge,  CB3  0HE,  United  Kingdom
}
\author{Crispin H. W. Barnes}
\affiliation{ Cavendish Laboratory,  Department of Physics,  University  of  
Cambridge,  Cambridge,  CB3  0HE,  United  Kingdom
}
\author{David R. M. Arvidsson-Shukur}
\affiliation{ Hitachi  Cambridge  Laboratory,  J.  J.  Thomson  Ave.,  
Cambridge,  CB3  0HE,  United  Kingdom
}
\author{Normann Mertig}
\affiliation{ Hitachi  Cambridge  Laboratory,  J.  J.  Thomson  Ave.,  
Cambridge,  CB3  0HE,  United  Kingdom
}

\date{\today}

\begin{abstract}
Adaptive variational quantum eigensolvers (ADAPT-VQEs) are promising candidates for simulations of strongly correlated systems on near-term quantum hardware.
To further improve the noise resilience of these algorithms, recent efforts have been directed towards compactifying, or \textit{layering}, their ansatz circuits.
Here, we broaden the understanding of the algorithmic layering process in three ways.
First, we investigate the non-commutation relations between the different elements that are used to build ADAPT-VQE ansätze. Doing so, we develop a framework for studying and developing layering algorithms, which produce shallower circuits.
Second, based on this framework, we develop a new subroutine that can reduce the number of quantum-processor calls by optimizing the selection procedure with which a variational quantum algorithm appends ansatz elements.
Third, we provide a thorough numerical investigation of the noise-resilience improvement available via layering the circuits of ADAPT-VQE algorithms. We find that layering leads to an improved noise resilience with respect to amplitude-damping and dephasing noise, which, in general, affect idling and non-idling qubits alike. With respect to depolarizing noise,  which tends to affect only actively manipulated qubits, we observe no advantage of layering.
\end{abstract}

\maketitle

\section{Introduction}

Quantum chemistry simulations of strongly correlated systems are challenging for classical computers \cite{QCCReview}.
While approximate methods often lack accuracy \cite{QCCReview, Self-consistent_field_with_exchange_for_beryllium, PhysRev.140.A1133, PhysRev.136.B864, ROSSI1999530}, exact methods become infeasible when the system sizes exceed more than 34 spin orbitals---the largest system for which a full configuration interaction (FCI) calculation has been conducted
\cite{ROSSI1999530}.
For this reason, simulations of many advanced chemical systems, such as enzyme active sites and surface catalysts, rely on knowledge-intense, domain-specific approximations \cite{doi:10.1021/acs.chemrev.9b00829}.
Therefore, developing general chemistry simulation methods for quantum computers could prove valuable.

Variational quantum eigensolvers (VQEs) \cite{Peruzzo2014, QCCReview, TILLY20221, UCCSD, ADAPT-VQE, QEB-ADAPT-VQE, eQEB-ADAPT-VQE, qubit-ADAPT-VQE, KieransNoise, Tetris-ADAPT-VQE, L-VQE} are a class of quantum-classical methods intended to perform chemistry simulations on near-term quantum hardware.
More specifically, VQEs calculate upper bounds to the ground state energy $E_0$ of a molecular Hamiltonian $H$ using the Rayleigh-Ritz variational principle 
\begin{align}
    E_0\le E(\vec\theta)\equiv\Trace\left(H 
\Lambda(\vec\theta)[\rho_0]\right).
\end{align}
A quantum processor is used to apply a parametrized quantum circuit to an initial state. In the presence of noise the quantum circuit can, in general, be represented by the parameterized completely positive trace-preserving (CPTP) map $\Lambda(\vec\theta)$ and the initial state can be represented by the density matrix $\rho_0$. We will use square brackets to enclose a state acted upon by a CPTP map.
This generates a parametrized trial state $\rho(\vec\theta) \equiv 
\Lambda(\vec\theta) [\rho_0]$ that is hard to represent on classical computers.
The energy expectation value $E(\vec\theta)$ of $\rho(\vec\theta)$ gives a bound on $E_0$, which can be accurately sampled using polynomially few measurements \cite{QCCReview, TILLY20221}.
A classical computer then varies $\vec\theta$ to minimize $E(\vec\theta)$ 
iteratively.
Provided that the ansatz circuit is sufficiently expressive, $E(\vec\theta)$ converges to $E_0$ and returns the ground state energy.
Initial implementations of VQEs on near-term hardware have been reported in 
\cite{Peruzzo2014,doi:10.1126/science.abb9811,PhysRevX.6.031007,Kandala2017,PhysRevX.8.031022,Xue2022}.
Despite these encouraging results, several refinements are needed to alleviate trainability issues \cite{BarrenPlateau,  BarrenPlateauGradFree, BarrenPlateauHighDeriv, NoiseBarrenPlateau} and to make VQEs feasible for molecular simulations with larger numbers of orbitals.
Moreover, recent results indicate that the noise resilience of VQE algorithms must be improved to enable useful simulations \cite{KieransNoise, NoiseBarrenPlateau, De_Palma_2023}.

Adaptive VQEs (ADAPT-VQEs) \cite{ADAPT-VQE} are promising VQE algorithms, which partially address the issues of trainability and noise resilience.
They operate by improving the ansatz circuits in $\tmax$ consecutive steps
\begin{align}
 \Lambda_t(\theta_t,\ldots,\theta_1) = A_t(\theta_t) \circ 
\Lambda_{t-1}(\theta_{t-1},\ldots,\theta_1),
\end{align}
starting from the identity map $\Lambda_0=\text{\normalfont id}$.
Here, $t=1,...,\tmax$ indexes the step and $\circ$ denotes functional composition of the CPTP maps.
An ansatz element $A_t(\theta_t)$ is added to the ansatz circuit in each step.
The ansatz element $A_t(\theta_t)$ is chosen from an ansatz-element pool $\pool$ by computing the energy gradient for each ansatz element and picking the ansatz element with the steepest gradient.
Numerical evidence suggests that such ADAPT-VQEs are readily trainable and can minimize the energy landscape \cite{BarrenPlateauSolution}.
In the original proposal of ADAPT-VQE, the ansatz-element pool was physically motivated, comprising single and double fermionic excitations.
Since then, different types of ansatz-element pools have been proposed to minimize the number of CNOT gates in the ansatz circuit and thus improve the noise resilience of ADAPT-VQE \cite{QEB-ADAPT-VQE, eQEB-ADAPT-VQE, qubit-ADAPT-VQE, EfficientCNOT}.

ADAPT-VQEs still face issues. 
Compared with other VQE algorithms, ADAPT-VQEs make more calls to 
quantum processors.
This is because in every iteration, finding the ansatz element with the steepest energy gradient requires at least $\bigO{|\pool|}$ quantum 
processor calls.
This makes more efficient pool-exploration strategies desirable.
Moreover, noise poses serious restrictions on the maximum depth of useful VQE ansatz circuits \cite{KieransNoise}.
This makes shallower ansatz circuits desirable.
A recent algorithm called TETRIS-ADAPT-VQE compresses VQE ansatz circuits into compact layers of ansatz elements \cite{Tetris-ADAPT-VQE}.
This yields shallower ansatz circuits.
However, it has not yet been demonstrated that shallower ansatz circuits lead to improved noise resilience.
It is, therefore, important to evaluate whether such shallow ansatz circuits boost the noise resilience of ADAPT-VQEs.

In this paper, we broaden the understanding of TETRIS-like layering algorithms.
First, we show how non-commuting ansatz elements can be used to define a topology on the ansatz-element pool.
Based on this topology, we present Subpool Exploration: a pool-exploration strategy to reduce the number of quantum-processor calls when searching for ansatz elements with large energy gradients.
We then investigate several flavors of algorithms to layer and shorten ansatz circuits.
Benchmarking these algorithms, we find that alternative layering strategies can yield equally shallow ansatz circuits as TETRIS-ADAPT-VQE.
Finally, we investigate whether shallow VQE circuits are more noise resilient.
We do this by benchmarking both standard and layered ADAPT-VQEs in the presence of noise.
For amplitude damping and dephasing noise, which globally affect idling and non-idling qubits alike, we observe an increased noise resilience due to shallower ansatz circuits.
On the other hand, we find that layering is unable to mitigate depolarizing noise, which acts locally on actively manipulated qubits.

The remainder of this paper is structured as follows:\
In Sec.~\ref{sec: background}, we introduce notation and the ADAPT-VQE algorithm.
In Sec.~\ref{sec: algorithm} and Sec.~\ref{sec: noiseless}, subpool exploration and layering for ADAPT-VQE are described and benchmarked, respectively. 
We study the runtime advantage of layering in Sec.~\ref{sec: complexity}.
In Sec.~\ref{sec: noise}, we investigate the effect of noise on layered VQE algorithms.
Finally, we conclude in Sec.~\ref{sec: conclusion}.

\section{Preliminaries: Notation and the ADAPT-VQE}
\label{sec: background}

In what follows, we consider second-quantized Hamiltonians on a finite set of $N$ spin orbitals:
\begin{align}
    H=\sum_{p,q=1}^{N}h_{pq}a_p^\dagger a_q
    \; + \sum_{p,q,r,s=1}^N h_{pqrs}a_p^\dagger a_q^\dagger a_r a_s.
\end{align}
$a_p^\dagger$ and $a_p$ denote fermionic creation and annihilation operators of the $p$th spin-orbital, respectively.
The coefficients $h_{pq}$ and $h_{pqrs}$ can be efficiently computed classically---we use the Psi4 package \cite{https://doi.org/10.1002/wcms.93}.

The \textit{Jordan-Wigner transformation} \cite{QCCReview} is used to represent creation and annihilation operators by
\begin{align}
        a_p^\dagger\mapsto Q_p^\dagger \mathcal Z_p \quad\quad
        a_p        \mapsto Q_p         \mathcal Z_p,
\end{align}
respectively. Here,
\begin{align}
 Q_p^\dagger \coloneqq \frac{1}{2}\left(X_p-iY_p\right)\quad\quad
 Q_p \coloneqq \frac{1}{2}\left(X_p+iY_p\right),
\end{align}
are the qubit creation and annihilation operators and $X_p, Y_p, Z_p$ denote Pauli operators acting on qubit $p$. 
The fermionic phase is represented by
\begin{align}
 \mathcal Z_p \coloneqq \displaystyle\bigotimes_{q<p}Z_q.
\end{align}

Anti-Hermitian operators $T$ generate \textit{ansatz elements} that form Stone's-encoded unitaries parametrized by one real parameter $\theta$:
\begin{align}
 A(\theta)[\rho] \coloneqq \exp(\theta T) \rho 
\exp(-\theta T).
\end{align}
Different ADAPT-VQE algorithms choose $T$ from different types of operator pools.
There are three common types of operator pools.
The fermionic pool $\pool^{\rm{Fermi}}$ \cite{ADAPT-VQE} contains fermionic single and double excitations generated by anti-Hermitian operators:
\begin{align}
    \label{eq:fermionic}
 T^p_q   &\coloneqq a_p^{\dagger} a_q - a_q^{\dagger} a_p,\\
 T^{pq}_{rs} &\coloneqq a_p^{\dagger} a_q^{\dagger} a_r a_s 
                   - a_s^{\dagger} a_r^{\dagger} a_q a_p,
\end{align}
where $p,q,r,s=1,...,N$.
The QEB pool $\pool^{\text{QEB}}$ \cite{QEB-ADAPT-VQE} contains single and double qubit excitations generated by anti-Hermitian operators:
\begin{align}
    \label{eq: QEB}
 T^p_q   &\coloneqq Q_p^{\dagger} Q_q - Q_q^{\dagger} Q_p,\\
 T^{pq}_{rs} &\coloneqq Q_p^{\dagger} Q_q^{\dagger} Q_r Q_s 
                - Q_s^{\dagger} Q_r^{\dagger} Q_q Q_p.\label{eq: QEB double}
\end{align}
The qubit pool $\pool^{\text{qubit}}$ \cite{qubit-ADAPT-VQE} contains parameterized unitaries generated by strings of Pauli-operators $\sigma_p\in\left\{X_p, Y_p, Z_p\right\}$:
\begin{align}
    \label{eq: qubit}
 T_{pq}   &\coloneqq i \sigma_p \sigma_q,\\
 T_{pqrs} &\coloneqq i \sigma_p \sigma_q \sigma_r \sigma_s.\label{eq: qubit double}
\end{align}
Further definitions and discussions of all three pools are given in \cref{appendix: pool definitions}. It is worth noting that all ansatz elements have quantum-circuit representations composed of multiple standard single- and two-qubit gates \cite{EfficientCNOT}.
All three pools contain $\bigO{N^4}$ elements.

ADAPT-VQEs optimize several objective functions.
At iteration step $t$, the energy landscape is defined by
\begin{align}
 E_t(\theta_t,...,\theta_1) \equiv \Trace\left[H 
A_t(\theta_t)\circ\ldots\circ A_1(\theta_1)[\rho_{0}]\right].
\end{align}
A global optimizer may repeatedly evaluate $E_t$ and its partial derivatives at the end of the $t$th iteration to return a set of optimal parameters:
\begin{align}
 \label{eq: GlobalMinimization}
 (\theta_t^{*},\ldots,\theta_{1}^{*}) = 
\argmin_{(\theta_t, \ldots, 
\theta_{1})\in\mathbb{R}^t}{E_t(\theta_t,...,\theta_1)}.
\end{align}
These parameters set the upper energy bound of the $t$th iteration:
\begin{align}
 \bound_t = E_t(\theta_t^{*},\ldots,\theta_{1}^{*}).
\end{align}

A loss function $L_t\colon\pool \to \mathbb R$ is used to pick an ansatz element from the operator pool $\pool$ at each iteration $t$:
\begin{align}
 \label{eq: LocalMinimization}
 A_t = \argmin_{A\in \pool}L_{t}(A).
\end{align}
Throughout this paper, we use the standard gradient loss of ADAPT-VQEs, defined in \cref{eq: loss}.
We denote the state after $t-1$ iterations with optimized parameters 
$\theta_{t-1}^{*},\ldots,\theta_1^{*}$ by
\begin{align}
 \rho_{t-1} = \Lambda_{t-1}(\theta_{t-1}^{*},\ldots,\theta_1^{*}) [\rho_0].
\end{align}
Further, we define the energy expectation after adding the ansatz element $A\in\pool$ as
\begin{align}
 E_{t,A}(\theta) = \Trace\left(H A(\theta)[\rho_{t-1}]\right).
\end{align}
Then, the loss is defined by
\begin{align}
    \label{eq: loss}
 L_t(A) 
 = -\left|\frac{\partial E_{t,A}(\theta)}{\partial \theta}\right|_{\theta=0}
= -\left|\Trace\left([H, T] \rho_{t-1}\right)\right|.
\end{align}
We consider alternative loss functions in \cref{appendix: Largest Energy Reduction Decision}.

The ADAPT-VQE starts by initializing a state $\rho_0$. Often, $\rho_0$ is the Hartree-Fock state $\rho_{\HF}$.
The algorithm then builds the ansatz circuit $\Lambda_t$ by first adding ansatz elements $A_t\in\pool$ of minimal loss $L_t$, according to \cref{eq: LocalMinimization}. Then, the algorithm optimizes the ansatz circuit parameters according to \cref{eq: GlobalMinimization}.
This generates a series of upper bounds,
\begin{align}
 \bound_0 > \bound_1 > ... >  \bound_t ,
\end{align}
until the improvement of consecutive bounds drops below a threshold $\varepsilon$ such that $\bound_{t-1} - \bound_t < \varepsilon$, or the maximum iteration number $\tmax$ is reached.
The final bound ($\bound_{t}$ or $\bound_{\tmax}$) is then returned to approximate $E_0$.
A pseudo-code of the ADAPT-VQE is given in \cref{alg: ADAPTVQE}.

\begin{algorithm}[H]
   \caption{ADAPT-VQE \cite{ADAPT-VQE}}
    \label{alg: ADAPTVQE}
    \begin{algorithmic}[1]
        \State Initialize state $\rho_0\gets\rho_{\HF}$, circuit 
$\Lambda_0\gets1$, and     pool $\mathcal P$.
        \State Initialize energy bound $\bound_0\gets\infty$ and accuracy $\varepsilon$.
        \For{$t=1,...,\tmax$}
            \State Select ansatz element: $A_t \gets \argmin_{A\in\pool}{L_t(A)}$
            \State Set circuit: 
                   $\Lambda_t(\theta_t,...,\theta_1) \gets 
                   A_t(\theta_t)\circ \Lambda_{t-1}(\theta_{t-1},...,\theta_1)$
            \State Optimize circuit:
                   $(\theta_t^{*},...,\theta_1^{*}) \gets 
                   \argmin{E_t(\theta_t,...,\theta_1)}$
            \State Update energy bound:
                $\bound_t \gets E_t(\theta_t^{*},...,\theta_1^{*})$ 
            \State Update state:
                $\rho_t \gets \Lambda_t(\theta_t^{*},...,\theta_1^{*})[\rho_0]$ 
            \If{$\bound_{t-1}-\bound_t<\varepsilon$}
                \State \textbf{return} energy bound $\bound_{t}$
            \EndIf
        \EndFor
        \State \textbf{return} energy bound $\bound_{\tmax}$
    \end{algorithmic}
\end{algorithm}

\section{Subpool exploration and layering for ADAPT-VQEs}
\label{sec: algorithm}

In this section, we present two subroutines to improve ADAPT-VQEs.
The first subroutine optimally layers ansatz elements, as depicted in \cref{fig: circuit structures}.
We call the process of producing dense (right-hand side) ansatz circuits instead of sparse (left-hand side) ansatz circuits ``layering''. 
This subroutine can be used to construct shallower ansatz circuits, which may make ADAPT-VQEs more resilient to noise.
The second subroutine is subpool exploration. 
It searches ansatz-element pools in successions of non-commuting ansatz elements.
Subpool exploration is essential for layering and can reduce the number of calls an ADAPT-VQE makes to a quantum processor.
Combining both subroutines results in algorithms similar to TETRIS-ADAPT-VQE \cite{Tetris-ADAPT-VQE}. Our work focuses on developing and understanding layering algorithms from the perspective of non-commuting sets of ansatz elements.

\begin{figure}[!h]
    \includegraphics[page=1]{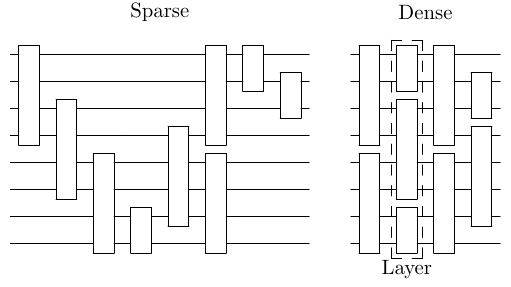}
\caption{Layering: A sparse ansatz circuit (left), as produced by standard ADAPT-VQEs, can be compressed to a dense structure (right) by layering. Boxes denote ansatz elements. Each line represents a single qubit. Note that ansatz circuit elements entangle two or four qubits.}
\label{fig: circuit structures}
\end{figure}

\subsection{Commutativity and Support}

Commutativity of ansatz elements is a central notion underlying our subroutines:
\begin{definition}[Operator Commutativity]
Two ansatz elements $A, B\in\pool$ are said to ``operator commute'' iff $A\left(\theta\right)$ and $B\left(\phi\right)$ commute for all $\theta$ and $\phi$:
\begin{align}
 \commuteO{A}{B}=0 \iff \forall\theta,\phi\in\mathbb{R} \; : \; 
\left[A\left(\theta\right),B\left(\phi\right)\right]=0.
\end{align}
Conversely, two ansatz elements $A, B\in\pool$ do not operator-commute iff there exist parameters for which the corresponding operators do not commute:
\begin{align}
\commuteO{A}{B}\neq0 \iff \exists\theta,\phi\in\mathbb{R} : 
\left[A\left(\theta\right),B\left(\phi\right)\right]\neq 0.
\end{align}
\end{definition}
\begin{definition}[Operator non-commuting set]
Given an ansatz-element pool $\pool$ and an ansatz element $A\in\pool$, we define its operator non-commuting set as follows
\begin{align}
\nset_{\textrm{O}}(\pool,A)\coloneqq\left\{ B\in\pool : \commuteO{A}{B}\neq0\right\}.
\end{align}
\end{definition}
Operator commutativity is central to layering.
The operator non-commuting set is central to subpool exploration.
Structurally similar and more-intuitive notions can be defined using qubit support:
\begin{definition}[Qubit support]
Let $\mathcal B\left(\mathcal H\right)$ denote the set of superoperators on a Hilbert space $\mathcal H$. Let $\mathcal H_{\mathcal Q}\coloneqq\bigotimes_{q\in\mathcal Q}\mathcal H_q$ denote the Hilbert space of a set of all qubits $\mathcal Q\equiv\left\{1,\ldots,N\right\}$ where $\mathcal H_q$ is the Hilbert space corresponding to the $q$th qubit. Consider a superoperator $A\in\mathcal B\left(\mathcal H_{\mathcal Q}\right)$. First, we define the superoperator subset that acts on a qubit subset $\mathcal W\subseteq\mathcal Q$ as
\begin{equation}
    \mathcal B_{\mathcal W}\coloneqq\left\{B\otimes\text{\normalfont id}\colon B\in\mathcal B\left(\mathcal H_{\mathcal W}\right)\right\}\subseteq\mathcal B\left(\mathcal H_{\mathcal Q}\right).
\end{equation}
Then, we define the qubit support of a superoperator $A$ as its minimal qubit subset $\mathcal W$:
\begin{equation}
    \support\left(A\right)\coloneqq\left.\smashoperator[r]{\bigcap_{\substack{{\mathcal W\subseteq\mathcal Q\colon A\in\mathcal B_{\mathcal W}}}}}\mathcal W\right. .
\end{equation}
The notion of support extends to parameterized ansatz elements:
\begin{equation}
    \support\left(B\right)\coloneqq \bigcup_\theta\support\left(B\left(\theta\right)\right),
\end{equation}
where $B$ is a parameterized ansatz element.
\end{definition}
Intuitively, the qubit support of an ansatz element $A$ is the set of all qubits the operator $A$ acts on nontrivially \cref{fig: commutation diagram}.
The concept of qubit support allows one to define support commutativity of ansatz elements as follows.
\begin{definition}[Support commutativity]
Two ansatz elements $A, B\in\pool$ are said to ``support-commute'' iff their qubit support is disjoint.
\begin{align}
 \commuteS{A}{B}=0 \iff \support(A)\cap\support(B)=\emptyset.
\end{align}
Conversely, two ansatz elements $A, B\in\pool$ do not support-commute iff their supports overlap
\begin{align}
 \commuteS{A}{B}\neq0 \iff \support(A)\cap\support(B)\neq\emptyset.
\end{align}
\end{definition}
\begin{figure}
    \resizebox*{\columnwidth}{!}{
        \includegraphics[page=1]{
    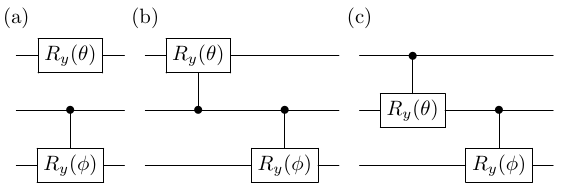}}
    \caption{A diagram comparing support and operator commutativity. (a) The elements both support and operator commute. Note that $R_y(\theta)$ only has qubit support on the top qubit line. (b) The elements operator commute but do not support commute. (c) The elements neither support nor operator commute.}
    \label{fig: commutation diagram}
\end{figure}
\begin{definition}[Support non-commuting set]
Given an ansatz-element pool $\pool$ and an ansatz element $A\in\pool$, we define the set of ansatz elements with overlapping support as
\begin{align}
    \nset_{\textrm{S}}(\pool,A)\coloneqq\left\{ B\in\pool : \{B,A\}_S\neq0\right\}.
\end{align}
\end{definition}

Operator commutativity and support commutativity are not equivalent---see \cref{fig: commutation diagram}.
However, the following properties hold.
Elements supported on disjoint qubit sets operator commute:
\begin{align}
 \forall A,B\in\pool \quad \commuteS{A}{B}=0 \implies 
\commuteO{A}{B}=0.\label{eq:commutation implies operator commutation}
\end{align}
Conversely, operator non-commuting ansatz elements act on, at least,  one common qubit which implies they are support non-commuting:
\begin{align}
 \forall A,B\in\pool \quad \commuteO{A}{B}\neq0 \implies 
\commuteS{A}{B}\neq0.
\end{align}
The last relation also implies that the operator non-commuting set of $A$ is contained in its support non-commuting set
\begin{align}
 \nset_{\textrm{O}}(\pool,A)\subseteq\nset_{\textrm{S}}(\pool,A).
\end{align}
We further generalize the notions of operator and support commutativity in \cref{appendix: generalized commutativity}. Henceforth, we will use generalized commutativity to denote either operator or support commutativity or any other type of commutativity specified in \cref{appendix: generalized commutativity}. Further, $\nset_{\textrm{G}}$ and $\commute{\bullet}{\bullet}$ will be used to denote the generalized non-commuting set and the generalized commutator, respectively.

For later reference, we note that generalized non-commuting sets induce a topology on $\pool$ via the following discrete metric.
\begin{definition}[Pool metric]
\label{def:pool-metric}
Let $d\colon\pool\times\pool\to\left\{0,1,2\right\}$ define a discrete metric such that:
(i) $\forall A\in\pool$, set $d(A,A)=0$.
(ii) $\forall A,B\in\pool$ with $A\neq B$ and $\commute{A}{B}\neq0$, set 
$d(A,B)=1$.
(iii) $\forall A,B\in\pool$ with $A\neq B$ and $\commute{A}{B}=0$, set 
$d(A,B)=2$.
\end{definition}
\begin{figure}
    \centering
    \includegraphics[page=1]{
    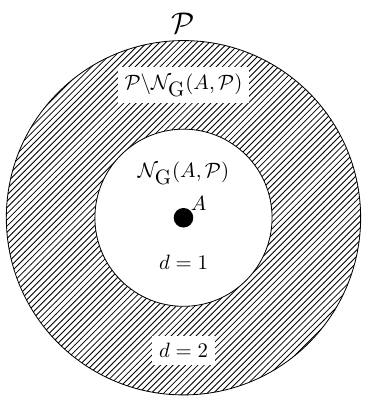}
    \caption{A diagram of the distance $d$ from an element $A\in\mathcal P$ under the pool metric, Definition~\ref{def:pool-metric}. The ansatz element $A$ (black dot) is surrounded by ansatz elements of the non-commuting set $\nset_{\textrm{G}}(\pool, A)$ of distance 1 (white circle). All other elements $\pool\setminus\nset_{\textrm{G}}(\pool, A)$ have distance 2 (gray shaded region).}
    \label{fig: metric diagram}
\end{figure}
With this metric, the generalized non-commuting elements $\nset_{\textrm{G}}(\pool, A)$ form a ball of distance one around each ansatz element $A\in\pool$. The metric is represented diagrammatically in \cref{fig: metric diagram}.
This allows us to identify an element $A\in\pool$ as a local minimum if there is no element with lower loss within $A$'s generalized non-commuting set.
\begin{property}[Local minimum]\label{property: local minimum}
Let $\pool$ be an ansatz-element pool with the pool metric of
\cref{def:pool-metric}, and let $L:\pool\to\mathbb{R}$ denote a loss function.
Then, any element $A\in\pool$ for which
\begin{equation}
 L\left(A\right)=\min_{B\in\nset_{\textrm{G}}\left(\pool,A\right)}L\left(B\right),
\end{equation}
is a local minimum on $\pool$ with respect to $L$.
\end{property}

This property is important as we will later show that subpool exploration always returns local minima.

To gain intuition about the previously defined notions, we consider the ansatz elements of the QEB and the Pauli pools, \cref{eq: QEB,eq: QEB double,eq: qubit,eq: qubit double}.
The ansatz elements of these pools have qubit support on either two or four qubits, as is illustrated in \cref{fig: circuit structures}.
Commuting ansatz elements with disjoint support can be packed into an ansatz-element layer, which can be executed on the quantum processor in parallel.
This is the core idea of layering, which helps to reduce the depths of ansatz circuits.
Moreover, as generalized non-commuting ansatz elements must share at least one qubit, we conclude that the generalized non-commuting set $\nset_{\textrm{G}}(\pool, A)$ has at most $\bigO{N^3}$ ansatz elements. 
This is a core component of subpool exploration.
Analytic expressions for the cardinalities of the generalized non-commuting sets are given in \cref{appendix:non-commuting set cardinalities}. In \cref{appendix: non-commuting sets}, we prove that two different fermionic excitations operator commute iff they act on disjoint or equivalent sets of orbitals. The same is true for qubit excitations.  Pauli excitations operator commute iff the generating Pauli strings differ in an even number of places within their mutual support.

\subsection{Subpool exploration}
\label{sec:subpool exploration}

In this section, we introduce \textit{subpool exploration}, a strategy to explore ansatz-element pools with fewer quantum-processor calls.
Subpool exploration differs from the standard ADAPT-VQE as follows.
Standard ADAPT-VQEs evaluate the loss of every ansatz element in the ansatz-element pool $\pool$ in every iteration of ADAPT-VQE, (\cref{alg: ADAPTVQE}, Line 4).
This leads to $\bigO{|\pool|}$ quantum-processor calls to identify the ansatz element of minimal loss.
Instead, subpool exploration evaluates the loss of a reduced number of ansatz elements by exploring a sequence of generalized non-commuting ansatz-element subpools.
This can lead to a reduced number of quantum-processor calls and returns an ansatz element which is a local minimum of the pool $\pool$.
The details of subpool exploration are as follows.

\textit{Algorithm:---}Let $\pool$ denote a given pool and $L$ a given loss function.
Instead of na\"ively computing the loss of every ansatz element in $\pool$, our algorithm explores $\pool$ iteratively by considering subpools, $\subpool_{m}\subsetneq\pool$, in consecutive steps.
During this process, the algorithm successively determines the ansatz element with minimal loss within subpool $\subpool_m$ as
\begin{align}
 A_m = \argmin_{A\in\subpool_m}{L(A)}.
\end{align}
Meanwhile, the corresponding loss value is stored:
\begin{align}
 \Lvalue_m = L(A_m).
\end{align}
Iterations are halted when loss values stop decreasing.
The key point of subpool exploration is to update the subpools $\subpool_m$ using the generalized non-commuting set generated by $A_m$:
\begin{equation}\label{eq: subpool update}
    \mathcal S_{m+1}=\nset_{\textrm{G}}\left(\pool\backslash\mathcal S_{\le m},A_m\right)\subseteq\nset_{\textrm{G}}\left(\pool,A_m\right),\!\quad\forall m\ge0,
\end{equation}
where $\mathcal S_{\le m}\coloneqq\cup_{l=0}^m\mathcal S_l$.
A pseudo-code summary of subpool exploration is given in \cref{alg: SubpoolExploration}, and a visual summary is displayed in \cref{fig: PoolExploration}.
We now discuss aspects of subpool exploration.
\begin{algorithm}[H]
   \caption{Subpool Exploration}
    \label{alg: SubpoolExploration}
   
\begin{algorithmic}[1]
    \State \textbf{Input:} Pool $\pool$ and loss function $L$.
    \State Initialize subpool $\subpool_0$ and loss value $\Lvalue_{0}\gets\infty$.
    \For{$m=0,...$}
        \State Select ansatz element  $A_m \gets 
\argmin_{A\in\subpool_{m}}{L(A)}$.
        \State Update loss value $\Lvalue_{m}\gets L(A_m)$.
        \If{$\Lvalue_{m} < \Lvalue_{m-1}$}
            \State Update subpool $\mathcal S_{m+1}=\nset_{\textrm{G}}\left(\mathcal P\backslash\mathcal S_{\le m},A_m\right)$.
        \Else
            \State \Return $A_m, \subpool_m$.
        \EndIf
    \EndFor
\end{algorithmic}
\end{algorithm}

\begin{figure}
\includegraphics[page=1]{
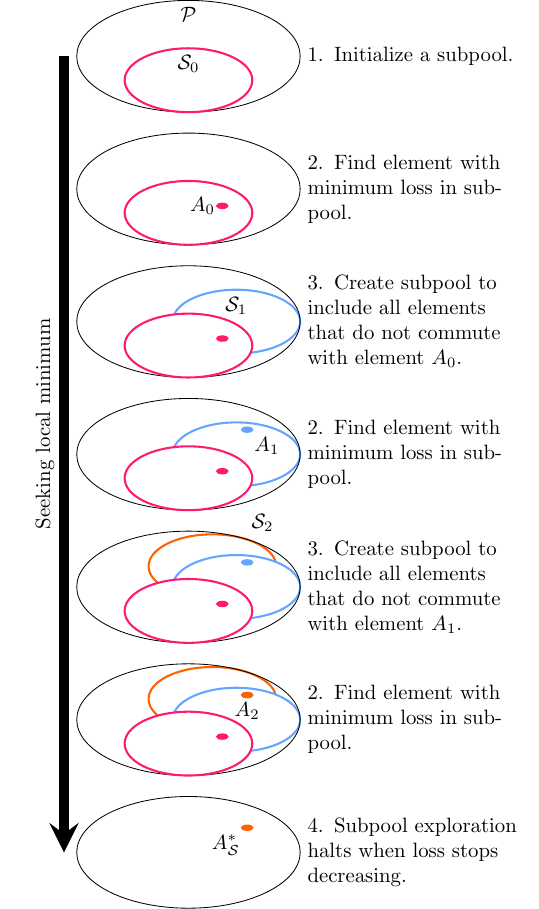}
\caption{Visualization of our strategy for subpool exploration. The pool $\pool$ is successively explored in subpools $\subpool_m$, with ansatz elements $A_m$ of minimal loss generating future subpools through their generalized non-commuting set.
}
\label{fig: PoolExploration}
\end{figure}

\textit{Efficiency:---}
Let $\mmax$ denote the index of the final iteration and define the set of searched ansatz elements as
\begin{align}
 \subpool \coloneqq \cup_{m=0}^{\mmax}\subpool_m.
\end{align}
As loss values of ansatz elements that have been explored are stored in a list, it follows that subpool exploration requires only $|\subpool|$ loss function calls.
On the other hand, exploring the whole pool in ADAPT-VQE requires $|\pool|$ loss-function calls.
Since $\subpool$ is a subset of $\pool$, subpool exploration may reduce the number of quantum-processor calls:
\begin{align}
 \subpool\subseteq\pool \implies |\subpool|\le |\pool|.
\end{align}
To give a specific example, consider the QEB and qubit pools. 
Those pools contain $\bigO{N^4}$ ansatz elements.
On the other hand, generalized non-commuting sets have $\bigO{N^3}$ ansatz elements.
Thus, by choosing an appropriate initial subpool, we can ensure that $|\subpool_m|=\bigO{N^3}$ for all subpools. 
Especially if the number of searched subpools is $\mmax=\bigO{1}$, subpool exploration can return ansatz elements of low loss while exploring only $\bigO{N^3}$ ansatz elements.

We note that this pool-exploration strategy ignores certain ansatz elements.
In particular, it may miss the optimal ansatz element with minimal loss.
Nevertheless, as explained in the following paragraphs, it will always return ansatz elements which are \textit{locally optimal}.
This ensures that the globally optimal ansatz element can always be added to the ansatz circuit later in the algorithm.

\textit{Optimality:---}As the set of explored ansatz elements $\subpool$ is a subset of $\pool$, the ansatz element returned by subpool exploration
\begin{align}
 A_{\subpool}^{*}\coloneqq\argmin_{A\in\subpool} L(A) ,
\end{align}
may be sub-optimal to the ansatz element returned by exploring the whole pool
\begin{align}
 A_{\pool}^{*}\coloneqq\argmin_{A\in\pool} L(A).
\end{align}
That is,
\begin{align}
 \subpool \subset \pool 
 \implies 
 \argmin_{A\in\pool} L(A) \le \argmin_{A\in\subpool} L(A).
\end{align}
Yet, there are a couple of useful properties that pertain to the output $A_{\subpool}^{*}$ of subpool exploration.
At first, the outputs of subpool exploration are local minima.
\begin{property}[Local Optimality]
Any ansatz element $A_{\subpool}^{*}$ returned by subpool exploration is a local minimum. 
\end{property}
The proof of this property is immediate.
Moreover, as subpool exploration constructs subpools from generalized non-commuting sets, the only ansatz elements $B\in\pool$ with $L(B)<L(A_{\subpool}^{*})$ must necessarily generalized commute with $A_{\subpool}^{*}\in\pool$.
\begin{property}
\label{prop: non-blocking}
\textnormal{\textbf{(Better ansatz elements generalized commute)}}
Let $\pool$ denote a pool and $L$ denote a loss function.
Let $A_{\subpool}^{*}\in\pool$ denote the final output of subpool exploration. Then,
\begin{align}
 \forall B\in\pool \text{ with } L(B)<L(A_{\subpool}^{*}) &\implies 
\commute{A_{\subpool}^{*}}{B}=0\\
&\implies 
\commuteO{A_{\subpool}^{*}}{B}=0.\label{eq:second non-blocking implication}
\end{align}
\end{property}
\begin{proof}
We prove this property by contradiction.
Assume that there is an ansatz element $L(B)<L(A_{\subpool}^{*})$ such that $\commute{A_{\subpool}^{*}}{B}\neq0$.
This implies that $B$ is in the generalized non-commuting set $\nset_{\textrm{G}}(\pool, A_{\subpool}^{*})$ and exploring the corresponding subpool would have produced $L(B)<L(A_{\subpool}^{*})$ leading to the exploration of $\nset_{\textrm{G}}(\pool, B)$.
This, in turn, can only return an ansatz element with a loss $L(B)$ or smaller. This would contradict $A_{\subpool}^{*}$ having been the final output of the algorithm. Finally, we use \cref{eq:commutation implies operator commutation} to show \cref{eq:second non-blocking implication}
\end{proof}
Property~\ref{prop: non-blocking} is useful as it ensures that subpool exploration can find better ansatz elements, which first were missed, in subsequent iterations.
To see this, suppose a first run of subpool exploration returns a local minimum $A\in\pool$.
Further, suppose there is another local minimum $B\in\pool$ such that $L(B)<L(A)$.
Property~\ref{prop: non-blocking} ensures that $A$ and $B$ generalized commute.
Hence, by running subpool exploration repeatedly on the remaining pool, we are certain to discover the better local minimum eventually.
Ultimately, this will allow for restoring the global minimum.

\textit{Initial subpool:---}So far, we have not specified any strategy for choosing the initial set $\subpool_{0}$.
This can, for example, be done by taking the subpool of a single random ansatz element $A_0\in\pool$.
Alternatively, one can compose $\subpool_0$ of random ansatz elements enforcing an appropriate pool size, e.g., $|\subpool_0|=\bigO{N^3}$ for QEB and qubit pools.

We will refer to the ADAPT-VQE with subpool exploration as the Explore-ADAPT-VQE. This algorithm is realized by replacing Line 4 in \cref{alg: ADAPTVQE} with subpool exploration, \cref{alg: SubpoolExploration}, with $L \rightarrow L_t$.

\subsection{Layering}

Below, we describe two methods for arranging generalized non-commuting ansatz elements into ansatz-element layers.
\begin{definition}[Ansatz-element layer]
Let $\layer$ be a subset of $\pool$.
We say that $\layer$ is an ansatz-element layer iff
\begin{align}
    \commute{A}{B}=0\quad\forall A, B \in \layer \text{ such that } A\neq B.
\end{align}
We denote the operator corresponding to the action of the mutually generalized-commuting ansatz elements of an ansatz-element layer $\layer$ with
\begin{align}
 \layer^{\text{\normalfont o}}(\vec{\theta}) = \prod_{A\in\layer}A(\theta_A) .
\end{align}
Here, $\vec{\theta}$ is the parameter vector for the layer:
\begin{align}
 \vec{\theta} \equiv \left\{\theta_A\colon A\in\layer\right\}.
\end{align}
We note that for support commutativity, the product can be replaced by the tensor product.
\end{definition}

Since ansatz-element layers depend on parameter vectors, the update rule is
\begin{align}
 \Lambda_t(\vec{\vartheta}_t) = 
    \layer^{\text{o}}(\vec{\theta}_t)\circ\Lambda_{t-1}(\vec{\vartheta}_{t-1}), 
    \quad \vec{\vartheta}_t = (\vec{\theta}_t,\vec{\vartheta}_{t-1}),
\end{align}
As before, the algorithm is initialized with  $\Lambda_0=\text{\normalfont id}$ and $\vec{\vartheta} = ( )$.
To make the dependence on the ansatz circuit explicit, we denote the energy landscape as
\begin{align}
    E_{\Lambda}(\vec{\vartheta}) \coloneqq
        \trace\left[ H \Lambda(\vec\vartheta) \left[\rho_0\right] \right] .
\end{align}
The energy landscape of the $t$th iteration is denoted as
\begin{align}
    E_{t}(\vec{\vartheta}_t) \equiv E_{\Lambda_t}(\vec{\vartheta}_t),
\end{align}
and its optimal parameters are
\begin{align}
    \label{eq:OptimalParameterVector}
    \vec{\vartheta}_t^{*} = \argmin_{\vec{\vartheta}} E_{t}(\vec{\vartheta}).
\end{align}
Further, the gradient loss (c.f. \cref{eq: loss}) is
\begin{align}\label{eq: intermediate loss}
 L_{\Lambda(\vec{\vartheta})}(A)\equiv
            -\left|\trace\left[
                \left[H,T_A\right]\Lambda(\vec{\vartheta})\left[\rho_0\right]
            \right]\right|,
\end{align}
where the definitions in \cref{eq: loss,eq: intermediate loss} satisfy the following relation
\begin{align}
    \label{eq: LayerLoss}
    L_{t}(A) = L_{\Lambda_{t-1}(\vec{\vartheta}_{t-1}^{*})}(A).
\end{align}
With this notation in place, we proceed to describe two methods to construct ansatz-element layers.

\subsubsection{Static layering}

Our algorithm starts by initializing an empty ansatz-element layer and the remaining pool $\pool'$ to be the entire pool $\pool$.
Further,  the loss is set such that $L\gets L_{\Lambda_{t-1}}$ for the $t$th layer.
The algorithm proceeds to fill the ansatz-element layer by successively running subpool exploration to pick an ansatz element $A_n$ in $n=0,\ldots, \nmax$ iterations. This naturally induces an ordering on the layer.
At every step of the iteration, the corresponding generalized non-commuting set $\nset_{\textrm{G}}(\pool', A_n)$ is removed from the remaining pool $\pool'$.
If the loss of the selected ansatz element $A_n$ is smaller than a predefined threshold $L(A)<\Lmax$,  it is added to the ansatz-element layer $\layer$.
The layer is completed once the pool is exhausted ($\pool'=\emptyset$) or the maximal iteration count $\nmax$ is reached.
A pseudocode summary of static layering is given in \cref{alg: BuildStaticLayer}.
\begin{algorithm}[H]
   \caption{Build Static Layer}
    \label{alg: BuildStaticLayer}
\begin{algorithmic}[1]
    \State \textbf{Input:} Pool $\pool$, loss $L$, max. loss $\Lmax$, max. 
iteration $\nmax$
    \State Initialize remaining pool $\pool'\gets\pool$
    \State Initialize ansatz layer $\layer\gets\emptyset$
    \For{$n=0,...,\nmax$}
        \State Set $A \gets \mathrm{SubpoolExploration}(\pool',L)$
        \If{$L(A) < \Lmax$}
            \State Update layer $\layer \gets \layer\cup\{A\}$.
        \EndIf
        \State Reduce pool $\pool'\gets\pool'\setminus\nset_{\textrm{G}}(\pool', A)$
        \If{$\pool'=\emptyset$}
            \textbf{break}
        \EndIf
    \EndFor
    \Return $\layer$
\end{algorithmic}
\end{algorithm}
In Static-ADAPT-VQE, static layering is used to grow an ansatz circuit iteratively.
In each iteration, the layer is appended to the ansatz circuit, and the ansatz-circuit parameters are re-optimized.
Iterations halt once the decrease in energy falls below $\accuracy$, the energy accuracy per ansatz element.
A summary of Static-ADAPT-VQE is given in \cref{alg: StaticLayeredADAPTVQE}.
\begin{algorithm}[H]
    \caption{Static-ADAPT-VQE}
    \label{alg: StaticLayeredADAPTVQE}
    \begin{algorithmic}[1]
        \State Initialize state $\rho_0\gets\rho_{\HF}$, ansatz circuit 
$\Lambda_0\gets1$, pool $\mathcal P$.
        \State Initialize energy bound $\bound_0\gets\infty$ and accuracy $\varepsilon$.
        \State Initialize maximal loss $\Lmax$ and iteration count $\nmax$.
        \For{$t=1,...,\tmax$}
            \State Get layer: $\layer_t \gets \mathrm{BuildStaticLayer}(\pool, 
                                                 L_{\Lambda_{t-1}}, \Lmax, \nmax)$
            \State Set ansatz circuit: 
                   $\Lambda_t(\vec\vartheta_t) \gets 
                   \layer^{\text{o}}_t(\vec\theta_t)\circ \Lambda_{t-1}(\vec\vartheta_{t-1})$
            \State Optimize ansatz circuit:
                   $\vec\vartheta_t^{*} \gets \argmin{E_t(\vec\vartheta_t)}$
            \State Set ansatz circuit:
                   $\Lambda_t \gets \Lambda_t(\vec\vartheta_t*)$
            \State Update energy bound:
                   $\bound_t \gets E_t(\theta_t^{*},...,\theta_1^{*})$ 
            \If{$\bound_{t-1}-\bound_t<\varepsilon\left|\layer_t\right|$}
                \State \textbf{return} energy bound $\bound_{t}$
            \EndIf
        \EndFor
        \State \textbf{return} energy bound $\bound_{\tmax}$
    \end{algorithmic}
    
\end{algorithm}
\begin{figure}[]
    \includegraphics[page=1, width=\columnwidth]{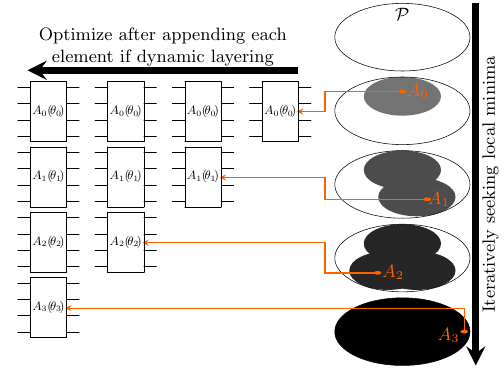}
\caption{Visualization of layer construction and successive pool reduction. Gray areas indicate the removal of generalized non-commuting sets corresponding to ansatz elements $A_n$ added to the layer $\layer$. Parameters can either be optimized once the whole layer is fixed (static layering) or after adding each ansatz element (dynamic layering).}
\label{fig: Layering}
\end{figure}
We establish the close relationship between static layering and TETRIS-ADAPT-VQE in the following property.
\begin{property}
\label{prop: StaticLayering}
Assume that all ansatz elements $A, B\in\pool$ have distinct loss $L(A)\neq L(B)$.
Using support commutativity and provided that $\Lmax=0$ and $\nmax$ are sufficiently large to ensure that the whole layer is filled, Static-ADAPT-VQE and TETRIS-ADAPT-VQE will produce identical ansatz-element layers.
\end{property}
This property is proven by induction.
Assume that the previous iterations of ADAPT-VQE have yielded a specific ansatz circuit $\Lambda_{t-1}(\vec{\vartheta}^{*})$.
The next layer of ansatz elements $\layer_t$ can be constructed either by Static-ADAPT-VQE or TETRIS-ADAPT-VQE.
For both algorithms, the equivalence of $\Lambda_{t-1}(\vec{\vartheta}^{*})$ implies that the loss function, \cref{eq: LayerLoss}, of any ansatz element is identical throughout the construction of the layer $\layer_t$.
First, by picking $\Lmax=0$, we ensure that both TETRIS-ADAPT-VQE and Static-ADAPT-VQE only accept ansatz elements with a non-zero gradient.
Next, we note that if an ansatz element is placed on a qubit by Static-ADAPT-VQE, then by Property \ref{prop: non-blocking}, there exists no ansatz element that acts actively on this qubit and generates a lower loss.
Moreover, there exists no ansatz element with identical loss that acts nontrivially on qubit, as we assume that all ansatz elements have a distinct loss.
Similarly, TETRIS-ADAPT-VQE places ansatz elements from lowest to highest loss and ensures no two ansatz elements have mutual support.  Thus, if an ansatz element is placed on a qubit by TETRIS-ADAPT-VQE, there exists no ansatz element with a lower loss that acts nontrivially on this qubit.
Again, there also exists no ansatz element with identical loss supported by this qubit, as we assume that all ansatz elements have a distinct loss.
Combining these arguments, both Static- and TETRIS-ADAPT-VQE will fill the ansatz-element layer $\layer_t$ with equivalent ansatz elements.
The ansatz elements may be chosen in a different order.
By induction, the equivalence of $\Lambda_{t-1}$ and $\layer_t$ implies the equivalence of the ansatz circuit $\Lambda_t$.

\begin{table*}[]
    \bgroup
    \def\arraystretch{1.5}
    \begin{tabular*}{\textwidth}{@{\extracolsep{\fill}}l|*{5}{c}} \toprule
        Name&\ce{H4}  &\ce{LiH} &\ce{H6} &\ce{BeH2}&\ce{H2O} \\\colrule
        Orbitals $N$&8&12&12&14&14\\\\
        Structure&

        \chemfig{@{a}H-[,,,,<->]@{b}H-[,,,,<->]@{c}H-[,,,,<->]@{d}H}
        \namebond{a}{b}{$d$}
        \namebond{b}{c}{$d$}
        \namebond{c}{d}{$d$}&
        
        \chemfig{@{a}Li-[,,,,<->]@{b}H}
        \namebond{a}{b}{$d$}&

\chemfig{@{a}H-[,,,,<->]@{b}H-[,,,,<->]@{c}H-[,,,,<->]@{d}H-[,,,,<->]@{e}H-[,,,,
<->]@{f}H}
        \namebond{a}{b}{$d$}
        \namebond{b}{c}{$d$}
        \namebond{c}{d}{$d$}
        \namebond{d}{e}{$d$}
        \namebond{e}{f}{$d$}&

        \chemfig{@{a}H-[,,,,<->]@{b}Be-[,,,,<->]@{c}H}
        \namebond{a}{b}{$d$}
        \namebond{b}{c}{$d$}&
    
        \chemfig{-[::-52.25, 
0.5,,,draw=none]@{a}H-[::104.5,,,,<->]@{b}O-[::-104.5,,,,<->]@{c}H}
        \namebond{a}{b}{$d$}
        \namebond{b}{c}{$d$}
        \arclabel{0.5cm}{a}{b}{c}{$\beta$}

\\&$d=\SI{3}{\text{\AA}}$&$d=\SI{1.546}{\text{\AA}}$&$d=\SI{0.735}{\text{\AA}}$&$d=\SI{1.316}
{\text{\AA}}$&$d=\SI{1.0285}{\text{\AA}},$
        \\&&&&&$\beta=\SI{96.84}{\degree}$
        \\\botrule
    \end{tabular*}
    \egroup
    \caption{Table of molecular conformations and the corresponding number of spin-orbitals $N$ used in numerical simulations.}
\label{table: geometries}
\end{table*}

\subsubsection{Dynamic layering}

In static layering, ansatz-circuit parameters are optimized after appending a whole layer with several ansatz elements.
In dynamic layering, on the other hand, ansatz-circuit parameters are re-optimized every time an ansatz element is appended to a layer.
The motivation for doing so is to simplify the optimization process. The price is having to run the global optimization more times. We now describe how to perform dynamic layering.

The starting point is a given ansatz circuit $\Lambda$, a set of optimal parameters $\vec\vartheta^{*}$ [\cref{eq:OptimalParameterVector}] and their corresponding energy bound $\bound\equiv E_{\Lambda}(\vec\vartheta^{*})$.
The remaining pool $\pool'$ is initiated to be the entire pool $\pool$. 
Starting from an empty layer $\layer'$ and a temporary ansatz circuit $\Lambda'=\Lambda$, a layer is constructed dynamically by iteratively adding ansatz elements $A$ to $\layer'$ and $\Lambda'$ while simultaneously re-optimizing the ansatz-circuit parameters $\vec\vartheta^{*}$.
Based on the loss $L'$ induced by the currently optimal ansatz circuit $\Lambda'(\vec\vartheta^{*})$, subpool exploration is used to select ansatz elements $A$.
Simultaneously, the pool of remaining ansatz elements $\pool'$ is shrunk by the successive removal of the generalized non-commuting sets $\nset_{\textrm{G}}(\pool', A)$.
Finally, ansatz elements are only added to the layer $\layer$ if their loss is below a threshold $\Lmax$ and the updated energy bound $\bound'$ exceeds a gain threshold of $\varepsilon$.
A pseudocode summary is given in \cref{alg:BuildDynamicLayer}

Dynamic-ADAPT-VQE iteratively builds dynamic layers $\layer_t$ and appends those to the ansatz circuit $\Lambda_{t-1}$. 
The procedure is repeated until an empty layer is returned; that is, no ansatz element is found that reduces the energy by more than $\varepsilon$.
Alternatively, the algorithm halts when the (user-specified) maximal iteration count $\tmax$ is reached.
A pseudocode summary is given in \cref{alg: DynamicallyLayeredADAPTVQE}.
\begin{algorithm}[H]
   \caption{Build Dynamic Layer}
    \label{alg:BuildDynamicLayer}
    \begin{algorithmic}[1]
        \State \textbf{Input:} Ansatz $\Lambda$, bound $\bound$, opt. 
    params $\vec{\vartheta}^{*}$ 
        \State \textbf{Get:} Pool $\pool$, accuracy $\accuracy$, max. loss $\Lmax$, 
    $\nmax$
        \State \textbf{Initialize:} pool $\pool'\gets\pool$, layer 
                        $\layer'\gets\emptyset$, ansatz circuit 
$\Lambda'\gets\Lambda$
        \For{$n=0,...,\nmax$}
            \State Update loss $L' \gets L_{\Lambda'(\vec{\vartheta}^{*})}$
            \State Set $A \gets \mathrm{SubpoolExploration}(\pool',L')$
            \If{$L'(A) < \Lmax$}
                \State Minimize $(\theta^{*},\vec\vartheta^{*})\gets 
                            \argmin_{(\theta, \vec\vartheta)}
                                    E_{A\circ\Lambda'}(\theta, \vec\vartheta)$
                \State Set bound $\bound'\gets E_{A\circ\Lambda'}(\theta^{*}, 
                                                                \vec\vartheta^{*})$
                \If{$\bound-\bound'\ge\accuracy$}
                    \State Update layer $\layer \gets \layer\cup\{A\}$
                    \State Update ansatz circuit $\Lambda'\gets A\circ\Lambda'$
                    \State Update opt. params 
                        $\vec{\vartheta}^{*}\gets(\theta^{*},\vec{\vartheta}^{*})$
                \EndIf
            \EndIf
            \State Reduce pool $\pool'\gets\pool'\setminus\nset_{\textrm{G}}(\pool', A)$
            \If{$\pool'=\emptyset$} \textbf{break} \EndIf
        \EndFor
        \Return Layer $\layer'$, optimal params $\vec{\vartheta}^{*}$, bound $\bound$.
    \end{algorithmic}
\end{algorithm}
\begin{algorithm}[H]
   \caption{Dynamic-ADAPT-VQE}
    \label{alg: DynamicallyLayeredADAPTVQE}
    \begin{algorithmic}[1]
        \State Initialize state $\rho_0\gets\rho_{\HF}$, ansatz circuit 
$\Lambda_0\gets1$, pool $\mathcal P$.
        \State Initialize accuracy $\varepsilon$ and maximal loss $\Lmax$.
        \State Initialize iteration counts $\tmax, \nmax$.
        \State Initialize energy bound $\bound_0\gets\infty$.
        \State Initialize optimal params $\vartheta_0^{*}$ as empty vector.
        \For{$t=1,...,\tmax$}
            \State $\layer_t, \vec\vartheta_t^{*}, \bound_t \gets 
                        \mathrm{BuildDynamicLayer}(L_{\Lambda_{t-1}}, 
                                          \vec\vartheta_{t-1}^{*}, \bound_{t-1})$
            \If{$\layer_t=\emptyset$}
                \State \textbf{return} energy bound $\bound_{t}$
            \EndIf
            \State Set ansatz circuit: $\Lambda_t \gets \layer^{\text{o}}_t \circ 
\Lambda_{t-1}$
        \EndFor
        \State \textbf{return} energy bound $\bound_{\tmax}$
    \end{algorithmic}    
\end{algorithm}

\section{Benchmarking Noiseless Performance}
\label{sec: noiseless}

In this section, we benchmark various aspects of subpool exploration and layering in noiseless settings.
To this end, we use numerical state-vector simulations to study a wide variety of molecules summarized in \cref{table: geometries}.
While \ce{BeH2} and \ce{H2O} are among the larger molecules to be benchmarked, \ce{H4} and \ce{H6} are prototypical examples of strongly correlated systems.
Our simulations demonstrate the utility of subpool exploration in reducing quantum-processor calls.
Further, we show that when compared to standard ADAPT-VQE, both Static- and Dynamic-ADAPT-VQE reduce the ansatz circuit depths to similar extents.
All simulations use the QEB pool because it gives a higher resilience to noise than the fermionic pool and performs similarly to the qubit pool \cite{KieransNoise}.
Moreover, unless stated otherwise, we use support commutativity to ensure that Static-ADAPT-VQE produces ansatz circuits equivalent to TETRIS-ADAPT-VQE.

\subsection{Efficiency of subpool exploration}\label{sec: explore results}

\begin{figure}[t]
    \includegraphics[width=\columnwidth, 
page=1]{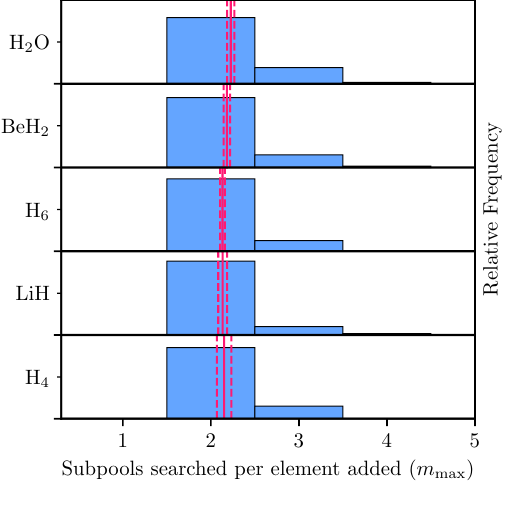}
    \caption{Histograms of the relative frequencies of the number of subpools searched for identifying a suitable ansatz element $\mmax$ with Explore-ADAPT-VQE. The mean and uncertainty in the mean are indicated by solid and dashed lines, respectively.}
\label{fig: PoolExplorationResults}
\end{figure}

We begin by illustrating the ability of subpool exploration to reduce the number of loss function calls when searching for a suitable ansatz element $A$ to append to an ansatz circuit.
To this end, we present Explore-ADAPT-VQE (ADAPT-VQE with subpool exploration) using the QEB pool, \cref{eq: QEB}, and operator commutativity.
We set the initial subpool, $\mathcal S_0$, such that it consists of a single ansatz element selected uniformly at random from the pool.
To provide evidence of a reduction in the number of loss-function calls, we track the number of subpools searched, $\mmax$, to find a local minimum.
The results are depicted in \cref{fig: PoolExplorationResults}.
There is a tendency to terminate subpool exploration after visiting two or three subpools.
This should be compared with the maximum possible QEB-pool values of $\mmax$: $N-2=6, 10, 10, 12, 12$ for \ce{H4}, \ce{LiH}, \ce{H6}, \ce{BeH2}, and \ce{H2O}, respectively. Thus, \cref{fig: PoolExplorationResults} shows that subpool exploration reduces the number of loss-function calls in the cases tested.

\subsection{Reducing ansatz-circuit depth}\label{sec: circuit depth}

\begin{figure}
\includegraphics[width=\columnwidth
]{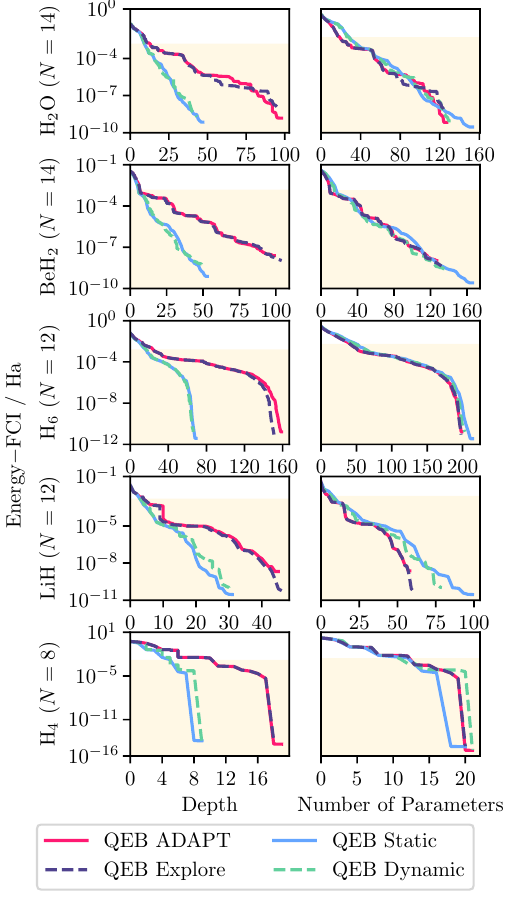}
\caption{The energy accuracy is plotted as a function of ansatz-circuit depth (left) and the number of ansatz-circuit parameters (ansatz elements, right), for QEB standard, Explore-, Static-(TETRIS)-, and Dynamic-ADAPT-VQE. These simulations use support commutation. Each row shows data for a specific molecule. The region of chemical accuracy is shaded.}
    \label{fig: DepthAndParameters}
\end{figure}
Next, we compare the ability of Static-(TETRIS)- and Dynamic-ADAPT-VQE to reduce the depth of the ansatz circuits as compared to standard and Explore-ADAPT-VQE.
The data is depicted in \cref{fig: DepthAndParameters}.
Here, we depict the energy error, 
\begin{align}
    \label{eq: EnergyError}
    \Delta_t = \bound_t - E_{FCI},
\end{align}
given as the distance of the VQE predictions $\bound_t$ from the FCI ground state energy $E_{FCI}$ as a function of (left) the ansatz-circuit depths and (right) the number of ansatz-circuit parameters.
The left column shows that layered ADAPT-VQEs achieve lower energy errors with shallower ansatz circuits.
Meanwhile, the right column demonstrates that all ADAPT-VQEs achieve similar energy accuracy with respect to the number of ansatz-circuit parameters.

\subsection{Reducing runtime}
\label{sec: numerical runtime}

\begin{figure}[t]
    \includegraphics[width=\columnwidth
    ]{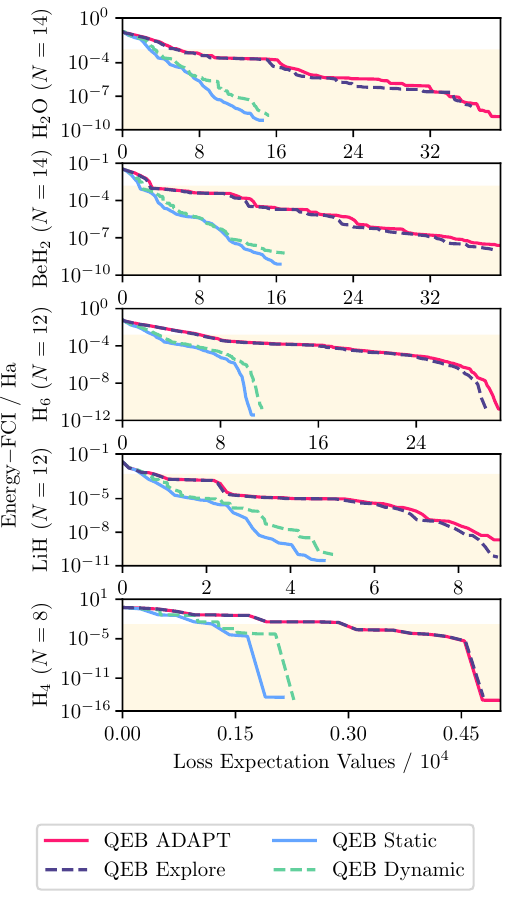}
    \caption{Energy accuracy against the number of loss function calls. Using the QEB pool and support commutation, we compare standard-, Explore-, Static-(TETRIS)-, and Dynamic-ADAPT-VQE. Each row shows data for a specific molecule, with the number of orbitals increasing up the page. Energy accuracies better than chemical accuracy are shaded in cream.}
        \label{fig: ProcessorCallsLoss}
\end{figure}
\begin{figure}[t]
    \includegraphics[width=\columnwidth
    ]{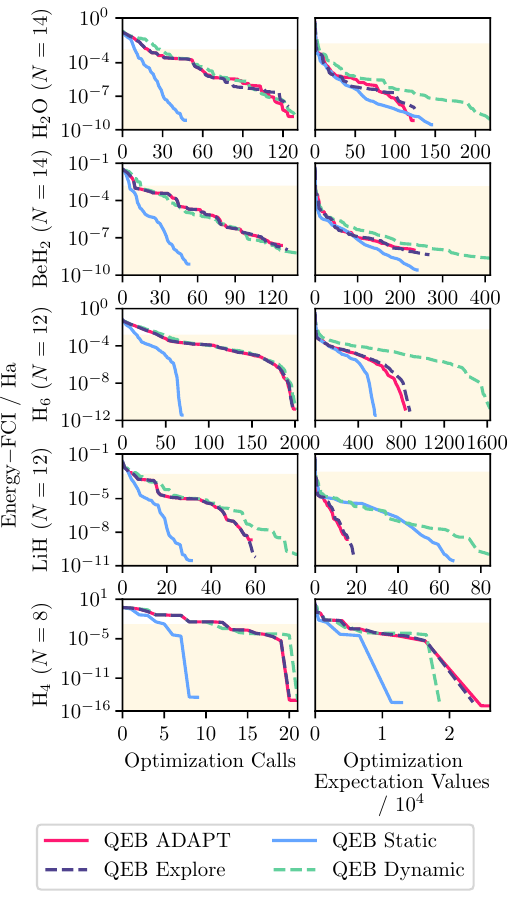}
    \caption{Energy accuracy against the number of times the ansatz is optimized (left column); and the number of expectation values calculated during optimizer calls (right column). Using the QEB pool and support commutation, we compare standard-, Explore-, Static-(TETRIS)-, and Dynamic-ADAPT-VQE. Each row shows data for a specific molecule, with the number of orbitals increasing up the page. Energy accuracies better than chemical accuracy are shaded in cream.}
        \label{fig: ProcessorCallsOpt}
    \end{figure}

In this section, we provide numerical evidence that subpool exploration and layering reduce the runtime of ADAPT-VQE.
A mathematical analysis of asymptotic runtimes will follow in \cref{sec: complexity}.
To provide evidence of a runtime reduction in numerical simulations, we show that layered ADAPT-VQEs require fewer expectation value evaluations (and thus shots and quantum processor runtime) to reach a given accuracy.
Our numerical results are depicted in \cref{fig: ProcessorCallsLoss,fig: ProcessorCallsOpt} for expectation-value evaluations related to calculating losses and parameter optimizations, respectively.
We now discuss our results.

To convert data accessible in numerical simulations (such as loss function and optimizer calls) into runtime data (such as expectation values and shots), we proceed as follows.
For our numerical data, we evaluate runtime in terms of the number of expectation value evaluations rather than processor calls or shots.
This is justified as the number of shots (or processor calls) is directly proportional to the number of expectation values in our simulations, as detailed in \cref{appendix: variance}.
Next, we evaluate the runtime requirements associated with loss-function evaluations by tracking the number of times a loss function is called.
The evaluation of the loss function over a subpool $\mathcal S$ is recorded as $\left|\mathcal S\right|+1$ expectation-value evaluations, assuming the use of a finite difference rule.
Thus we produce the data presented in \cref{fig: ProcessorCallsLoss}.
Finally, we evaluate the runtime requirements of the optimizer by tracking the number of energy expectation values or gradients it requests.
The gradient of $P$ variables is then recorded as $P+1$ energy expectation value evaluations, assuming the use of a finite difference rule.
This gives the data in \cref{fig: ProcessorCallsOpt}.

In \cref{fig: ProcessorCallsLoss}, we show that layered ADAPT-VQEs require fewer loss-related expectation-value evaluations to reach a given energy accuracy. 
We attribute this advantage to subpools gradually shrinking during layer construction.
They thus require fewer loss function evaluations per ansatz element added to the ansatz-element circuit.
We further notice that Explore-ADAPT-VQE does not reduce the loss-related expectation values required for standard ADAPT-VQE.
We attribute this result to our examples' small pool sizes, with only 8 to 14 qubits.
As qubit sizes increase, we expect a more noticeable advantage for Explore-ADAPT-VQE, as discussed in Sec.~\ref{sec: complexity}.

In \cref{fig: ProcessorCallsOpt} (left), we show that Static-ADAPT-VQE reduces the number of optimizer calls needed to reach a given accuracy. 
As expected, the left column shows that Static-ADAPT-VQE calls the optimizer $\bigO{N}$ times less than any other algorithm.
This is expected, as standard, Explore-, and Dynamic-ADAPT-VQE calls the optimizer each time a new ansatz element is added to the ansatz-element circuit.
Meanwhile, Static-ADAPT-VQE calls the optimizer only after adding a whole layer of $\bigO{N}$ ansatz elements to the ansatz-element circuit.
In \cref{fig: ProcessorCallsOpt} (right), we analyze how the reduced number of optimizer calls translates to the number of optimizer-related expectation values required to reach a given accuracy.
The data was obtained using a BFGS optimizer with a gradient norm tolerance of $10^{-12}$ Ha and a relative step tolerance of zero.
Compared to the optimizer calls on the left of the figure, we notice two trends.
Dynamic-ADAPT-VQE, while being on par with standard and Explore-ADAPT-VQE for optimizer calls, tends to use a higher number of expectation value evaluations.
Similarly, Static-ADAPT-VQE, while having a clear advantage over standard and Explore-ADAPT-VQE for optimizer calls, tends to have a reduced advantage (and for LiH, even a disadvantage) when it comes to optimizer-related expectation value evaluations.
These observations hint towards an increased optimization difficulty for layered ADAPT-VQEs.
These observations may be highly optimizer dependant and should be further investigated in the future.

\subsection{Additional bemchmarks}

We close this section by referring the reader to additional benchmarking data presented in the appendices.
In \cref{appendix: commutativity vs. support}, we compare support to operator commutativity for the qubit pool. In \cref{appendix: Largest Energy Reduction Decision}, we compare the steepest-gradient loss to the largest-energy-reduction loss. We also compare the QEB pool to the qubit pool in \cref{appendix: Largest Energy Reduction Decision}.

\section{Runtime analysis}
\label{sec: complexity}

In this section, we analyze the asymptotic runtimes of standard, Explore-, Dynamic-, and Static-ADAPT-VQE.
We find that under reasonable assumptions, Static-ADAPT-VQE can run up to $\bigO{N^2}$ faster than standard ADAPT-VQE.
In what follows, we quantify asymptotic runtimes using $\bigO{x}$, $\bigOmega{x}$, or $\bigTheta{x}$ to state that a quantity scales at most, at least, or exactly with  $x$, respectively.
For definitions, see \cref{appendix: big O}.
We begin our runtime analysis by listing some observations, assumptions, and approximations.
\begin{itemize}
\item[(a)] Each algorithm operates on $N$ qubits.
\item[(b)] Ansatz circuits are improved by successively adding ansatz elements with a single parameter to the ansatz circuit.
This results in  iterations $p=1,...,P$, where the $p$th ansatz circuit has $p$ parameters.
\item[(c)] In each iteration $p$, the algorithm spends runtime on evaluating $N_L(p)$ loss functions.
\item[(d)] In each iteration $p$, the algorithm spends runtime on optimizing $p$ circuit parameters.
\item[(e)] Using the finite difference method, we approximate each of the $N_L(p)$ loss functions in (c) by using two energy-expectation values.
This results in evaluating at most $2 N_L(p)$ energy expectation values on the quantum computer in the $p$th iteration.
\item[(f)] We assume that the optimizer in (d) performs a heuristic optimization of ansatz circuits with $p$ parameters in polynomial time.
Thus, in the $p$th iteration, a quantum computer must conduct $N_O(p) = \bigTheta{p^{\alpha}}$ evaluations of the energy landscape and  $N_O(p)= \bigTheta{p^{\alpha}}$ evaluations of energy expectation values.
\item[(g)] For each energy expectation value in (e) and (f), we assume that a constant number of shots $N_S =\bigTheta{1}$ is needed to reach a given accuracy.
This is a standard assumption in VQE \cite{QCCReview, TILLY20221}, and further details justifying this assumption are discussed in \cref{appendix: variance}.
\item[(h)] For each shot in (g), one must execute an ansatz circuit with $p$ ansatz elements.
Here, we assume that the runtime $C(p)$ of an ansatz circuit with $p$ ansatz elements is proportional to its depth $d(p)$, i.e., $C(p)=\bigTheta{d(p)}$.
\end{itemize}
Combining (e,g,h) and (f,g,h), we can estimate the runtime each algorithm spends on evaluating losses and performing the optimization, respectively:
\begin{subequations}
\begin{align}
 R_L &= \sum_{p=1}^{P} 2 N_L(p) N_S C(p) 
     = \sum_{p=1}^{P} N_L(p) \bigTheta{d(p)}, \\
 R_O &= \sum_{p=1}^{P} N_O(p) N_S C(p)
     = \sum_{p=1}^{P} \bigTheta{p^{\alpha}} \bigTheta{d(p)}.
\end{align}
\end{subequations}
Below we analyze further these runtime estimates for standard, Explore-, Dynamic-, and Static-ADAPT-VQE.

In standard ADAPT-VQE, we re-evaluate the loss of each ansatz element in every iteration. Thus, $N_L(p) = |\pool|$.
Moreover, the circuit depth $d(p)$ is upper bounded by $d(p)=\bigO{p}$.
In the best-case scenario, ADAPT-VQE may arrange ansatz elements into layers accidentally. (An effect more likely for large $N$.)
This can compress the circuit depths down to $d=\bigOmega{p/N}$.
We summarize this range of possible circuit depths using the compact expression $d(p)=\bigTheta{pN^{-\gamma}}$, with $\gamma\in[0,1]$.
In numerical simulations, we typically observe that $\gamma\approx0$, i.e., the depth of an ansatz circuit, is proportional to the number of ansatz elements.
These expressions for $N_L(p)$ and $d(p)$ allow us to estimate the runtime of standard ADAPT-VQE algorithms:
\begin{subequations}
 \label{eq: ADAPT-VQE}
\begin{align}
 R_L^{A} = |\pool|\bigTheta{P^2 N^{-\gamma}}, \\
 R_O^{A} = \bigTheta{P^{2+\alpha_A}  N^{-\gamma}}.
\end{align}
\end{subequations}

Explore-ADAPT-VQE results in circuits of the same depths as ADAPT-VQE, i.e., $d=\bigTheta{pN^{-\gamma}}$.
However, the use of subpool exploration in Explore-ADAPT-VQE may reduce the number of loss-function evaluations $N_L(p)$.
As discussed in \cref{sec:subpool exploration} (paragraph on \textit{Efficiency}), in the best case scenario, the number of loss function evaluations per iteration is lower bounded by $N_L(p) = \bigOmega{|\pool|/N}$.
In the worst case scenario, subpool exploration may explore the whole pool of ansatz elements, such that $N_L(p) = \bigO{|\pool|}$.
Based on these relations, we can estimate the runtime of Explore-ADAPT-VQE:
\begin{subequations}
 \label{eq: Explore-ADAPT-VQE}
\begin{align}
 R_L^{E} &= \begin{cases}
        |\pool| \bigOmega{P^2 N^{-(1+\gamma)}} \\
        |\pool| \bigO{ P^2 N^{-\gamma} }
       \end{cases}, \\
 R_O^{E} &= \bigTheta{P^{2+\alpha_E}  N^{-\gamma}}.
\end{align}
\end{subequations}

Dynamic-ADAPT-VQE has the same scaling of the number of loss function evaluations per iteration, $N_L(p)$, as  Explore-ADAPT-VQE.
Thus, $N_L(p) = \bigOmega{|\pool|/N}$ in the best case and $N_L(p) = \bigO{|\pool|}$ in the worst case.
The circuit depth of  Dynamic-ADAPT-VQE scales as $d(p)=\bigTheta{p/N}$.
One can observe a clear benefit from layering. The  upper bound, $d(p)=\bigO{p}$, in standard and Explore-ADAPT-VQE  becomes $d(p)=\bigO{p/N}$ in Dynamic-ADAPT-VQE.
Using these relations for $N_L(p)$ and $d(p)$, we can estimate the runtime of Dynamic-ADAPT-VQE:
\begin{subequations}
\begin{align}
 R_L^{D} &= \begin{cases}
        |\pool| \bigOmega{P^2 N^{-2}} \\
        |\pool| \bigO{P^2 N^{-1}}
       \end{cases}, \\
 R_O^{D} &= \bigTheta{P^{2+\alpha_D}  N^{-1}}.
\end{align}
\end{subequations}

The analysis of Static-ADAPT-VQE's runtime is more straightforward with respect to the layer count $t$ than to the parameter count $p$. Therefore, we revisit and modify our previous observations, assumptions, and approximations.
\begin{itemize}
\item[(a)] Static-ADAPT-VQE operates on $N$ qubits.
\item[(b)] Static-ADAPT-VQE builds ansatz circuits in  layers indexed by $t=1,...,\tmax$.
The $t$th layer contains $n_{\textrm{tot}}(t)=\bigTheta{N}$ ansatz elements.
Since each ansatz element depends on a single parameter, a layer contains $n_{\textrm{tot}}(t) = \bigTheta{N}$ circuit parameters.
Summing the parameters in each layer gives the total number of parameters in the circuit: $P = \sum_{t=1}^{\tmax}n_{\textrm{tot}}(t)$.
\item[(c)] For each layer $t$, Static-ADAPT-VQE spends runtime on evaluating the loss of $N_L(t)=|\pool|$ ansatz elements.
\item[(d)] For each layer $t$, Static-ADAPT-VQE  spends runtime on optimizing $p(t)=\sum_{t'=1}^{t}n_{\textrm{tot}}(t')=\bigTheta{N}t$ circuit parameters.
\item[(e)] Using the finite difference method, we approximate each of the $N_L(t)$ loss functions in (c) by using two energy expectation values.
This results in evaluating at most $2 N_L(t)=2|\pool|$ energy expectation values on the quantum computer in the $t$th iteration.
\item[(f)] Again, we assume that the optimizer in (d) performs a heuristic optimization of ansatz circuits with $p(t)$ parameters in polynomial time.
Thus, in the $t$th layer a quantum computer must conduct $N_O(t) = \bigTheta{p(t)^{\alpha_S}}$ evaluations of the energy landscape and  $N_O(t)= \bigTheta{p(t)^{\alpha_S}}$ evaluations of energy expectation values.
Using $p(t)=\bigTheta{N}t$ from (d), this implies that $N_O(t) = \bigTheta{N^{\alpha_S}t^{\alpha_S}}$.
\item[(g)] As before, for each energy expectation value in (e) and (f), we assume that a constant number of shots $N_S =\bigTheta{1}$ is needed to reach a given accuracy.
\item[(h)] For each shot in (g), one must execute an ansatz circuit with $p(t)$ ansatz elements.
Again, we assume that the runtime $C(t)$ of an ansatz circuit with $p(t)$ ansatz elements is proportional to its depth $d(p(t))$, i.e., $C(t)= \bigTheta{d(p(t))}$.
Due to layering, the circuit depth of Static-ADAPT-VQE scales as $d(p) = \bigTheta{p/N}$. (This scaling is identical for Dynamic-ADAPT-VQE.)
This results in $C(t)=\bigTheta{p(t)/N}$.
Further, using $p(t)=\bigTheta{N}t$ from (d), we find that each shot in (g) requires a circuit runtime of $C(t) =\bigTheta{t}$.
\end{itemize}
Combining the updated (e,g,h) and (f,g,h), we find the loss- and optimization-related runtimes of Static-ADAPT-VQE, respectively:
\begin{subequations}
\begin{align}
 R_L^S &= \sum_{t=1}^{\tmax}2 N_L(t) N_S C(t)
        = |\pool| \bigTheta{\tmax^2}, \\
 R_O^S &= \sum_{t=1}^{\tmax} N_O(t) N_S C(t) = 
\bigTheta{N^{\alpha_S}\tmax^{\alpha_S+2}}.
\end{align}
\end{subequations}
Since $P=\bigTheta{N}\tmax$ implies $\tmax=P \bigTheta{N^{-1}}$, we can simplify these runtime estimates:
\begin{subequations}
\begin{align}
 R_L^S &= |\pool| \bigTheta{P^2 N^{-2}}, \\
 R_O^S &= \bigTheta{P^{2+\alpha_S}  N^{-2}}.
\end{align}
\end{subequations}

We summarize this section by listing the ratios of asymptotic runtimes for Explore-, Dynamic-, and Static-ADAPT-VQE divided by the asymptotic runtime of standard ADAPT-VQE in \cref{table: ratios}.
Here, we assume equal polynomial scaling ($\alpha_A = \alpha_E = \alpha_D = \alpha_S$) of the optimization runtime for standard, Explore, Dynamic-, and Static-ADAPT-VQE.
As expected from our numerical runtime analysis in \cref{sec: numerical runtime}, for typical ADAPT-VQE circuit depth (where $\gamma=0$), Static-ADAPT-VQE can provide the largest runtime reduction. This reduction is quadratic in the number of qubits: $\bigTheta{N^{-2}}$.
Further improvements to bounding the number of losses in Explore- and Dynamic-ADAPT-VQE are discussed in \cref{appendix: sequences of support sets}.
\begin{table}[h]
    \bgroup
    \def\arraystretch{1.3}
    \centering
    \begin{tabularx}{\columnwidth}{lXlXl}\toprule
        Algorithm && $R_L/R_L^{\text{ADAPT}}$ 
                                       && $R_O/R_O^{\text{ADAPT}}$  
\\\colrule
        Explore   && $\bigOmega{N^{-1}}$, $\bigO{1}$ && $\bigTheta{1}$  
     \\
        Dynamic   && $\bigOmega{N^{-2+\gamma}}$, $\bigO{N^{-1+\gamma}}$ 
                       && $\bigTheta{N^{-1+\gamma}}$ \\
        Static    && $\bigTheta{N^{-2+\gamma}}$  
                       && $\bigTheta{N^{-2+\gamma}}$ \\
    \botrule
    \end{tabularx}
    \egroup
    \caption{The ratio of the runtimes of the listed algorithms to the runtime of standard ADAPT-VQE. $N$ is the qubit number, and $P$ is the number of parameters in the final ansatz circuit. See text for further explanation.
}
\label{table: ratios}
\end{table}

\section{Noise}
\label{sec: noise}

\begin{figure*}[]
    \includegraphics[width=\textwidth, page=1]{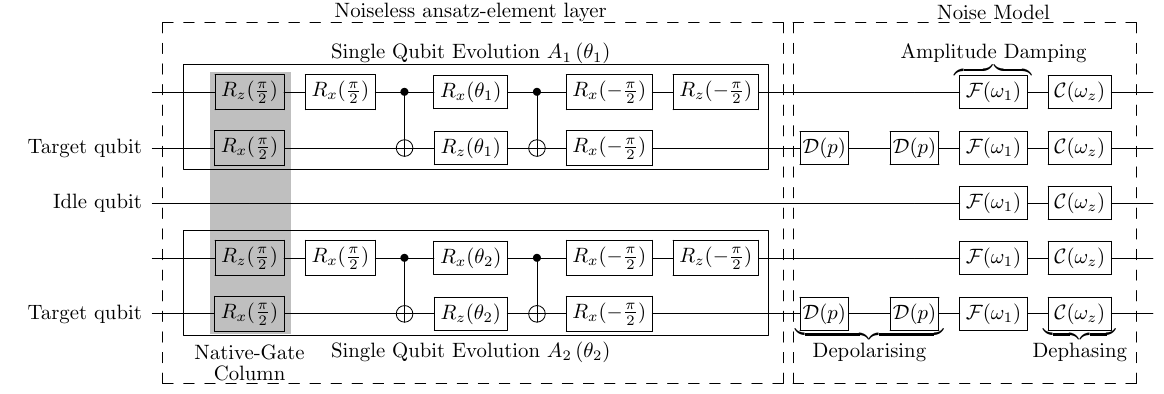}
    \caption{Circuit diagram visualizing the layer-by-layer noise model: on the left, a noiseless ansatz-element layer with two support-commuting ansatz elements is decomposed into columns of native gates. On the right, noise is added to the ansatz-element layer. For global amplitude damping (or dephasing) noise, the channel $\damp$ (or $\dephase$) is applied to each qubit. For local depolarizing noise, the channels $\depolarize$ are applied to the target of the noisy (two-qubit, CNOT) gate.}
    \label{fig: noise model}
\end{figure*}

In this section, we explore the benefits of reducing ADAPT-VQEs' ansatz-circuit depths with respect to noise.
Our main finding is that the use of layering to reduce ansatz-circuit depths mitigates global amplitude-damping and global dephasing noise, where idling and non-idling qubits are affected alike.
However, reduced ansatz-circuit depths do not mitigate the effect of local depolarizing noise, which exclusively affects qubits operated on by noisy (two-qubit, CNOT) gates.
The explanation for this, we show, is that the ansatz-circuit depth is a good predictor for the effect of global amplitude-damping and dephasing noise.
On the other hand, we show that the errors induced by local depolarizing noise are approximately proportional, not to the depth, but to the number of (CNOT) gates.
For this reason, a shallower ansatz circuit with the same number of noisy two-qubit gates will not reduce the sensitivity to depolarizing noise.

\subsection{Noise models}

Our noise models focus on superconducting architectures, where the native gates are arbitrary single-qubit rotations and two-qubit CZ or iSWAP gates \cite{10.1063/1.5089550}.
Further, we assume all-to-all connectivity.
We tune our analysis towards the  \lstinline{ibmq_quito} (IBM Quantum Falcon r4T) processor. For this processor, the quoted two-qubit gate times are for CNOT gates. Thus, we will take CNOT gates to be our native two-qubit gate.
To this end, our simulations use one- and two-qubit gate-execution times of \SI{35.5}{\ns} and \SI{295.1}{\ns}, respectively.\footnote{These values were taken from the IBM Quantum services.}
Similar native gates and execution times apply to silicon quantum processors \cite{Stano2022}.

In our simulations, we model \textit{amplitude damping} of a single qubit by the standard amplitude-damping channel. (For detailed expressions of the amplitude-damping channel and the other noise channels we use, see \cref{appendix: single qubit noise model}.)
Its decay constant is determined by the inverse $T_1$ time: $\omega_1=1/T_1$.
Similarly, we model \textit{dephasing} of a single qubit by the standard dephasing channel.
The $T_1$ and $T_2^*$ times determine its phase-flip probability via the decay constant $\omega_z=2/T_2^*-1/T_1$.
Finally, we model \textit{depolarization} of a single qubit by a symmetric depolarizing channel with depolarization strength $p \in [0,1]$, where $p=0$ leaves a pure qubit pure and $p=1$ brings it to the maximally mixed state.

In our simulations, we model the effects of amplitude damping, dephasing, and depolarizing noise on the ansatz circuits $\Lambda_t$ in a layer-by-layer approach.
This is illustrated in \cref{fig: noise model}.
We decompose the ansatz circuit $\Lambda_t$ into $l = 1, ..., L_t$ layers of support-commuting ansatz-element layers $\left\{\layer_l\right\}$:
\begin{align}
    \label{eq: LayerDecomposition}
 \Lambda_t = \layer^{\text{o}}_{L_t} \circ\cdots\circ \layer^{\text{o}}_l 
\circ\cdots\circ \layer^{\text{o}}_1.
\end{align}
For amplitude damping and dephasing noise, each ansatz-element layer $\layer^{\text{o}}_l$ is transpiled into columns of native gates that can be implemented in parallel (see Ref.~\cite{EfficientCNOT} for more details).
The native gate with the longest execution time of each native-gate column sets the column execution time.
The sum of the column execution times then gives the execution time $\dt_l$ of the ansatz-element layer $\layer^{\text{o}}_l$.
After each ansatz-element layer $\layer^{\text{o}}_l$, amplitude damping is implemented by applying an amplitude-damping channel to every qubit $r=1,..., 
N$ 
in an amplitude-damping layer.
This results in an amplitude-damped ansatz circuit $\Lambda_t(\omega_1)$.
Similarly, after each ansatz-element layer $\layer^{\text{o}}_l$, dephasing is implemented by applying a dephasing channel to every qubit $r=1,..., N$ in a dephasing layer.
This results in a dephased ansatz circuit $\Lambda_t(\omega_z)$.
Finally, for depolarizing noise, we apply the whole ansatz-element layer and then a depolarizing channel to each qubit. The strength of a qubit's depolarizing channel is determined by the exact number of times that the qubit was the target in a CNOT gate in the preceding layer.
This results in a depolarized ansatz circuit $\Lambda_t(p)$.
For a visualization of our layer-by-layer-based noise model, see \cref{fig: noise model}. For detailed mathematical expressions, see \cref{appendix: layer noise model}.

We note that applying the noise channels after each ansatz-element layer could be refined by applying the noise channels after each gate in the ansatz-element layer.
However, as shown in Ref.~\cite{KieransNoise}, such a gate-by-gate noise model, as opposed to our layer-by-layer based noise model, would increase computational costs and has limited effect on the results.
In what follows, we collectively refer to amplitude-damped ansatz circuits $\Lambda_t(\omega_1)$, dephased ansatz circuits $\Lambda_t(\omega_z)$, and depolarized ansatz circuits $\Lambda_t(p)$ as $\Lambda_t(\alpha)$.
Here, $\alpha$ refers to the key noise parameters $\omega_1$, $\omega_z$, or $p$ of each respective noise model.

\subsection{Energy error and noise susceptibility}

\begin{figure}[b]
    \includegraphics[width=\columnwidth, 
page=1]{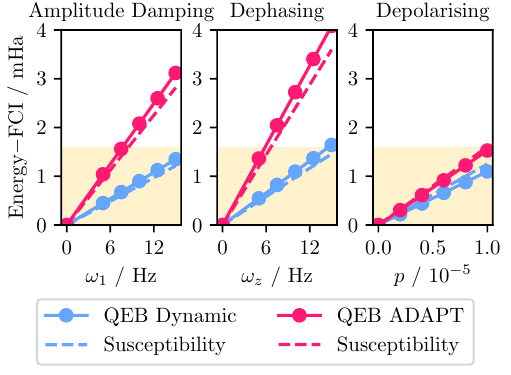}
    \caption{Energy error $\Delta_t(\alpha)$ for an ansatz circuit $\Lambda_t$ of \ce{H4} as a function of noise strength $\alpha$. Connected dots are calculated using full density-matrix simulations. Dashed lines show the corresponding extrapolation using noise susceptibility.}
    \label{fig: NoiseSuscpetibility}
\end{figure}
Going forward, we analyze the effect of noise on the energy error [c.f. \cref{eq: EnergyError}]
\begin{align}
    \label{eq:EnergyErrorNoise}
    \Delta_t(\alpha) = \bound_t(\alpha) - E_{FCI} .
\end{align}
$\Delta_t(\alpha)$ now depends, not only on the iteration step $t$, but also the noise parameter $\alpha$, via the noise-dependent expectation value
\begin{align}
    \label{eq:EnergyErrorNoiseTwo}
    \bound_t(\alpha) = \trace\left[H\Lambda_t(\alpha)\left[\rho_0\right]\right].
\end{align}
To analyze the energy error, we expand the methodology of Ref.~\cite{KieransNoise}.
More specifically, we decompose the energy error into two  contributions:
\begin{align}
  \Delta_t(\alpha) = \Delta_t(0) + \left[\Delta_t(\alpha)-\Delta_t(0)\right].
\end{align}
The first term, $\Delta_t(0)$, is the energy error of the noiseless ansatz circuit. 
The second term, $\Delta_t(\alpha)-\Delta_t(0)$, is the energy error due to noise.
Subsequently, we Taylor expand the energy error due to noise to first order:
\begin{align}
  \left[\Delta_t(\alpha)-\Delta_t(0)\right]=\chi_t\alpha + \bigO{\alpha^2}.
\end{align}
As depicted in \cref{fig: NoiseSuscpetibility}, in the regime of small noise parameters $\alpha$ (where energy errors are below chemical accuracy), the linear approximation is an excellent predictor for the energy error.
Conveniently, this allows us to summarize the effect of noise on the energy error through the noise susceptibility $\chi_t$, defined as
\begin{align}
    \label{eq: NoiseSusceptibilityDef}
  \chi_t \coloneqq
  \left.\frac{\partial\bound_t(\alpha)}{\partial \alpha} \right|_{\alpha=0}
  = \trace\left[H
    \left.\frac{\partial\Lambda_t(\alpha)}{\partial \alpha}\right|_{\alpha=0}
            \left[\rho_0\right]
          \right].
\end{align}
In \cref{appendix: noise susceptibility}, we calculate the noise susceptibility $\chi_t$ of amplitude damping $\damp$, dephasing $\dephase$, and depolarizing $\depolarize$ noise:
\begin{subequations}
\label{eq: NoiseSusceptibilityExplicit}
\begin{align}
 \chi_t^\damp &\coloneqq
 \left.\frac{\partial\bound_t(\omega_1)}{\partial\omega_1}\right|_{\omega_1=0}
 = L_t N \times d\bound(\Lambda_t,\damp),\\
 \chi_t^\dephase &\coloneqq
 \left.\frac{\partial\bound_t(\omega_z)}{\partial\omega_z}\right|_{\omega_z=0}
 = L_t N \times d\bound(\Lambda_t,\dephase),\\
  \chi_t^\depolarize &\coloneqq
 \left.\frac{\partial\bound_t(p)}{\partial p}\right|_{p=0}
 = N_{II} \times d\bound(\Lambda_t,\depolarize),
\end{align}
\end{subequations}
respectively.
Here, $N$ denotes the number of qubits; $L_t$ is the number of ansatz-element layers in the ansatz circuit $\Lambda_t$; $N_{II}$ is the number of noisy (two-qubit, CNOT) gates in the ansatz circuit $\Lambda_t$; and the $d\bound$'s denote the average energy fluctuations, defined in Eqs.~\eqref{eq: AverageEnergyFluctuations} of \cref{appendix: noise susceptibility}.
As discussed further in \cref{appendix: noise susceptibility}, the average energy fluctuations can be calculated from noiseless expectation values.
This allows us to compute the noise susceptibility with faster state-vector simulations rather than computationally demanding density-matrix simulations.

\subsection{Benchmarking layered circuits with noise}

\begin{figure}[]
    \includegraphics[width=\columnwidth, 
page=1]{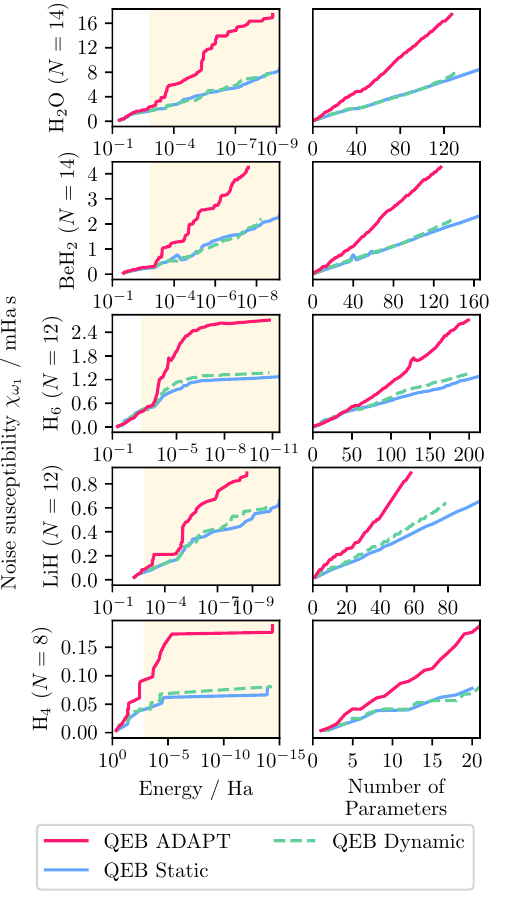}
\caption{Amplitude-damping noise susceptibility as a function of (left) energy accuracy and (right) the number of parameters in ansatz circuits $\Lambda_t$. QEB-ADAPT is compared to support-based Dynamic-ADAPT-VQE. Each row shows data for a specific molecule, with the number of orbitals increasing up the page. Energy accuracies better than chemical accuracy are shaded in cream.}
    \label{fig: amplitude damping}
\end{figure}
\begin{figure}[]
    \includegraphics[width=\columnwidth, 
page=1]{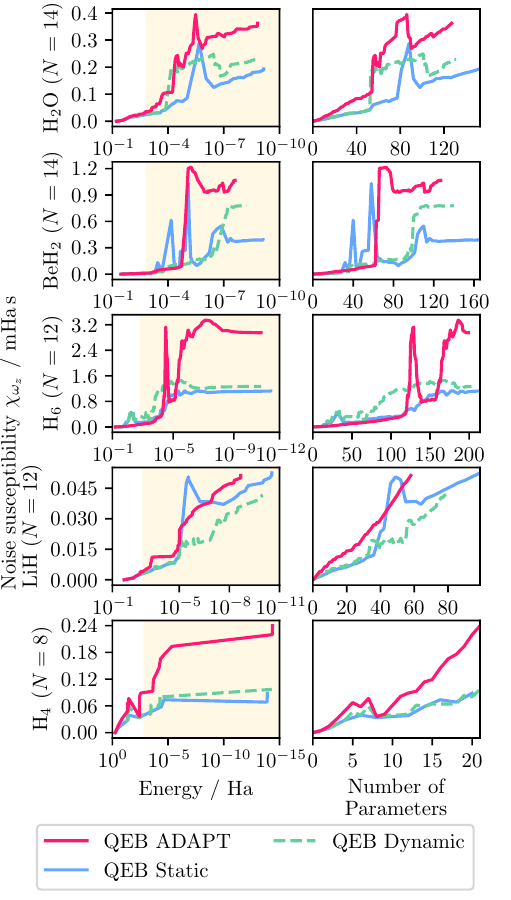}
\caption{Same as \cref{fig: amplitude damping}, but for dephasing noise.}
    \label{fig: dephasing}
\end{figure}
\begin{figure}[]
    \includegraphics[width=\columnwidth, 
page=1]{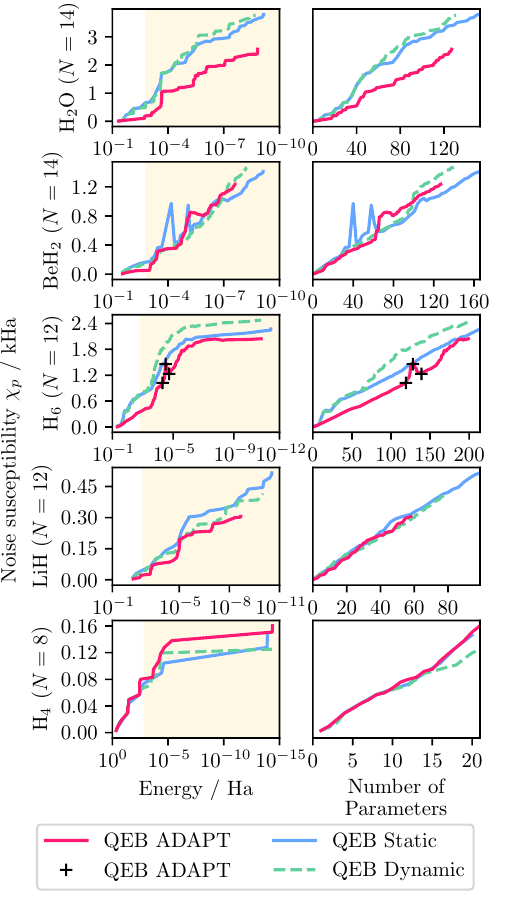}
\caption{Same as \cref{fig: amplitude damping}, but for depolarizing noise. In addition, black crosses correspond to density matrix simulations corroborating noise susceptibility via finite differences. The black crosses are discussed further in \cref{appendix: noise susceptibility peaks}.}
    \label{fig: depolarizing}
\end{figure}

In this section, we compare the noise susceptibility of standard, Static- (TETRIS-), and Dynamic-ADAPT-VQE in the presence of noise.
As before, we showcase these algorithms on a range of molecules (summarized in \cref{table: geometries}) using the QEB pool with support commutativity.
When performing our comparison, we grow the ansatz circuits $\Lambda_t$ and optimize its parameters in noiseless settings, as previously discussed in 
\cite{KieransNoise}.
We then compute the noise susceptibility of $\Lambda_t$ as described in the previous section.
The results for amplitude damping, dephasing, and depolarizing noise are depicted in \cref{fig: amplitude damping}, \cref{fig: dephasing}, and \cref{fig: depolarizing}, respectively.
In all three figures, we plot the noise susceptibility as a function of (left) the noiseless energy accuracy $\Delta_t(0)$ or (right) the number of parameters.
The rows of each plot depict different molecules in order of increasing spin orbitals from bottom to top: \ce{H4}, \ce{H6}, \ce{LiH}, \ce{BeH2}, and \ce{H2O}.

\textit{Layering benefits:---}
From \cref{fig: amplitude damping}, it is evident that layering is successful in mitigating the effect of amplitude-damping noise.
Here, we observe that the noise susceptibility of Static- and Dynamic-ADAPT-VQE is approximately half that of standard ADAPT-VQE.
This is a clear indication that layering can reduce the effect of noise.
In \cref{fig: dephasing}, we observe that layering also tends to reduce the noise susceptibility in the presence of dephasing noise.
However, in this scenario, the advantage is less consistent across different ansatz circuits and molecules.
Finally, in \cref{fig: depolarizing}, we observe that for depolarizing noise, all algorithms tend to produce similar noise susceptibilities.
Sometimes one shows an advantage over the other, and vice versa, depending on the ansatz circuit and molecule.
Our simulations indicate no clear disadvantage of using layering in the presence of depolarizing noise.
In summary, our numerical simulations suggest that layering is useful for mitigating global amplitude damping and dephasing noise.
Moreover, layering seems to have neither a beneficial nor a detrimental effect in the presence of local depolarizing noise.
In order to explain these findings, we further investigate the dependence of noise susceptibility on several circuit parameters in Sec.~\ref{sec: NoiseSusceptibilityScaling}.

\textit{Gate-fidelity requirements:---}
We now use the noise susceptibility data in \cref{fig: amplitude damping}, \cref{fig: dephasing}, and \cref{fig: depolarizing} to estimate the fidelity requirements for operating ADAPT-VQEs.
For this estimation, recall that quantum chemistry simulations of energy eigenvalues target an accuracy of $1.6$ $\mathrm{mHa}$.
To achieve this chemical accuracy, we require the energy error due to noise to be smaller than $\approx 1$ milli-Hartree: $\chi_t \alpha \lesssim 1 \mathrm{mHa}$.
Applying this condition to amplitude damping (where $\alpha=1/T_1$), dephasing (where $\alpha\approx1/T_2^*$), and depolarizing noise (where $\alpha=p$), we find  a set of gate fidelity requirements:
\begin{align}
 T_1 \gtrsim \frac{\chi_t^{\damp}}{1\mathrm{mHa}}, \quad
 T_2^* \gtrsim \frac{\chi_t^{\dephase}}{1\mathrm{mHa}}, \quad
 p \lesssim \frac{1\mathrm{mHa}}{\chi_t^{\depolarize}}.
\end{align}
The data presented in Figs.~\ref{fig: amplitude damping}, \ref{fig: dephasing}, and \ref{fig: depolarizing}, suggests the following requirements for the gate operations to enable chemically accurate simulations:
\begin{align}
 T_1 \gtrsim 1      \mathrm{s},\quad
 T_2^* \gtrsim 100      \mathrm{ms},\quad
 p   \lesssim 10^{-6}.
\end{align}
A more detailed breakdown of the maximal $p$ and minimal $T_1$ and $T_2^*$ times for each algorithm and molecule is presented in \cref{fig: best case parameters}.
These requirements are beyond the current state-of-the-art quantum processors \cite{ding2023highfidelity,Stano2022}.
How much these requirements can be improved by error-mitigation techniques \cite{QEM} remains an open question for future research.

\begin{figure}
    \includegraphics[width=\columnwidth, page=1]{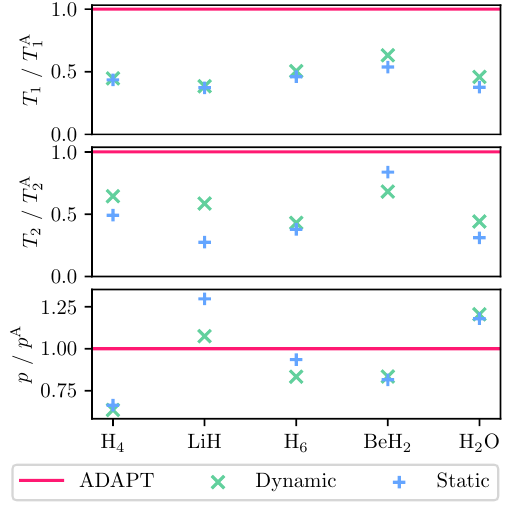}
    \caption{The ratio of the minimal $T_1$ and $T_2^*$ times and maximal depolarizing probability $p$ for Dynamic- and Static-ADAPT-VQE to standard ADAPT-VQE required to reach an accuracy of $\Delta=10^{-7}$ $\mathrm{mHa}$ for each molecule.}
\label{fig: best case parameters}
\end{figure}

\begin{figure*}[]
    \includegraphics[width=\textwidth, 
page=1]{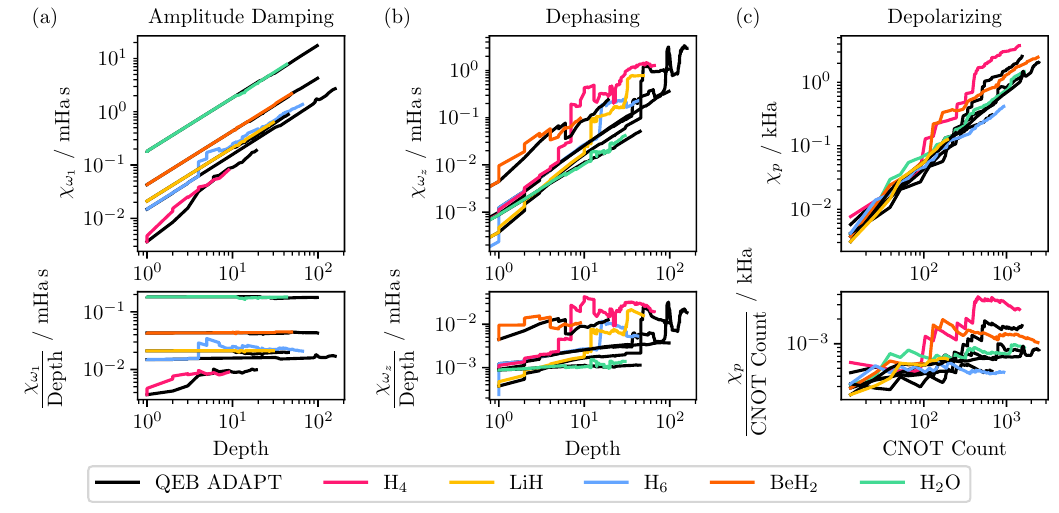}
\caption{Noise susceptibility for (a) amplitude-damping, (b) dephasing, and (c) depolarizing noise as a function of (a, b) ansatz-circuit depths $d$ and (c) the number of noisy CNOT-gates $N_{II}$. The top panels show noise susceptibility as a function of $d$ (a, b) and $N_{II}$ (c). The bottom panels of (a) and (b) show noise susceptibility divided by $d$, as a function $d$. The bottom panel of (c) shows noise susceptibility divided by $N_{II}$ as a function of $N_{II}$. The colored lines correspond to the simulation of Dynamic-ADAPT-VQE using support commutation for a range of molecules. The black lines represent simulations based on QEB-ADAPT-VQE.}
    \label{fig: predictors}
\end{figure*}

\subsection{Noise-susceptibility scalings}
\label{sec: NoiseSusceptibilityScaling}

In this section, we investigate the dependence of noise susceptibility on basic circuit parameters, such as the number of qubits $N$, circuit depth $d\propto L_t$, or the number of noisy (two-qubit, CNOT) gates $N_{II}$.
Our analysis will help in understanding why layering can mitigate global amplitude damping and dephasing noise but not local depolarizing noise.

We study numerically how noise susceptibility scales with circuit depth and the number of noisy (two-qubit, CNOT) gates $N_{II}$. The data is presented \cref{fig: predictors}.
The top panels show the noise susceptibility in the presence of amplitude damping (left), dephasing (center), and depolarizing noise (right) for various algorithms and molecules.
The noise-susceptibility data is presented on a log-log plot as a function of circuit depths $d$ (left and center) as well as $N_{II}$ (right), respectively.
From \cref{fig: predictors}, we find that the noise susceptibility scales roughly linearly with the plotted parameters. 
To further analyze this rough linearity, we produce a log-log plot in the bottom panels of $\chi_t^{\damp}/d$ (left), $\chi_t^{\dephase}/d$ (center), and $\chi_t^{\depolarize}/N_{II}$ (right) as a function of $d$, $d$, and $N_{II}$, respectively. Had the scalings of interest been linear, the bottom panels would have depicted constant curves. This is not entirely the case. But, the curves' deviations from constants are sufficiently sublinear to support our claim that the curves in the upper plots are roughly linear. 

The scalings observed in \cref{fig: predictors} confirm our previous intuition. Based on \cref{eq: NoiseSusceptibilityExplicit}, and using the assumption that $d\bound$ is roughly constant, we would expect that the noise susceptibility in the presence of amplitude damping or dephasing noise is proportional to  the circuit depth and the number of qubits:
\begin{align}
 \label{eq: ScalingADND}
\chi_t^{\damp} \appropto Nd \quad \text{and} \quad \chi_t^{\dephase} \appropto N d.
\end{align}
This claim is supported by \cref{fig: predictors}. Moreover, previous studies \cite{KieransNoise} have found that the noise susceptibility scales linearly with the number of depolarizing two-qubit gates:
\begin{align}
 \label{eq: ScalingDepol}
 \chi_t^{\depolarize} \appropto N_{II}. 
\end{align}
Also, this claim is supported by \cref{fig: predictors}.

Thus, for global (amplitude damping and dephasing) noise, which affects idling and non-idling qubits alike, our analysis indicates that circuit depth is a good predictor of noise susceptibility.
On the other hand, for local (depolarizing) noise, which affects only the qubits which are nontrivially operated on,  $N_{II}$ is a good predictor of the noise susceptibility. Consequently, we expect that compressing the depth of an ansatz circuit by layering can mitigate noise in the former, but not the latter, of these settings.

\section{Summary and Conclusion}
\label{sec: conclusion}

In this paper, we introduced layering and subpool-exploration strategies for ADAPT-VQEs that reduced circuit depth, runtime, and susceptibility to noise.
In noiseless numerical simulations, we demonstrate that layering reduces the depths of an ansatz circuit when compared to standard ADAPT-VQE.
We further showed that our layering algorithms achieve circuits that are as shallow as TETRIS-ADAPT-VQE.
The reduction in ansatz circuit depth is achieved without increasing the number of ansatz elements, circuit parameters, or CNOT gates in the ansatz circuit.
The noiseless numerical simulations further provide evidence that layering and subpool-exploration can reduce the runtime of ADAPT-VQE by up to $\bigO{N^2}$, where $N$ is the number of qubits in the simulation.
Finally, we benchmarked the effect of reducing the depth of ADAPT-VQEs on the algorithms' noise susceptibility.
For global noise models, which affect idling and non-idling qubits alike (such as our amplitude-damping and dephasing model), we show that the noise susceptibility is approximately proportional to the ansatz-circuit depth. 
For these noise models, reduced circuit depth due to layering is beneficial in reducing the noise susceptibility of ADAPT-VQEs.
For local noise models, where only non-idling qubits are affected by noise (as with our depolarizing noise model), we show that the noise susceptibility is approximately proportional to the number of noisy (two-qubit, CNOT) gates.
For these noise models, layering strategies are neither useful nor harmful, as they hardly change the CNOT count of ADAPT-VQEs.
We finish our paper by stating three conclusions from our work.

\textit{To layer or not to layer?:---}Depending on the dominant noise source of a quantum processor, layering may or may not lead to improved noise resilience. For processors where global noise dominates, we recommend layering.

\textit{Static or dynamic layering?:---}Our paper considered static and dynamic layering. Which of the two should be used? Static layering optimizes each layer once, while dynamic layering optimizes the ansatz after adding each ansatz element. Both layering strategies lead to ansatz circuits of similar depths and require a similar number of parameters and CNOT gates to reach a certain energy accuracy. However, static layering calculates significantly fewer energy expectation values on the quantum processor. Therefore,  we recommend static layering for the small molecules studied in this work. For larger molecules, dynamic layering could be preferable.

\textit{How useful is subpool exploration?:---}Our paper introduced a new pool-exploration strategy, that reduces the number of loss-function evaluations and, thereby, the number of calls to the quantum processor.
However, in the examples studied in this work, the number of loss-function evaluations was exceeded by the energy-expectation-value calls. Thus, subpool exploration had little impact on the algorithms. Again, this could change when larger molecules are studied.

\subsection*{Acknowledgements}
The authors thank Yordan S. Yordanov for the use of his codebase for VQE protocols and Wilfred Salmon for insightful discussions.
We further thank Sophia E Economou, Nicholas J Mayhall, Edwin Barnes, Panagiotis G Anastasiou and the Virginia Tech group for fruitful discussions.
We acknowledge the use of IBM Quantum services for this work. The views expressed are those of the authors, and do not reflect the official policy or position of IBM or the IBM Quantum team.

\bibliography{Layering_and_subpool_exploration_for_adaptive_Variational_Quantum_Eigensolvers}

\newlength{\doublecolumnwidth}
\setlength{\doublecolumnwidth}{\columnwidth}

\onecolumngrid

\appendix
\renewcommand{\thesubsection}{\Alph{section}.\arabic{subsection}}

\section{Big O, Omega, and Theta notations}
\label{appendix: big O}
This appendix defines big O, big Omega (Knuth definition), and big Theta notations. For our purposes, the notations can, respectively, be defined as:
\begin{align}
    f\left(x\right)=\bigO{g\left(x\right)}&\iff\lim_{x\to\infty}\frac{f\left(x\right)}{g\left(x\right)}<\infty,\\
    f\left(x\right)=\bigOmega{g\left(x\right)}&\iff\lim_{x\to\infty}\frac{f\left(x\right)}{g\left(x\right)}>0,\\
    f\left(x\right)=\bigTheta{g\left(x\right)}&\iff 0<\lim_{x\to\infty}\frac{f\left(x\right)}{g\left(x\right)}<\infty.
\end{align}

\section{Variance, Required Shots, and Runtime}
\label{appendix: variance}

\begin{figure*}[htp]
    \begin{minipage}[t]{\doublecolumnwidth}
    \begin{figure}[H]
        \includegraphics[width=\columnwidth, page=1]{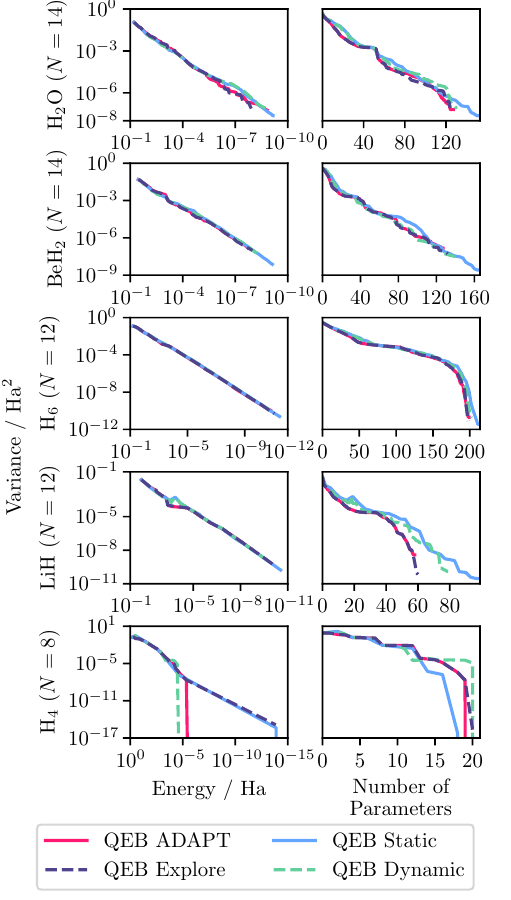}
        \caption{Variance against energy accuracy (left) and the number of ansatz-circuit parameters (ansatz elements, right), for standard, Explore-, Static-, and Dynamic-ADAPT-VQE using the QEB pool and support commutation. Each row shows data for a specific molecule, with the number of orbitals increasing up the page. Energy accuracies better than chemical accuracy are shaded in cream.}
        \label{fig: variances}
    \end{figure}
    \end{minipage}\hfill\begin{minipage}[t]{\doublecolumnwidth}
    \begin{figure}[H]
        \includegraphics[width=\columnwidth, page=1]{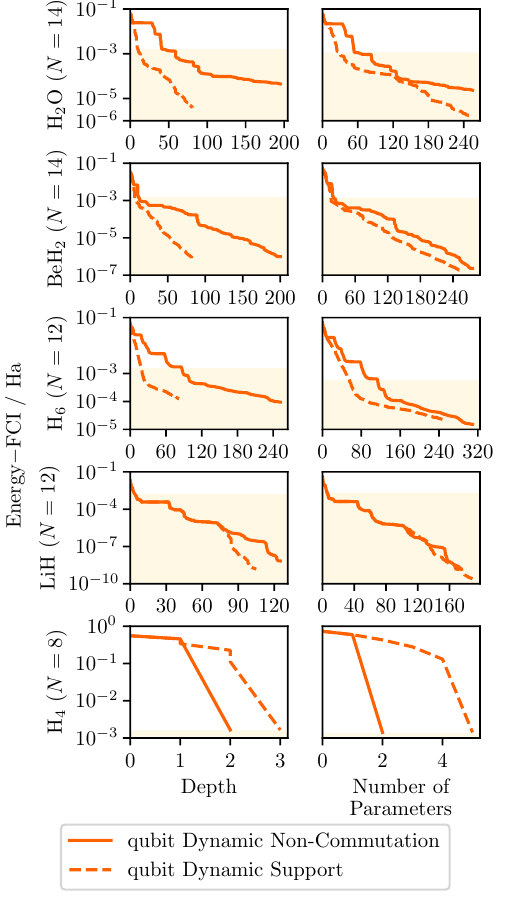}
        \caption{Energy accuracy against ansatz-circuit depths (left) and the number of ansatz-circuit parameters (ansatz elements, right), for qubit-Dynamic-ADAPT-VQE using support and operator commutation---dashed and solid lines, respectively. Each row shows data for a specific molecule, with the number of orbitals increasing up the page. Energy accuracies better than chemical accuracy are shaded in cream.}
        \label{fig: support vs. commutation}
    \end{figure}
\end{minipage}
\end{figure*}

\begin{figure*}[htp]
    \includegraphics[width=\textwidth, page=1]{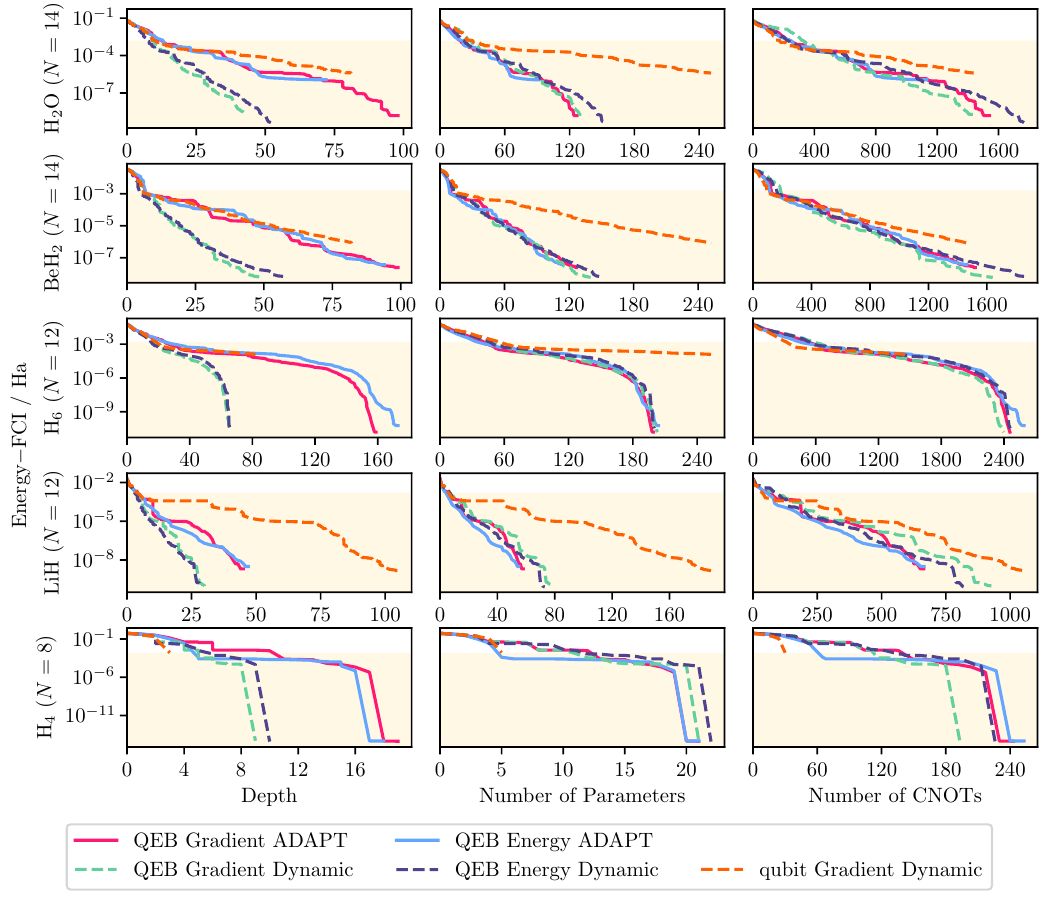}
    \caption{Energy accuracy against ansatz-circuit depths (left), the number of ansatz-circuit parameters (ansatz elements, center), and the number of CNOT gates in the ansatz circuit (right), for qubit-Dynamic-ADAPT-VQE using support commutation as well as QEB standard and Dynamic-ADAPT-VQE using support commutation with both steepest gradient and largest energy reduction decision rules, are compared. Each row shows data for a specific molecule, with the number of orbitals increasing up the page. Energy accuracies better than chemical accuracy are shaded in cream.}
    \label{fig: energy vs. gradient}
\end{figure*}

\begin{figure*}[htp]
    \includegraphics[width=\textwidth, page=1]{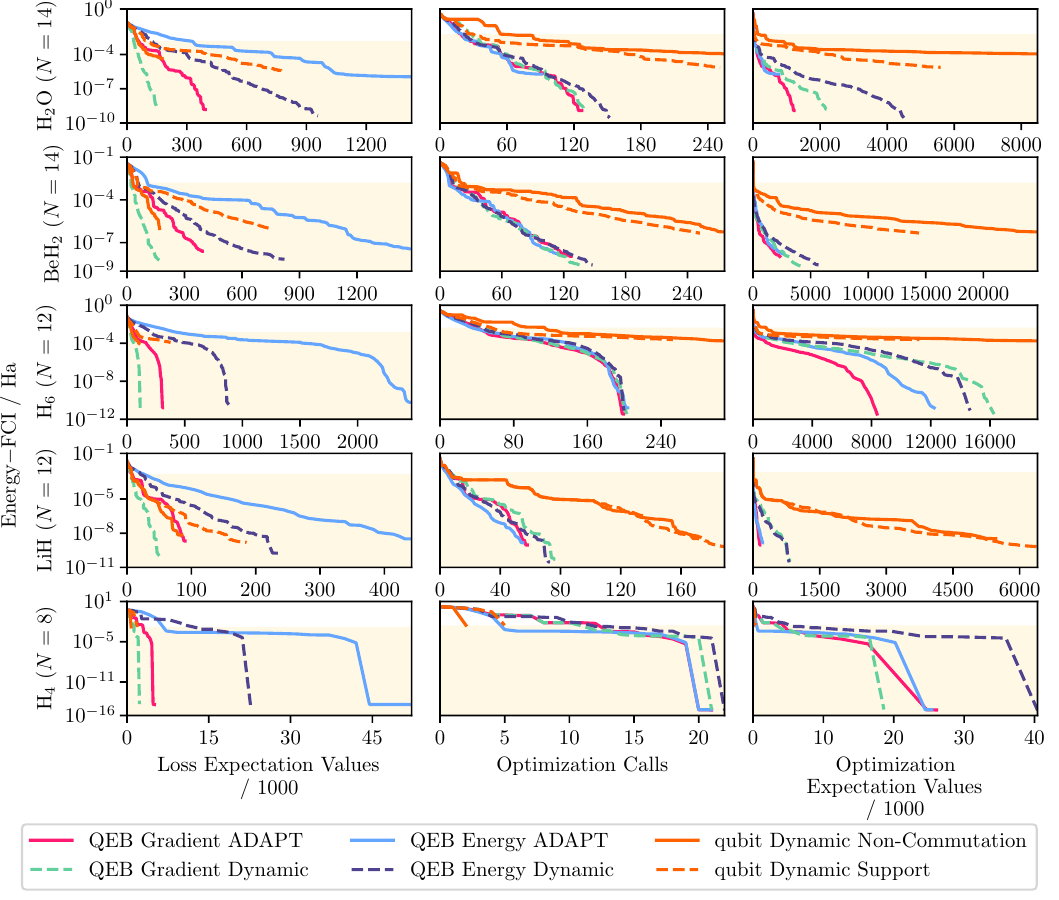}
    \caption{Energy accuracy against the number of loss function calls (left), the number of times the ansatz is optimized (center), and the number of expectation values calculated during optimizer calls (right), for qubit-Dynamic-ADAPT-VQE using support and operator commutation, as well as QEB standard and Dynamic-ADAPT-VQE using support commutation with both the steepest gradient and largest energy reduction decision rules, are compared. Each row shows data for a specific molecule, with the number of orbitals increasing up the page. Energy accuracies better than chemical accuracy are shaded in cream.}
    \label{fig: energy vs. gradient calls}
\end{figure*}

This appendix analyzes the relationship between quantum-processor calls and expectation value evaluations. As in Sec.~\ref{sec: numerical runtime} and Sec.~\ref{sec: complexity}, here we will assume gradients are calculated through the method of finite differences. Thus, we need only consider expectation values of the Hamiltonian $H$. Note that the runtimes to calculate two different expectation values will not generally be the same. The standard error $\epsilon$ in the estimate of the expectation value of the Hamiltonian $H$ can be bounded using Chebyshev's inequality to give
\begin{equation}
\mathbb P\left(\left|\widetilde{\expval{H}}-\expval{H}\right|\ge\epsilon\right)\le\frac{\operatorname{Var}H}{s\epsilon^2}
\end{equation}
where $s$ is the number of independent samples of the observable $H$ and $\widetilde{\expval{H}}$ is the sample mean estimator for the expectation value $\expval{H}$. However, our Hamiltonian $H$ is represented by a sum of $g$ Pauli strings, so the total number of shots is $S=gs$. Thus, we find the number of shots required is bounded:
\begin{equation}\label{eq: required shots}
    S\le\frac{g\operatorname{Var}H}{q\epsilon^2}.
\end{equation}
where we have picked some desired confidence $q=\mathbb P\left(\left|\widetilde{\expval{H}}-\expval{H}\right|\ge\epsilon\right)$.

Two factors contribute to variation in expectation value runtimes. Assuming equality in \cref{eq: required shots} allows us to read off the first factor. That is, the runtime is directly proportional to
\begin{equation}\label{eq: variance}
    \operatorname{Var}H\equiv\Trace\left[H^2\rho_t\right]-\left(\Trace\left[H\rho_t\right]\right)^2,
\end{equation}
which will generally vary throughout the algorithm through the dependence on $\rho_t$. The evolutions of the variance $\operatorname{Var}H$ throughout the four algorithms benchmarked in Sec.~\ref{sec: circuit depth} and \ref{sec: numerical runtime} are presented in \cref{fig: variances}. We observe that the variance is well predicted by the energy accuracy independent of the algorithm (left column). This leads to the variance, like the energy accuracy, being well predicted by the number of ansatz parameters independent of the algorithm (right column).

The second factor is that each shot will have a runtime directly proportional to the ansatz-circuit depth---this assumes that initialization and readout are negligible.

Thus, we treat each expectation value evaluation as an oracle. The cost of each oracle call is approximately the same for each algorithm for equal energy accuracies or equal numbers of ansatz parameters---providing the algorithms produce ansatz circuits with approximately equivalent depths. This is not true when comparing Static- and Dynamic-ADAPT-VQE to standard and Explore-ADAPT-VQE. That is the number of shots required per expectation value $S\left(x\left(P\right)\right)$ (see Sec.~\ref{sec: complexity}) as a function of the number of parameters that is approximately the same for the four algorithms considered.

Further, comparisons of total expectation value evaluations (see Sec.~\ref{sec: numerical runtime}) for a given energy convergence are a good proxy for directly comparing the runtimes providing the algorithms produce ansatz circuits with approximately equivalent ansatz-circuit depths. In the layered \textit{versus} non-layered comparison, the layering-based algorithms will have a faster runtime given by the ratio of the depths---at best $\bigO{N}$.

\section{Commutativity vs Support}
\label{appendix: commutativity vs. support}

In this article, we introduced two notions of commutation that can be leveraged in constructing ADAPT-VQE algorithms. They were operator commutation and support commutation, and we will compare them here.

As noted in \cref{appendix: non-commuting sets}, the operator and support non-commuting sets of the ansatz elements in the QEB pool differ by, at most, two ansatz elements. Additionally, both notions of commutation result in layers of constant depth with respect to the number of qubits. On the other hand, the operator and support non-commuting sets of the qubit pool differ by $\bigTheta{N^3}$ ansatz elements. Thus, the two types of commutation could construct vastly different ansätze. The ansätze constructed based on support commutation will have layers of constant depth, while those constructed based on operator commutation will have layers of depth $\bigO{N^3}$. Thus, we consider both the support and operator commutation variants of Dynamic-ADAPT-VQE using the qubit pool to highlight the differences between the commutation types. \cref{fig: support vs. commutation} shows the energy error (\cref{eq: EnergyError}) as a function of depth and the number of parameters. For \ce{LiH}, \ce{H6}, \ce{BeH2}, and \ce{H2O} we observe that support commutation outperforms operator commutation. However, for \ce{H4}, we find operator commutation outperforms support commutation.

\section{Pool and Decision Rule Comparisons}\label{appendix: Largest Energy Reduction Decision}

This appendix includes additional comparisons of Dynamic-ADAPT-VQE with two different loss functions and with the QEB and Pauli pools. We refer to the loss function in \cref{eq: loss} as gradient selection. An alternative loss function is
\begin{equation}
    L_t(A) = \min_\theta E_{t,A}(\theta),
\end{equation}
which we will refer to as energy selection. That is, we optimize the ansatz with respect to the last parameter for each ansatz element in the subpool and pick the ansatz element that reduces the energy by the most.

In \cref{fig: energy vs. gradient}, we see that both the energy and gradient selection rules perform similarly in energy accuracy for a given depth, number of parameters, and number of CNOT gates. Standard ADAPT-VQE with both gradient and energy selection are included as reference points.

The QEB-Dynamic-ADAPT-VQE algorithms require fewer parameters and shallower ansatz circuits than the qubit-Dynamic-ADAPT-VQE algorithms for a given energy accuracy. However, the energy accuracy of QEB-Dynamic-ADAPT-VQE and qubit-Dynamic-ADAPT-VQE for a given number of CNOT gates is similar, suggesting that the number of two-qubit gates in the ansatz could be a good pool-independent predictor of energy convergence.

Additionally, in \cref{fig: energy vs. gradient calls}, we see that evaluating the energy selection rule is more expensive than the gradient selection rule. However, as the optimization dominates the total number of expectation values we see for Dynamic-ADAPT-VQE, the selection rule makes little difference for \ce{LiH}, \ce{H6}, and \ce{BeH2}. That said, the energy selection rule never significantly outperforms the gradient selection rule, which justifies the use of the gradient selection rule throughout this article.

\section{Noise Susceptibility Peaks}
\label{appendix: noise susceptibility peaks}

In this appendix, we comment on the peaks which appear in the noise susceptibility data of \cref{fig: dephasing} and \cref{fig: depolarizing}.
To verify that these apparent features are not numerical errors, we computed not only the noise susceptibility but also performed full density-matrix simulations with depolarizing noise for a few points on the \ce{H6} plots.
We compute $\bound_t(p=10^{-6})$ and $\bound_t(p=0)$ at the values of $t$ located, before the peak, at the peak, and after the peak in \cref{fig: depolarizing}, respectively. 
Then, we used the finite-differences method to estimate $\chi_t^{\depolarize}$.
The corresponding data points are depicted with black crosses in \cref{fig: depolarizing}, row 3, for \ce{H6}.
These points are in perfect agreement with the noise-susceptibility data.

\section{Pool Definitions}
\label{appendix: pool definitions}

In this appendix, we make more rigorous definitions of the pools referred to throughout this article. These definitions will prove useful in analyzing the runtime scalings of the algorithms.

Let $\mathcal V$ be some set of even integers. In the body of this paper, we take $\mathcal V=\left\{2, 4\right\}$ for all three pool definitions. Additionally, let
\begin{equation}
    \mathcal I\left(\mathcal V\right)\coloneqq\left\{\left(\vec k,\vec l\:\right)\colon\vec k, \vec l\in\mathbb Z^{q/2}_n\text{ such that } k_0\le k_i,l_i\text{ and all }k_i,l_i\text{ are distinct}\quad\forall i\in\left[q/2\right]\quad\forall q\in\mathcal V\right\}
\end{equation}

\begin{definition}[Generalized Fermionic Pool]
    All the distinct fermionic excitations that act on $q\in\mathcal V$ distinct qubits are given by:
    \begin{equation}
        \pool^{f}\left(N\right)\coloneqq\left\{T_{\vec l}^{\vec k}\colon\left(\vec k,\vec l\:\right)\in\mathcal I\left(\mathcal V\right)\right\}\quad\text{where}\quad T_{\vec l}^{\vec k}\coloneqq\left[\prod_{i\in \vec k}a^\dagger_i\right]\left[\prod_{j\in \vec l}a_j\right]-h.c.
    \end{equation}
    This is the fermionic pool over $N$ qubits.
\end{definition}

\begin{definition}[Generalized QEB Pool]
    All the distinct qubit excitations that act on $q\in\mathcal V$ distinct qubits are given by:
    \begin{equation}
        \pool^{\text{\normalfont QEB}}\left(N\right)\coloneqq\left\{T_{\vec l}^{\vec k}\colon\left(\vec k,\vec l\:\right)\in\mathcal I\left(\mathcal V\right)\right\}\quad\text{where}\quad T_{\vec l}^{\vec k}\coloneqq\left[\prod_{i\in \vec k}Q^\dagger_i\right]\left[\prod_{j\in \vec l}Q_j\right]-h.c.
    \end{equation}
    This is the QEB pool over $N$ qubits.
\end{definition}

The cardinality of both the fermionic and QEB pools is:
\begin{equation}\label{eq:f QEB cardinality}
    \left|\pool^{f,\text{QEB}}\left(N\right)\right|=\sum_{\substack{q\in\mathcal V\\\colon q\le N}}\begin{pmatrix}q-1\\\frac{1}{2}q\end{pmatrix}\begin{pmatrix}N\\q\end{pmatrix}=\bigTheta{N^{\max\mathcal V}},
\end{equation}
where the second choice is the number of sets of $q$ distinct qubits of $N$, and the first is the number of permutations of these qubits that give distinct excitations.

\begin{definition}[Generalized qubit Pool]
    The qubit pool is the set of all the Pauli excitations with an odd number of $Y$ gates in the generator that act on $q\in\mathcal V$ distinct qubits.
\end{definition}

The cardinality of the qubit pool is:
\begin{equation}\label{eq: qubit cardinality}
    \left|\pool^{\text{qubit}}\left(N\right)\right|=\sum_{\substack{q\in\mathcal V\\\colon q\le N}}2^{q-1}\begin{pmatrix}N\\q\end{pmatrix}=\Theta\left(N^{\max\mathcal V}\right),
\end{equation}
where the combinatorics follow as for the fermionic and QEB pools, but the first choice is replaced with the factor $2^{q-1}$.

\begin{lemma}\label{lemma: logarithmically concave}
    For each $\pool\in\left\{\pool^{\rm{Fermi}}, \pool^{\text{\normalfont 
QEB}}, \pool^{\text{\normalfont qubit}}\right\}$, there exists a finite constant 
$V$, that will depend on the pool definition, such that 
$\left|\pool\left(N\right)\right|$ is a logarithmically concave function of $N$ 
for $N\ge V$.
\end{lemma}
\begin{proof}
    Consider the polynomial
    \begin{equation}
        f_q\left(x\right)\coloneqq\prod_{n=0}^{q-1}\left(x-n\right),
    \end{equation}
    which has $q$ integer roots: $\left[0,q-1\right]$. We note that the $k$th derivative can be expressed as
    \begin{equation}
        \dv[k]{f_q}{x}=\sum_{\substack{\vec\alpha\in\left[0,q-1\right]^k\\\colon\alpha_i\ne\alpha_j\forall i,j}}f_{q,\vec\alpha}\left(x\right)\quad\text{where}\quad f_{q,\vec\alpha}\left(x\right)\equiv\prod_{\substack{n\in\left[0,q-1\right]\\\colon n\not\in\vec\alpha}}\left(x-n\right).
    \end{equation}

    Further, let
    \begin{equation}
        g_{q,\vec\alpha}\left(x\right)\coloneqq\frac{f_q^2\left(x\right)}{\displaystyle\prod_{l\in\vec\alpha}\left(x-l\right)}.
    \end{equation}
    Using this notation, we will consider the following products
    \begin{equation}
        \dv[2]{f_q}{x}f_{q}\left(x\right)=\sum_{\substack{\vec\alpha\in\left[0,q-1\right]^2\\\colon\alpha_1\ne\alpha_2}}f_{q,\vec\alpha}\left(x\right)f_q\left(x\right)=\sum_{\substack{\vec\alpha\in\left[0,q-1\right]^2\\\colon\alpha_1\ne\alpha_2}}g_{q,\vec\alpha}\left(x\right),
    \end{equation}
    and
    \begin{equation}
        \dv{f_q}{x}\dv{f_q}{x}=
        \sum_{\vec\alpha\in\left[0,q-1\right]^2}f_{q,\alpha_1}\left(x\right)f_{q,\alpha_2}\left(x\right)=
        \sum_{\substack{\vec\alpha\in\left[0,q-1\right]^2\\\colon\alpha_1\ne\alpha_2}}g_{q,\vec\alpha}\left(x\right)+\sum_{\alpha_1=0}^{q-1}g_{q,\left(\alpha_1,\alpha_1\right)}\left(x\right).
    \end{equation}

    Now note the first terms from each cancel in the difference of these products:
    \begin{equation}
        \dv{f_q}{x}\dv{f_{q'}}{x}-\dv[2]{f_q}{x}f_{q'}\left(x\right)=\sum_{\alpha_1=0}^{q-1}g_{q,\left(\alpha_1,\alpha_1\right)}\left(x\right),
    \end{equation}
    which is non-negative for $x\ge q$.

    Note $f_q\left(x\right)$ is logarithmically concave within some convex domain iff $\dv[2]{f_q}{x}f_q\left(x\right)\le\left[\dv{f_q}{x}\right]^2$ within the convex domain. Thus, $f_q\left(x\right)$ is logarithmically concave for $x\ge q$. As the discrete function $\begin{pmatrix}N\\q\end{pmatrix}=\frac{1}{q!}f_q\left(N\right)$ it must also be logarithmically concave for $N\ge q$. However, $\left|\pool\left(N\right)\right|$ is a linear combination with non-negative coefficients of such functions with $q\in\mathcal V$. Therefore there will be a finite constant depending on these coefficients for which the function $\begin{pmatrix}N\\\max\mathcal V\end{pmatrix}$ dominates sufficiently that $\left|\pool\left(N\right)\right|$ is logarithmically concave.
\end{proof}

\section{Generalized Non-Commutation Sets}
\label{appendix: non-commuting sets}

In this appendix, we derive rules to determine whether two ansatz elements from the same pool operator commute. We consider the qubit pool (\cref{appendix: qubit pool}), QEB pool (\cref{appendix: QEB pool}), and Fermionic pool (\cref{appendix: Fermionic pool}). All these ansatz elements are Stones encoded unitaries, $e^{\theta T}$ for some skew-hermitian generator $T$. Thus, two ansatz elements operator commute iff the corresponding generators commute. Below we derive the conditions under which the generators commute and hence the ansatz elements operator commute.

\subsection{Pauli Excitations}
\label{appendix: qubit pool}

For Pauli excitations, the generators are simply the Pauli strings of length two and four with an odd number of $Y$ operators. Thus, the Pauli strings operator commute iff the tensor factors in the strings with mutual support differ in an even number of places.

\subsection{Qubit Excitations}
\label{appendix: QEB pool}

First, we will consider a generalization of qubit excitation generators which will later prove useful for Fermionic excitation generators. Consider the following definitions:

\begin{definition}[Singleton Matrix]\label{def:M matrix}
    Let $\left[M_{ij}\right]_{ab}\coloneqq \delta_{ai}\delta_{bj}$ be a singleton matrix. Note that these matrices have the following properties:
    \begin{enumerate}
        \item $M_{kl}M_{\gamma\delta}\equiv\delta_{l\gamma}M_{k\delta}$,
        \item $\left[M_{kl}, M_{\gamma\delta}\right]\equiv\delta_{l\gamma}M_{k\delta}-\delta_{k\delta}M_{\gamma l}$.\label{property:M commutator}
    \end{enumerate}
\end{definition}
\begin{definition}[Set of Singleton Matrices]\label{def:M matrix set}
    Let $\mathcal M_l\coloneqq\left\{M_{ij}\colon i,j\in\left[2^l\right]\right\}\subset\mathbb R^{2^l\times 2^l}$ for all $l\in\mathbb N_0$. Note that this family of sets has the following properties:
    \begin{enumerate}
        \item $\mathcal M_0\equiv\left\{1\right\}$,
        \item $\mathcal M_{l+m}\equiv\mathcal M_l\otimes\mathcal M_m\quad\forall l\ge 1$.\label{property:factorisation}
    \end{enumerate}
\end{definition}

Representing the set of operators $\mathcal Q_l\coloneqq\left\{Q,Q^\dagger\right\}^{\otimes l}$ in the computational basis is an injection $\mathcal Q_l\to\left\{M_{ij}\in\mathcal M_l\colon i\ne j\right\}$ for $l\ge 1$. Thus, we consider the following skew-Hermitian operator $T=M_{ij}+a M_{ji}$, which could represent a qubit excitation generator in the computational basis when $a=-1$.

\begin{theorem}[Singleton Matrix Excitation Commutation]\label{theorem: singleton matrix excitation commutation}
    Consider the following two operators acting on a tripartite vector space:
    \begin{align}
        T_1&=\mathds{1}\otimes M_{ij}\otimes M_{kl}+a_1\mathds{1}\otimes M_{ji}\otimes M_{lk},\\
        T_2&=M_{\alpha\beta}\otimes \mathds{1}\otimes M_{\gamma\delta}+a_2M_{\beta\alpha}\otimes \mathds{1}\otimes M_{\delta\gamma},
    \end{align}
    such that $i\ne j$, $k\ne l$, $\alpha\ne \beta$, and $\gamma\ne \delta$.
    
    $T_1$ and $T_2$ commute if and only if any of the following three conditions hold:
    \begin{enumerate}
        \item $T_1$ and $T_2$ have disjoint support (i.e. $M_{kl}=M_{\gamma\delta}=1$),\label{condition: first}
        \item $T_1\propto T_2$,\label{condition: second}
        \item $T_1$ and $T_2$ have equivalent support and $k$, $l$, $\gamma$ and 
$\delta$ are all distinct.\label{condition: third}
    \end{enumerate}
\end{theorem}
\begin{proof}
    Condition \ref{condition: first} follows trivially, as operators with disjoint support always commute. Thus, henceforth, we will assume $M_{kl},M_{\gamma\delta}\in\mathcal M_l$ for $l\ge 1$ (the compliment of Condition \ref{condition: first}). Now consider the product
    \begin{align}
        T_1T_2=
        &M_{ij}\otimes M_{\alpha\beta}\otimes \left(M_{kl}M_{\gamma\delta}\right)\\
        &+a_1M_{ji}\otimes M_{\alpha\beta}\otimes\left(M_{lk}M_{\gamma\delta}\right)\\
        &+a_2M_{ij}\otimes M_{\beta\alpha}\otimes \left(M_{kl}M_{\delta\gamma}\right)\\
        &+a_1a_2M_{ji}\otimes M_{\beta\alpha}\otimes\left(M_{lk}M_{\delta\gamma}\right).
    \end{align}
    
    Thus, we find the commutator using Property \ref{property:M commutator} of \cref{def:M matrix}:
    \begin{align}
        \left[T_1,T_2\right]=
        &M_{ij}\otimes M_{\alpha\beta}\otimes \left(\delta_{l\gamma}M_{k\delta}-\delta_{k\delta}M_{\gamma l}\right)\label{eq: commutator 1}\\
        &+a_1M_{ji}\otimes M_{\alpha\beta}\otimes\left(\delta_{k\gamma}M_{l\delta}-\delta_{l\delta}M_{\gamma k}\right)\label{eq: commutator 2}\\
        &+a_2M_{ij}\otimes M_{\beta\alpha}\otimes \left(\delta_{l\delta}M_{k\gamma}-\delta_{k\gamma}M_{\delta l}\right)\label{eq:commutator 3}\\
        &+a_1a_2M_{ji}\otimes M_{\beta\alpha}\otimes\left(\delta_{k\delta}M_{l\gamma}-\delta_{l\gamma}M_{\delta k}\right).\label{eq: commutator 4}
    \end{align}
    
    First, suppose the tensor factors in parentheses are non-zero and $M_{ij}\ne 1$. Line \ref{eq: commutator 1} cannot cancel with Line \ref{eq: commutator 2} or \ref{eq: commutator 4} due to the first tensor factor: $M_{ji}\ne M_{ij}$. 
Further, Line \ref{eq: commutator 1} cannot cancel with Line \ref{eq:commutator 
3} due to the bracketed tensor factor as $\gamma\ne\delta$. Thus, $T_1$ and 
$T_2$ do not commute if $M_{ij}\ne 1$ and the first condition is not met. By symmetry, the same argument can be applied if we suppose the tensor factors in parentheses are non-zero and $M_{\alpha\beta}\ne 1$. Therefore, if the first condition is not met, then we require $M_{ij}=M_{\alpha\beta}=1$ in order for $T_1$ and $T_2$ to commute---that is, we require equivalent support.

    Suppose now that $T_1$ and $T_2$ do have equivalent support. We can simplify the commutator to:
    \begin{align}
        \left[T_1,T_2\right]=
        &\delta_{l\gamma}M_{k\delta}-\delta_{k\delta}M_{\gamma l}\label{eq:commutator simple 1}\\
        &+a_1\left(\delta_{k\gamma}M_{l\delta}-\delta_{l\delta}M_{\gamma k}\right)\label{eq:commutator simple 2}\\
        &+a_2\left(\delta_{l\delta}M_{k\gamma}-\delta_{k\gamma}M_{\delta l}\right)\label{eq:commutator simple 3}\\
        &+a_1a_2\left(\delta_{k\delta}M_{l\gamma}-\delta_{l\gamma}M_{\delta k}\right).\label{eq:commutator simple 4}
    \end{align}

    By noting if $\delta_{ab}=1$ and $a\ne c$ then $\delta_{cb}=0$ and that $M_{ab}$ is linearly independent from $M_{cd}$ if $a\ne c$ or $b\ne d$ we can pair the lines as follows: $A\coloneqq\left\{\eqref{eq:commutator simple 1},\eqref{eq:commutator simple 4}\right\}$ and $B\coloneqq\left\{\eqref{eq:commutator simple 2},\eqref{eq:commutator simple 3}\right\}$ where no term in $A$ can cancel with a term in $B$. For the terms in $A$ to cancel we require $k=\delta$ and $l=\gamma$ and $a_1a_2=1$ as $k\ne l$ and $\gamma\ne\delta$, or $k\ne\delta$ and $l\ne\gamma$. Similarly, for terms in $B$ to cancel we require $k=\gamma$ and $l=\delta$ and $a_1=a_2$ as $k\ne l$ and $\gamma\ne\delta$ , or $k\ne\gamma$ and $l\ne\delta$.

    Now we can try to combine these conditions. First, consider combining the conditions:
\begin{equation}\begin{cases}
k=\delta\text{ and } l=\gamma \text{ and } \alpha_1\alpha_2=1,\\
k=\gamma\text{ and } l=\delta\text{ and } \alpha_1=\alpha_2
\end{cases}\implies k=l=\gamma=\delta\text{ and }\alpha_1=\alpha_2=\pm1.\end{equation}
However, we know $k\ne l$ and $\gamma\ne\delta$, so these conditions cannot apply simultaneously. Next, we try the combination:
\begin{equation}\begin{cases}
k=\delta\text{ and } l=\gamma \text{ and } \alpha_1\alpha_2=1,\\
k\ne\gamma\text{ and } l\ne\delta
\end{cases}\implies k=\delta\ne l=\gamma \text{ and } \alpha_1\alpha_2=1,\end{equation}
which is possible and corresponds to $T_1\propto T_2$ (Condition \ref{condition: second}). Similarly:
\begin{equation}\begin{cases}
k\ne\delta\text{ and } l\ne\gamma,\\
k=\gamma\text{ and } l=\delta\text{ and } \alpha_1=\alpha_2
\end{cases}\implies k=\gamma\ne l=\delta\text{ and } \alpha_1=\alpha_2,\end{equation}
is possible and also corresponds to $T_1\propto T_2$ (Condition \ref{condition: second}). Finally, consider the combination:
\begin{equation}\begin{cases}
k\ne\delta\text{ and } l\ne\gamma,\\
k\ne\gamma\text{ and } l\ne\delta
\end{cases}\implies \left\{k,l,\gamma,\delta\right\} \text{ are all distinct,}\end{equation}
which is Condition \ref{condition: third}; where we have used the $k\ne l$ and $\gamma\ne\delta$.
\end{proof}
\begin{corollary}[Qubit Excitation Commutation]
    Now consider qubit excitation generators defined as follows:
    \begin{equation}
        T=G+aG^\dagger\quad\text{where } G\in\mathcal Q_l.
    \end{equation}
    If two such generators have equivalent support, then either Condition \ref{condition: second} or \ref{condition: third} must hold. Thus, two-qubit excitation generators commute iff they have disjoint or equivalent support.
\end{corollary}

\subsection{Fermionic Excitations}
\label{appendix: Fermionic pool}

A Fermionic excitation generator generalizes to an operator of the form:
\begin{definition}[Generalized Fermionic Excitation]
    \begin{equation}
        T=K+bK^\dagger\quad\text{where }K\coloneqq\prod_{i\in\vec i}a_i^{\left(x_i\right)},\quad a^{\left(x_i\right)}_i\coloneqq\begin{cases}
            a^\dagger_i,&x_i=0\\
            a_i,&x_i=1
        \end{cases},
    \end{equation}
    and $x$ is a bit-string and $\vec i$ is a tuple of unique orbital indices.
\end{definition}

In the Jordan-Wigner encoding $a_i=Q_i\otimes\mathcal Z_i$ where $\mathcal Z_i=\bigotimes_{j=1}^{i-1}Z_j$ is a Pauli string of $Z$ operators. Now we will define the commutation product and use it to express the generalized Fermionic excitation generators as the tensor product of a singleton matrix excitation generator and a Pauli string of $Z$ operators:

\begin{definition}[Commutation Product]
    Let $\left(\bullet,\bullet\right):\mathcal H^{\times 2}\times \mathcal H^{\times 2}\to\left\{0, \pm1\right\}$ be a mapping from a pair of operators on the Hilbert space $\mathcal H$ to $\left\{0, \pm1\right\}$ such that
    \begin{equation}
        \left(A,B\right)\coloneqq\begin{cases}
            1,&\left[A,B\right]=0,\\
            -1,&AB+BA=0,\\
            0,&\text{otherwise},
        \end{cases}
    \end{equation}
    which satisfies:
    \begin{enumerate}
        \item $\left(A,A\right)=1$,
        \item $\left(A,B\right)=\left(B,A\right)$,
        \item $\left(A,BC\right)=\left(A,B\right)\left(A,C\right)=\left(A,CB\right)$,\label{property:ordering}
        \item $AB=\left(A,B\right)BA$,
        \item $AB=\left(A,B\right)BA$ if $\left(A,B\right)=\pm1$.
    \end{enumerate}
    Using \ref{property:ordering} let
    \begin{equation}
        \left(A,\prod_iB_i\right)\equiv\left(A,\left\{B_i\right\}_i\right).
    \end{equation}
\end{definition}

Note that $Q$ and $Q^\dagger$ anti-commute with $Z$ (i.e. $\left(Q^{\left(\dagger\right)},Z\right)=-1$).

\begin{lemma}[Generalized Fermionic Excitation Tensor Factorisation]Generalized Fermionic excitation generators can be expressed as follows:
    \begin{equation}
        T=\pm T_s\otimes\tilde{\mathcal Z},
    \end{equation}
    where $T_s$ is a singleton matrix excitation generator and $\tilde{\mathcal Z}\in\left\{\mathds{1}, Z\right\}^{\otimes m}$ for some $m\in\mathbb N$.
\end{lemma}
\begin{proof}
    \begin{align}
        K&\coloneqq\prod_{i\in\vec i}a_i^{\left(x_i\right)}\\
        &=\prod_{i\in\vec i}\left[Q^{\left(x_i\right)}_i\mathcal Z_i\right],\quad\text{where }Q^{\left(x_i\right)}_i\coloneqq\begin{cases}
            Q^\dagger_i,&x_i=0\\
            Q_i,&x_i=1
        \end{cases}\\
        &=\left[\bigotimes_{i\in\vec i}Q^{\left(x_i\right)}_i\right]\cdot\left[\prod_{n=1}^{\dim\vec i}\left(\left\{Q_{i_m}^{\left(x_{i_m}\right)}\right\}_{m=1}^{n-1}, \mathcal Z_{i_n}\right)\mathcal Z_{i_n}\right]\\
        &=\left[\prod_{n=1}^{\dim\vec i}\left(\left\{Q_{i_m}^{\left(x_{i_m}\right)}\right\}_{m=1}^{n-1}, \mathcal Z_{i_n}\right)\right]\cdot\left[\bigotimes_{i\in\vec i}Q^{\left(x_i\right)}_i\right]\cdot\left[\prod_{n=1}^{\dim\vec i}\mathcal Z_{i_n}\right].
    \end{align}

    Now note that
    \begin{align}
        &QZ=-Q;&
        Q^\dagger Z=Q^\dagger.&
    \end{align}
    Thus, the mutual support of $\left[\bigotimes_{i\in\vec i}Q^{\left(x_i\right)}_i\right]$ and $\left[\prod_{n=1}^{\dim\vec i}\mathcal Z_{i_n}\right]$ can be removed by replacing the $Z$ operators within the mutual support with $\mathds{1}$ and obtaining a factor of $1$ or $-1$. In this process the factor $\left[\bigotimes_{i\in\vec i}Q^{\left(x_i\right)}_i\right]$ is unchanged and $\left[\prod_{n=1}^{\dim\vec i}\mathcal Z_{i_n}\right]\mapsto\tilde{\mathcal Z}$. Additionally,
    \begin{equation}
        \left[\prod_{n=1}^{\dim\vec i}\left(\left\{Q_{i_m}^{\left(x_{i_m}\right)}\right\}_{m=1}^{n-1}, \mathcal Z_{i_n}\right)\right]\in\left\{1,-1\right\},
    \end{equation}
    which can be combined with the factor of $1$ or $-1$ from earlier to produce a factor $\pm1$. Therefore we find:
    \begin{align}
        &K\equiv\pm\left[\bigotimes_{i\in\vec i}Q^{\left(x_i\right)}_i\right]\otimes\tilde{\mathcal Z};&
        K^\dagger\equiv\pm\left[\bigotimes_{i\in\vec i}Q^{\left(x_i\right)}_i\right]^\dagger\otimes\tilde{\mathcal Z}.&
    \end{align}
    Finally, we can substitute this into the form of $T$:
    \begin{equation}
        T=\pm\left(\left[\bigotimes_{i\in\vec i}Q^{\left(x_i\right)}_i\right]+b\left[\bigotimes_{i\in\vec i}Q^{\left(x_i\right)}_i\right]^\dagger\right)\otimes\tilde{\mathcal Z},
    \end{equation}
    and note that the tensor factor in brackets is a singleton matrix excitation generator.
\end{proof}

Moving forward with this form of generalized fermionic excitation generators allows us to use \cref{theorem: singleton matrix excitation commutation} to show the following commutation relations hold for generalized fermionic excitations:

\begin{theorem}[Generalized Fermionic Excitation Commutation]\label{theorem: generalized fermionic excitation commutation}
    Two generalized fermionic excitations $T_1$ and $T_2$ commute iff any of the following conditions are satisfied:
    \begin{enumerate}
        \item $T_{s1}$ and $T_{s2}$ have disjoint support and $\left|\support T_{s1}\cap\support \tilde{\mathcal Z}_2\right|+\left|\support T_{s2}\cap\support \tilde{\mathcal Z}_1\right|$ is even,\label{condition: fermionic first}
        \item $T_{s1}\propto T_{s2}$,\label{condition: fermionic second}
        \item $T_{s1}$ and $T_{s2}$ have equivalent support and can be written as
        \begin{align}
            &T_{s1}\equiv M_{rs}+b_1M_{sr};
            &T_{s2}\equiv M_{tu}+b_2M_{ut}&.
        \end{align}
        where $r$, $s$, $t$ and $u$ are all distinct.\label{condition: fermionic third}
    \end{enumerate}
\end{theorem}
\begin{proof}
    Consider the product
    \begin{align}
        T_1T_2&=c_1T_{s1}\otimes\tilde{\mathcal Z}_1\cdot c_2T_{s2}\otimes\tilde{\mathcal Z}_2\\
        &=c_1c_2T_{s1}\otimes\tilde{\mathcal Z}_1\cdot T_{s2}\otimes\tilde{\mathcal Z}_2\\
        &=c_1c_2\left(\tilde{\mathcal Z}_1, T_{s2}\right)T_{s1}T_{s2}\tilde{\mathcal Z}_1\tilde{\mathcal Z}_2,\quad\text{where tensor products with the identity are implied}\\
        &=c_1c_2\left(\tilde{\mathcal Z}_1, T_{s2}\right)T_{s2}T_{s1}\tilde{\mathcal Z}_1\tilde{\mathcal Z}_2+c_1c_2\left(\tilde{\mathcal Z}_1, T_{s2}\right)\left[T_{s1},T_{s2}\right]\tilde{\mathcal Z}_1\tilde{\mathcal Z}_2\\
        &=\left(\tilde{\mathcal Z}_1, T_{s2}\right)\left(\tilde{\mathcal Z}_2, T_{s1}\right)c_2T_{s2}\otimes\tilde{\mathcal Z}_2\cdot c_1T_{s1}\otimes\tilde{\mathcal Z}_1+c_1c_2\left(\tilde{\mathcal Z}_1, T_{s2}\right)\left[T_{s1},T_{s2}\right]\tilde{\mathcal Z}_1\tilde{\mathcal Z}_2\\
        &=\left(\tilde{\mathcal Z}_1, T_{s2}\right)\left(\tilde{\mathcal Z}_2, T_{s1}\right)T_2T_1+c_1c_2\left(\tilde{\mathcal Z}_1, T_{s2}\right)\left[T_{s1},T_{s2}\right]\tilde{\mathcal Z}_1\tilde{\mathcal Z}_2,
    \end{align}
    where $c_1,c_2\in\left\{1,-1\right\}$ and we have used the fact that $\left(\tilde{\mathcal Z}_m, T_{sn}\right)\in\left\{1,-1\right\}$ for all $m, n\in\left\{1,2\right\}$.

    Therefore, we can write the commutator as
    \begin{equation}
        \left[T_1,T_2\right]=\left[\left(\tilde{\mathcal Z}_1, T_{s2}\right)\left(\tilde{\mathcal Z}_2, T_{s1}\right)-1\right]T_2T_1+c_1c_2\left(\tilde{\mathcal Z}_1, T_{s2}\right)\left[T_{s1},T_{s2}\right]\tilde{\mathcal Z}_1\tilde{\mathcal Z}_2.
    \end{equation}

    The entries of $T_{sn}$ are drawn from the set $\left\{0,1\right\}$. Thus, if $T_{s1}$ and $T_{s2}$ do not commute, then $T_{s2}T_{s1}$ and $\left[T_{s1},T_{s2}\right]$ are linearly independent, so for $\left[T_1,T_2\right]=0$ we require both both terms to vanish:
    \begin{numcases}{}
        \left[\left(\tilde{\mathcal Z}_1, T_{s2}\right)\left(\tilde{\mathcal Z}_2, T_{s1}\right)-1\right]T_2T_1=0,&and\\
        c_1c_2\left(\tilde{\mathcal Z}_1, T_{s2}\right)\left[T_{s1},T_{s2}\right]\tilde{\mathcal Z}_1\tilde{\mathcal Z}_2=0.
    \end{numcases}
    Which simplifies to
    \begin{numcases}{}
        \left[\left(\tilde{\mathcal Z}_1, T_{s2}\right)\left(\tilde{\mathcal Z}_2, T_{s1}\right)-1\right]=0\quad\text{or}\quad T_2T_1=0,&and\label{condition: even Zs}\\
        \left[T_{s1},T_{s2}\right]=0.\label{condition: M matrices commute}
    \end{numcases}
    
    Thus, Condition \ref{condition: M matrices commute} requires we satisfy at least one of the three conditions from \cref{theorem: singleton matrix excitation commutation}. Therefore, each condition in \cref{theorem: generalized fermionic excitation commutation} corresponds to the condition with the same number in \cref{theorem: singleton matrix excitation commutation}. However, Condition \ref{condition: even Zs} requires us to strengthen the conditions in \cref{theorem: generalized fermionic excitation commutation}.

    First, consider when $T_2T_1=0$. The $\tilde{\mathcal Z_n}$ factors will only produce phase factors and so $T_2T_1=0\iff T_{s2}T_{s1}=0$. Using the identities
    \begin{align}
        &T_{s1}\equiv M_{rs}+b_1M_{sr},
        &T_{s2}\equiv M_{tu}+b_2M_{ut}&,
    \end{align}
    where $r\ne s$ and $t\ne u$ we find
    \begin{equation}
        0=T_{s2}T_{s1}=\delta_{ur}M_{ts}+b_2\delta_{tr}M_{us}+b_1\delta_{us}M_{tr}+b_2b_1\delta_{ts}M_{ur}.
    \end{equation}
    We note $M_{ts}$, $M_{us}$, $M_{tr}$, and $M_{ur}$ are linearly independent as $r\ne s$ and $t\ne u$. Thus, $T_2T_1=0$ iff:
    \begin{numcases}{}
        u\ne r,\\
        t\ne r,\\
        u\ne s,\\
        t\ne s.
    \end{numcases}
    Combining these conditions with $r\ne s$ and $t\ne u$ implies $T_2T_1=0$ iff $r$, $s$, $t$ and $u$ are all distinct. We already obtained this condition from Condition \ref{condition: M matrices commute}, giving us Condition \ref{condition: fermionic third}.

    Finally, as $Q$ and $Q^\dagger$ anti-commute with $Z$, then $\left[\left(\tilde{\mathcal Z}_1, T_{s2}\right)\left(\tilde{\mathcal Z}_2, T_{s1}\right)-1\right]=0$ occurs iff $\left|\support T_{s1}\cap\support \tilde{\mathcal Z}_2\right|+\left|\support T_{s2}\cap\support \tilde{\mathcal Z}_1\right|$ is even. As $T_{s1}\propto T_{s2}$ implies $\left|\support T_{s1}\cap\support \tilde{\mathcal Z}_2\right|+\left|\support T_{s2}\cap\support \tilde{\mathcal Z}_1\right|$ is even we need only modify Condition \ref{condition: fermionic first}.
\end{proof}
\begin{corollary}[Fermionic Excitation Commutation]
    Now consider fermionic excitation generators defined as follows:
    \begin{equation}
        T=K-K^\dagger\quad\text{where }K\coloneqq\prod_{i\in\vec i}a_i^{\left(x_i\right)},\quad a^{\left(x_i\right)}_i\coloneqq\begin{cases}
            a^\dagger_i,&x_i=0\\
            a_i,&x_i=1
        \end{cases},
    \end{equation}
    and $x$ is a bit-string and $\vec i$ is a tuple of unique orbital indices.
    
    If, for two such generators, the sets of orbitals acted upon are equivalent, then either conditions \ref{condition: fermionic second} or \ref{condition: fermionic third} must hold. Thus, one can simplify the conditions for commutation to:
    \begin{enumerate}
        \item The sets of orbitals acted upon are disjoint and $\left|\support T_{s1}\cap\support \tilde{\mathcal Z}_2\right|+\left|\support T_{s2}\cap\support \tilde{\mathcal Z}_1\right|$ is even, \label{condition: corollary Fermionic excitation commutation condition 1}
        \item The sets of orbitals acted upon are equivalent.
    \end{enumerate}
\end{corollary}
\begin{corollary}[Electron Conserving Fermionic Excitation Commutation]
    Fermionic Excitations that conserve electron number have even $\dim\vec i$, so if the sets of orbitals acted upon by two different fermionic excitations are disjoint, then $\left|\support T_{s1}\cap\support \tilde{\mathcal Z}_2\right|+\left|\support T_{s2}\cap\support \tilde{\mathcal Z}_1\right|$ must be even. Hence, we can simplify Condition \ref{condition: corollary Fermionic excitation commutation condition 1} to: the set of orbitals acted upon are disjoint.
\end{corollary}
\section{Generalized Non-Commuting Set Cardinalities}
\label{appendix:non-commuting set cardinalities}
In this appendix, we derive the cardinalities of the generalized non-commuting sets for both the QEB pool (\cref{appendix: size QEB}) and the qubit pool (\cref{appendix: size qubit}) using operator and support commutation. Further, we will derive their asymptotic scalings.

To proceed, we will define
\begin{equation}
    \pool_q\left(N\right)\coloneqq\left\{A\in\pool\left(N\right)\colon \left|\support{A}=q\right|\right\},
\end{equation}
as the subpool of elements supported only by $q$ qubits.

\subsection{Generalized QEB Pool}
\label{appendix: size QEB}

The cardinality of the support non-commuting set of $A\in\pool_q\left(N\right)$ is independent of $A$ and is given by
\begin{align}
    C^S_q\left(N\right)\equiv\left|\nset_{\textrm{S}}\left[\pool\left(N\right),A\in\pool_q\left(N\right)\right]\right|&=\sum_{\substack{p\in\mathcal V\\\colon p\le N}}\sum_{a=1}^{\min\left\{p,q\right\}}\begin{pmatrix}p-1\\\frac{1}{2}p\end{pmatrix}\begin{pmatrix}q\\a\end{pmatrix}\begin{pmatrix}N-q\\p-a\end{pmatrix}-1\\
    &=\bigTheta{N^{\max\mathcal V-1}}.
\end{align}
Here we sum over the possible numbers of overlapping qubits, $a$. The third choice is the number of sets of $p-a$ distinct qubits of $N-q$, the second is the number of sets of $q$ distinct qubits of $a$, and the first is the number of permutations of these qubits that give distinct qubit excitations. The $-1$ removes the operator inducing the support non-commuting set.

Similarly, we can use the operator commutation rules for qubit excitations to show that
\begin{align}
    C^O_q\left(N\right)\equiv\left|\nset_{\textrm{O}}\left[\pool\left(N\right),A\in\pool_q\left(N\right)\right]\right|&=\sum_{\substack{p\in\mathcal V\\\colon p\le N}}\sum_{a=1}^{\alpha\left(p,q\right)}\begin{pmatrix}p-1\\\frac{1}{2}p\end{pmatrix}\begin{pmatrix}q\\a\end{pmatrix}\begin{pmatrix}N-q\\p-a\end{pmatrix}\\&=\bigTheta{N^{\max\mathcal V-1}},
\end{align}
where
\begin{equation}
    \alpha\left(p,q\right)\coloneqq\begin{cases}
    \min\left\{p,q\right\}, & p\ne q\\
    q-1, & p=q.
    \end{cases}
\end{equation}
The choices are the same as for the cardinality of the support non-commuting set. However, the summation excludes equivalent support.
We note that the operator and support non-commuting sets are almost equivalent:
\begin{equation}
    C^S_q\left(N\right)-C^O_q\left(N\right)=\begin{pmatrix}q-1\\\frac{1}{2}q\end{pmatrix}-1=\Theta\left(1\right).
\end{equation}
Because $\nset_{\textrm{O}}\left(\pool,A\right)\subseteq\nset_{\textrm{S}}\left(\pool,A\right)$, then $C^S_q\left(N\right)-C^O_q\left(N\right)$ is number of elements by which the sets differ. For single excitations, this difference vanishes, and for double excitations is 
two.

Note that because the support and operator non-commuting sets differ by $\Theta\left(1\right)$ elements, then an ansatz circuit of mutually 
operator commuting ansatz elements will be, at most, exactly 
$\displaystyle\max_{q\in\mathcal 
V}\begin{pmatrix}q-1\\\frac{1}{2}q\end{pmatrix}=\Theta\left(1\right)$ times 
deeper than an ansatz circuit of mutually support-commuting ansatz elements.

\subsection{Generalized qubit Pool}
\label{appendix: size qubit}

Similarly, the support non-commuting sets for the qubit pool have cardinalities
\begin{align}
    C^S_q\left(N\right)\equiv\left|\nset_{\textrm{S}}\left[\pool\left(N\right),A\in\pool_q\left(N\right)\right]\right|&=\sum_{\substack{p\in\mathcal V\\\colon p\le N}}\sum_{a=1}^{\min\left\{p,q\right\}}2^{p-1}\begin{pmatrix}q\\a\end{pmatrix}\begin{pmatrix}N-q\\p-a\end{pmatrix}-1\\&=\bigTheta{N^{\max\mathcal V-1}}.
\end{align}

Further, half of the ansatz elements with partial support with $A\in\pool_q\left(N\right)$ operator commute with $A$. Additionally, those with equivalent support with $A$ operator commute with $A$. Thus,
\begin{align}
    C^O_q\left(N\right)\equiv\left|\nset_{\textrm{O}}\left[\pool\left(N\right),A\in\pool_q\left(N\right)\right]\right|&=\sum_{\substack{p\in\mathcal V\\\colon p\le N}}\sum_{a=1}^{\alpha\left(p,q\right)}2^{p-2}\begin{pmatrix}q\\a\end{pmatrix}\begin{pmatrix}N-q\\p-a\end{pmatrix}\\&=\bigTheta{N^{\max\mathcal V-1}}.
\end{align}
Therefore, the support and operator non-commuting sets differ asymptotically in size by a perfector of two:
\begin{equation}
    C^S_q\left(N\right)-C^O_q\left(N\right)=\frac{1}{2}C^S_q\left(N\right)+2^{q-1}
    -1=\bigTheta{N^{\max\mathcal V-1}}.
\end{equation}
This is the cardinality of the set of ansatz elements that do not support commute but do operator commute with $A\in\pool_q\left(N\right)$. Thus, an ansatz circuit of mutually operator-commuting ansatz elements can be, at most, $\Omega\left(N^{\max\mathcal V-1}\right)$ times deeper than an ansatz circuit of mutually support-commuting ansatz elements. But as we can construct any ansatz circuit of mutually operator-commuting ansatz elements by splitting the pool into $\bigTheta{N^{\max\mathcal V-1}}$ disjoint subsets of mutually support non-commuting ansatz elements and using each subset to construct a layer of the ansatz circuit. Then we can remove all the ansatz elements that do not appear in the ansatz circuit we originally wished to construct---this ansatz circuit is $\mathcal O\left(N^{\max\mathcal V-1}\right)$ layers deep. Thus, an ansatz circuit of mutually operator-commuting ansatz elements is at most $\bigTheta{N^{\max\mathcal V-1}}$ times deeper than an ansatz circuit of mutually support-commuting ansatz elements.

\section{Statistical Analysis of Sequences of Non-Commuting Sets}
\label{appendix: sequences of support sets}

In this appendix, we derive expressions for the expected cardinality of the subpool searched to find a local minimum by subpool exploration. This will allow us to derive the expected number of loss function evaluations per element appended for Explore- and Dynamic-ADAPT-VQE. First, we will consider how the subpools at each step of subpool exploration vary for support commutation in \cref{appendix: support commutation sequences}. Next, we will extend this result to operator commutation for the QEB pool only in \cref{appendix: operator commutation sequences}. Using these results in \cref{appendix: statistical support analysis} and \cref{appendix: explore}, we will derive an upper bound for the expected number of loss function evaluations per element for Explore-ADAPT-VQE in terms of the number of local minima (see property \ref{property: local minimum}) in the pool
\begin{equation}
    M\coloneqq\left|\left\{A\in\pool\colon L\left(A\right)=\smashoperator[r]{\min_{B\in\nset_{\textrm{G}}\left(\pool,A\right)}} L\left(B\right)\right\}\right|.
\end{equation}
Finally, we will use this upper bound to also upper bound the expected number of loss function evaluations per element for Dynamic-ADAPT-VQE in \cref{appendix: dynamic}.

\subsection{Support-Based-Commutation Sequences}
\label{appendix: support commutation sequences}

Within this subsection, we do not specify a pool as support commutation allows for a pool agnostic analysis. First, we will consider the case in which $\mathcal S_0$ is a singleton. Next, we will use this to upper bound the general case.

If $\mathcal S_0$ is a singleton, using the recursion relation in \cref{eq: subpool update}, then $\mathcal S_{m+1}$ has no operators that support any of $\left\{A_i\right\}_{i=0}^{m-1}$. Thus, at step $m$ the remaining pool $\pool\backslash\mathcal S_{\le m}$ will only have support on $N_m\coloneqq N-\sum_{i=0}^{m-1}\nu_m$ qubits, where
\begin{equation}
    \nu_m\coloneqq\begin{cases}
        \left|\support A_m\cap\left(\cup_{l=0}^{m-1}\support A_l\right)\right|\in\left[0,q_m-1\right],&m\ge1\\
        q_0,&m=0
    \end{cases}
\end{equation}
is the cardinality of the mutual support of the previous ansatz elements and $A_m$, with $q_m\coloneqq\left|\support A_m\right|$. Therefore,
\begin{align}
    \left|\mathcal S_{m+1}\right|\equiv\left|\nset_{\textrm{S}}\left[\pool\backslash\mathcal S_{\le m},A_m\in\pool_{q_m}\left(N\right)\right]\right|&=C^S_{{q_m}-\nu_m}\left(N_m\right)+\min\left\{1,\nu_m\right\}\\&=\Theta\left(N_m^{\max\mathcal V-1}\right).
\end{align}

If $\left|\mathcal S_0\right|\ge 1$, then we will have pre-evaluated some ansatz 
elements in the pool, and so the equivalence is demoted to an inequality:
\begin{align}\label{eq:cardinality bound support}
    \left|\mathcal S_{m+1}\right|\equiv\left|\nset_{\textrm{S}}\left[\pool\backslash\mathcal S_{\le m},A_m\in\pool_{q_m}\left(N\right)\right]\right|&\le C^S_{q_m-\nu_m}\left(N_m\right)+\min\left\{1,\nu_m\right\}\\&=\mathcal O\left(N_m^{\max\mathcal V-1}\right)
\end{align}

\subsection{Operator-Based-Commutation Sequences}
\label{appendix: operator commutation sequences}

Here we consider the more complex case of operator commutation. We can no longer remain pool agnostic; we will only consider the QEB pool for ease. First, we note a property of the QEB pool that allows us to establish an approximate equivalence between operator and support commutation sequences for the QEB pool. Using this, we will modify the results for support commutation sequences to QEB operator-commutation sequences.

First, note that the QEB pool has the following property. Consider an element $A$ in the operator non-commuting set of $B\in\pool$. All the ansatz elements that do not support commute but do operator commute with $A$ form a subset of the operator non-commuting set of $B\in\pool$. That is
\begin{equation}
    A\in \nset_{\textrm{O}}\left[\pool,B\right]\iff\nset_{\textrm{S}}\left[\pool,A\right]\backslash \nset_{\textrm{O}}\left[\pool,A\right]\subset\nset_{\textrm{O}}\left[\pool,B\right].
\end{equation}
Thus, all ansatz elements that support and operator commute with $A_m$ will be in the set $\mathcal S_{\le m}$ for $m\ge1$ if $A_m\in\mathcal S_{\le m}$. The condition of $A_m\in\mathcal S_{\le m}$ is true in subpool exploration. If, however, the ansatz elements that support and operator commute with $A_0$ are in $\mathcal S_0$, then this is extended to $m\ge 0$. Thus:
\begin{align}
    \left|\mathcal S_{m+1}\right|&\equiv\left|\nset_{\textrm{O}}\left[\pool\backslash\mathcal S_{\le m},A_m\in\pool_{q_m}\left(N\right)\cap\left(\mathcal S_{\le m}\right)\right]\right|\\
&\le \begin{cases}
C^O_{q_m}\left(N_0\right), &m=0\\
C^S_{q_m-\nu_m}\left(N_1\right)+\begin{pmatrix}\operatorname{Supp}A_0-1\\\frac{1}{2}\operatorname{Supp}A_0\end{pmatrix}, &\text{for $m=1$ if }\nset_{\textrm{S}}\left[\pool,A_0\right]\backslash \nset_{\textrm{O}}\left[\pool,A_0\right]\not\subseteq\mathcal S_0\\
C^S_{q_m-\nu_m}\left(N_m\right)+1, &\text{otherwise},
\end{cases}\label{eq:cardinality bound commutation}
\end{align}
with equality if $\mathcal S_0\equiv\left\{A_0\right\}$ or 
$\left\{A_0\right\}\cup\nset_{\textrm{S}}\left[\pool,A_0\right]\backslash \mathcal 
N_O\left[\pool,A_0\right]$. This means that after the first two steps, subpool exploration is the same, when using the QEB pool, independent of whether support and operator commutativity is used.

\subsection{Statistical Analysis of Support-Based-Commutation Sequences}
\label{appendix: statistical support analysis}

This subsection outlines an upper bound for the probability of having terminated during subpool exploration after $m$ steps under a reasonable assumption: subpool exploration is more effective than random sampling. With this framework, we will bound the expected number of steps and elements searched for Explore- and Dynamic-ADAPT-VQE in the subsequent subsections.

Suppose we are seeking a local minimum, as in subpool exploration, but instead, we generate the sequence
\begin{equation}
    \mathcal S_{m+1}=\nset_{\textrm{G}}\left(\pool\backslash\mathcal S_{\le m},A_m\right)\subseteq\nset_{\textrm{G}}\left(\pool,A_m\right)\quad\forall m\ge0,
\end{equation}
using $A_m$ drawn from a distribution independent of  $\left\{A_l\right\}_{l=0}^{m-1}$. Suppose further that the set $\mathcal M$ of $M$ local minima is distributed randomly in the pool. Let $\tilde m$ be the random variable for the number of steps taken such that the minimum ansatz element in $\mathcal S_{\le\tilde m-1}$ is a local minimum. Let the probability of $\tilde m\le m$ given a pool $\pool$ over $N$ qubits and the sequences $\vec q$ and $\vec\nu$, be denoted by $\mathbb P\left(\tilde m\le m\middle|N,\vec q,\vec\nu\right)$. As $\mathcal M\cap\mathcal S_{\le m-1}=\emptyset\implies\tilde m>m$ then $\mathbb P\left(\tilde m> m\middle|N,\vec q,\vec\nu\right)\le\mathbb P\left(\mathcal M\cap\mathcal S_{\le m-1}=\emptyset\middle|N,\vec q,\vec\nu\right)$. That is $\mathbb P\left(\tilde m\le m\middle|N,\vec q,\vec\nu\right)\ge 1-\mathbb P\left(\mathcal M\cap\mathcal S_{\le m-1}=\emptyset\middle|N,\vec q,\vec\nu\right)$.

We can calculate $\mathbb P\left(\mathcal M\cap\mathcal S_{\le m-1}=\emptyset\middle|N,\vec q,\vec\nu\right)$ by fixing the subset $\mathcal S_{\le m-1}\subset\pool$ and then sequentially placing the local minima randomly in the pool. The probability that the first local minimum is not placed in $\mathcal S_{\le m-1}$ is $\left(\left|\pool\right|-\left|\mathcal S_{\le m-1}\right|\right)/\left|\pool\right|$. However, note that a local minimum cannot lie within the generalized non-commuting set of another local minimum. Thus, we must remove the ansatz elements of the generalized non-commuting set from both the $\pool$ and $\mathcal S_{\le m-1}$ before the next step. Thus, the probability of not placing the $i$th local minimum in $\mathcal S_{\le m-1}$ conditional on all previous local minima not being placed in $\mathcal S_{\le m-1}$ and the cardinality of the support of the $j$th local minimum being $\mu_j$ is given by
\begin{equation}
    \frac{\left|\pool\left(N_{m-1}-\sum_{j=1}^{i-1}\mu_j\right)\right|}{\left|\pool\left(N-\sum_{j=1}^{i-1}\mu_j\right)\right|},
\end{equation}
for support commutation. Here we have assumed our initial subpool is a singleton. If the initial subpool is not a singleton, then we will have pre-evaluated some ansatz elements in future subpools, and so the numerator will decrease. Therefore, this probability constitutes an upper bound. Further, we can bound this probability as follows:

\begin{lemma}
    There exists a finite constant $V$, such that if $N_{m-1}-\sum_{j=1}^{i-1}\mu_j\ge V$, then
    \begin{equation}
        \frac{\left|\pool\left(N_m-\sum_{j=1}^{i-1}\mu_j\right)\right|}{\left|\pool\left(N-\sum_{j=1}^{i-1}\mu_j\right)\right|}\le\frac{\left|\pool\left(N_m\right)\right|}{\left|\pool\left(N\right)\right|}.
    \end{equation}
\end{lemma}
\begin{proof}
    Consider the ratio $r\coloneqq\frac{f\left(x-a\right)}{f\left(x\right)}$ where $a$ is non-negative. Now consider when the derivative is non-negative:
    \begin{align}
        \dv{r}{x}\equiv\frac{f'\left(x-a\right)}{f\left(x\right)}-\frac{f\left(x-a\right)f'\left(x\right)}{f^2\left(x\right)}\ge0\\
        \iff\dv{x}\ln\left[f\left(x-a\right)\right]\ge\dv{x}\ln\left[f\left(x\right)\right].
    \end{align}
    This is true iff $f\left(x\right)$ is logarithmically concave. Note by \cref{lemma: logarithmically concave} there exists some finite constant $V$ such that $\left|\pool\left(N\right)\right|$ is logarithmically concave.
\end{proof}

Putting these results together, we can bound $\mathbb P\left(\mathcal M\cap\mathcal S_{\le m}=\emptyset\middle|N,\vec q,\vec\nu,\vec\mu\right)$ for $m$ small enough to satisfy $N_m-\sum_{j=1}^{M-1}\mu_j\ge V$ and then use a looser bound for the remaining terms:
\begin{align}
    \mathbb P\left(\mathcal M\cap\mathcal S_{\le m}=\emptyset\middle|N,\vec q,\vec\nu,\vec\mu\right)
    &\le\prod_{i=1}^M\frac{\left|\pool\left(N_m-\sum_{j=1}^{i-1}\mu_j\right)\right|}{\left|\pool\left(N-\sum_{j=1}^{i-1}\mu_j\right)\right|}\\
    &\le\begin{cases}\displaystyle \prod_{i=1}^M\frac{\left|\pool\left(N_m\right)\right|}{\left|\pool\left(N\right)\right|},&m\le m'_{\max}\\
    \displaystyle \prod_{i=1}^M\frac{\left|\pool\left(N_{m'_{\max}}\right)\right|}{\left|\pool\left(N\right)\right|},&m>m'_{\max}\end{cases}\\
    &=\begin{cases}\displaystyle\left[\frac{\left|\pool\left(N_m\right)\right|}{\left|\pool\left(N\right)\right|}\right]^M,&m\le m'_{\max}\\
    \displaystyle\left[\frac{\left|\pool\left(N_{m'_{\max}}\right)\right|}{\left|\pool\left(N\right)\right|}\right]^M,&m>m'_{\max},\end{cases}
\end{align}
where $m'_{\max}$ be the largest $m$ that satisfies $N_m-\sum_{j=1}^{M-1}\mu_j\ge V$. Similarly, $\mmax$ is the largest $m$ that satisfies $N_m-\norm{\vec\mu}_1\ge0$. Thus, there will at most be $V-\min\mathcal V\ge\mmax-m'_{\max}$ terms for which $m$ is too large to satisfy $N_m-\sum_{j=1}^{M-1}\mu_j\ge V$ and we will find these terms will asymptotically yield second-order contributions to our quantities of interest.

As the final upper bound is independent of $\vec\mu$ and $\vec q$ then we can marginalize over $\vec\mu$ and $\vec q$ to get
\begin{equation}
    \mathbb P\left(\mathcal M\cap\mathcal S_{\le m}=\emptyset\middle|N,\vec\nu\right)\le\begin{cases}\displaystyle\left[\frac{\left|\pool\left(N_m\right)\right|}{\left|\pool\left(N\right)\right|}\right]^M,&m\le m'_{\max}\\
    \displaystyle\left[\frac{\left|\pool\left(N_{m'_{\max}}\right)\right|}{\left|\pool\left(N\right)\right|}\right]^M,&m>m'_{\max},\end{cases}
\end{equation}
and finally, lower bound the probability of interest:
\begin{equation}\label{eq: probability lower bound}
    \mathbb P\left(\tilde m\le m\middle|N,\vec\nu\right)\ge \begin{cases}\displaystyle 1-\left[\frac{\left|\pool\left(N_m\right)\right|}{\left|\pool\left(N\right)\right|}\right]^M,&m\le m'_{\max}\\
        \displaystyle 1-\left[\frac{\left|\pool\left(N_{m'_{\max}}\right)\right|}{\left|\pool\left(N\right)\right|}\right]^M,&m>m'_{\max}.\end{cases}
\end{equation}
Finally, we note the asymptotic scaling:
\begin{equation}\label{eq: probability asymtotics}
    \left[\frac{\left|\pool\left(N_m\right)\right|}{\left|\pool\left(N'\right)\right|}\right]^M\sim\left[\frac{N_m}{N}\right]^{M\max\mathcal V},
\end{equation}
where we note this asymptotic scaling does not necessarily hold for $m>m'_{\max}$, but we no longer have any terms of this form.

However, in our proposed pool exploration strategy, the $A_m$ are drawn from a dependent distribution. We take $A_m$ as the ansatz element with the minimum loss in $\mathcal S_{\le m}$. Assuming this dependent sequence is a better exploration strategy than independent random sampling then this strategy should drive the sequence to a local minima faster than the independent proposal strategy. Thus, one would expect the cumulative probability distribution of termination to shift to larger probabilities:
\begin{equation}
    \mathbb P\left(\tilde m\le m\middle|N,\vec\nu\right)\le\mathbb P\left(\hat m\le m\middle|N,\vec\nu\right),
\end{equation}
where $\hat m$ is the random variable for the number of steps taken to find a local minimum using this dependent sequence. That is, $\hat m$ is the random variable for the number of steps taken such that the minimum ansatz element in $\mathcal S_{\le\hat m-1}$ is a local minimum.

We can write the mean of a monotonically increasing function $f$ with respect to $\tilde m$ as
\begin{align}
    \mathbb E\left(f\left(\tilde m\right)\middle|N,\vec\nu\right)&\equiv\sum_{m=0}^{\mmax}f\left(m\right)\mathbb P\left(\tilde m=m\middle|N,\vec\nu\right)\label{eq: mean definition}\\
    &=f\left(0\right)+\sum_{m=1}^{\mmax}\left[f\left(m\right)-f\left(m-1\right)\right]\mathbb P\left(\tilde m\ge m\middle|N,\vec\nu\right)\label{eq: replace with cumulative}\\
    &=f\left(0\right)+\sum_{m=0}^{\mmax-1}\left[f\left(m+1\right)-f\left(m\right)\right]\left[1-\mathbb P\left(\tilde m\le m\middle|N,\vec\nu\right)\right]\\
    &=f\left(\mmax\right)-\sum_{m=0}^{\mmax-1}\left[f\left(m+1\right)-f\left(m\right)\right]\mathbb P\left(\tilde m\le m\middle|N,\vec\nu\right),
\end{align}
and similarly for $\hat m$. Using this, we can lower bound the expected values:
\begin{align}
    \mathbb E\left(\tilde m\middle|N,\vec\nu\right)&\ge\mathbb E\left(\hat m\middle|N,\vec\nu\right),\label{eq: mean inequality}\\
    \mathbb E\left(\left|\mathcal S_{\le\tilde m}\right|\middle|N,\vec\nu\right)&\ge\mathbb E\left(\left|\mathcal S_{\le\hat m}\right|\middle|N,\vec\nu\right).\label{eq: loss mean inequality}
\end{align}

Next, we proceed by considering specifically Explore-ADAPT-VQE.

\subsection{Statistical Analysis of Explore-ADAPT-VQE}
\label{appendix: explore}
Now we substitute for the bounding cumulative probability distribution to obtain the asymptotic scaling for the mean:
\begin{align}
    \mathbb E\left(\left|\mathcal S_{\le\tilde m}\right|\middle|N,\vec\nu\right)
    &=\sum_{m=0}^{\mmax}\left|\mathcal S_{\le m}\right|\mathbb P\left(\tilde m= m\middle|N,\vec\nu\right)\label{eq: sub in S}\\
    &=\left|\mathcal S_0\right|+\sum_{m=1}^{\mmax}\left|\mathcal S_{m}\right|\mathbb P\left(\tilde m\ge m\middle|N,\vec\nu\right)\label{eq: follow derivation}\\
    &=\left|\mathcal S_0\right|+\sum_{m=0}^{\mmax-1}\left|\mathcal S_{m+1}\right|\mathbb P\left(\tilde m> m\middle|N,\vec\nu\right)\\
    &\le\left|\mathcal S_0\right|+\sum_{m=0}^{m'_{\max}-1}\left|\mathcal S_{m+1}\right|\left[\frac{\left|\pool\left(N_m\right)\right|}{\left|\pool\left(N\right)\right|}\right]^M+\sum_{m=m'_{\max}}^{\mmax-1}\left|\mathcal S_{m+1}\right|\left[\frac{\left|\pool\left(N_{m'_{\max}}\right)\right|}{\left|\pool\left(N\right)\right|}\right]^M\label{eq: sub bound}\\
    &=\left|\mathcal S_0\right|+\Theta\left[\sum_{m=0}^{m'_{\max}-1}N_m^{\max\mathcal V-1}\left(\frac{N_m}{N}\right)^{M\max\mathcal V}+\sum_{m=m'_{\max}}^{\mmax-1}N_m^{\max\mathcal V-1}\left(\frac{N_{m'_{\max}}}{N}\right)^{M\max\mathcal V}\right].\label{eq: sub in asymtotics}
\end{align}
Where \cref{eq: sub in S,eq: follow derivation} follow from \cref{eq: mean definition,eq: replace with cumulative} and in \cref{eq: sub bound,eq: sub in asymtotics} we have substituted in the inequality given in \cref{eq: probability lower bound} and then the asymptotics in Eqs.~\eqref{eq: probability asymtotics} and \eqref{eq:f QEB cardinality} or \eqref{eq: qubit cardinality}. However $N_{m}=\bigO{1}$ for $m\ge m'_{\max}$, so both $\sum_{m=m'_{\max}}^{\mmax-1}N_m^{\max\mathcal V-1}\left(\frac{N_{m'_{\max}}}{N}\right)^{M\max\mathcal V}=\bigO{N^{-M\max\mathcal V}}$ and $\sum_{m=m'_{\max}}^{\mmax-1}N_m^{\max\mathcal V-1}\left(\frac{N_m}{N}\right)^{M\max\mathcal V}=\bigO{N^{-M\max\mathcal V}}$. On the other hand, the first term in the first summation of \cref{eq: sub in asymtotics} is $\bigTheta{N^{\max\mathcal V-1}}$. Therefore, assuming $\max\mathcal V>1$, we can neglect the second summation and extend the first summation from $m_{\max}'$ to $\mmax$, to leading order:
\begin{equation}
    \mathbb E\left(\left|\mathcal S_{\le\tilde m}\right|\middle|N,\vec\nu\right)=\left|\mathcal S_0\right|+\mathcal O\left[\sum_{m=0}^{\mmax-1}N_m^{\max\mathcal V-1}\left(\frac{N_m}{N}\right)^{M\max\mathcal V}\right].
\end{equation}

Now consider the second term:
\begin{equation}
    \sum_{m=0}^{\mmax-1}N_m^{\max\mathcal V-1}\left(\frac{N_m}{N}\right)^{M\max\mathcal V}=N^{\max\mathcal V-1}\sum_{m=0}^{\mmax-1}\left(1-\frac{1}{N}\sum_{i=0}^{m-1}\nu_i\right)^{\left(M+1\right)\max\mathcal V-1}.\label{eq:near integral form}
\end{equation}
Noting we can bound $\frac{1}{N}\sum_{l=0}^{m-1}\nu_l$ as follows:
\begin{equation}
    \frac{1}{N}\sum_{l=0}^{m-1}\nu_l\ge \frac{m}{N}=a\frac{m}{\mmax},\quad\text{where}\quad a\equiv\frac{\mmax}{N}=\bigTheta{1}.
\end{equation}
Note $a\in\left[0,1\right]$, using this we bound \cref{eq:near integral form} with
\begin{align}
    N^{\max\mathcal V-1}\sum_{m=0}^{\mmax-1}\left(1-\frac{1}{N}\sum_{j=0}^{m-1}\nu_j\right)^{\left(M+1\right)\max\mathcal V-1}&
    \le N^{\max\mathcal V-1}\sum_{m=0}^{\mmax-1}\left(1-a\frac{m}{\mmax}\right)^{\left(M+1\right)\max\mathcal V-1}\\
    &\sim N^{\max\mathcal V-1}\mmax\smashoperator{\int\limits_0^1}\left(1-ax\right)^{\left(M+1\right)\max\mathcal V-1}\dd{x}\equiv I.
\end{align}
Next, using the substitution $y=1-ax$ we can evaluate the integral:
\begin{align}
    I&= N^{\max\mathcal V-1}\frac{\mmax}{a}\int\limits_{1-a}^1y^{\left(M+1\right)\max\mathcal V-1}\dd{y}\\
    &=\frac{N^{\max\mathcal V-1}\mmax}{a\left(M+1\right)\max\mathcal V}\left[1-\left(1-a\right)^{\left(M+1\right)\max\mathcal V}\right].
\end{align}

Indeed, $M$ can depend on $N$, and so will have an asymptotic scaling 
$M=\bigOmega{1},\bigO{N}$. 
Thus, the asymptotic scaling of the mean is
\begin{equation}
    \mathbb E\left(\left|\mathcal S_{\le\tilde m}\right|\middle|N,\vec\nu\right)=\left|\mathcal S_0\right|+\bigO{\frac{N^{\max\mathcal V}}{M}},
\end{equation}
as $m_s=\bigO{N}$.

Following a similar method, one can show
\begin{equation}
    \mathbb E\left(\tilde m\middle|N,\vec\nu\right)=\bigO{\frac{N}{M}}.
\end{equation}

Further, inequalities \eqref{eq: mean inequality} and \eqref{eq: loss mean inequality} yield
\begin{align}
    \mathbb E\left(\left|\mathcal S_{\le\hat m}\right|\middle|N,\vec\nu\right)&=\bigO{\left|\mathcal S_0\right|+\frac{N^{\max\mathcal V}}{M}},\\
    \mathbb E\left(\hat m\middle|N,\vec\nu\right)&=\bigO{\frac{N}{M}}.
\end{align}

As these scalings are independent of $\vec v$, then conditioning our probabilities and means with some prior distribution for $\vec v$ will leave the scalings invariant:
\begin{align}
    \mathbb E\left(\left|\mathcal S_{\le\hat m}\right|\middle|N\right)&=\bigO{\left|\mathcal S_0\right|+\frac{N^{\max\mathcal V}}{M}},\\
    \mathbb E\left(\hat m\middle|N\right)&=\bigO{\frac{N}{M}}.
\end{align}

\subsection{Statistical Analysis of Dynamic-ADAPT-VQE}
\label{appendix: dynamic}

Finally, when we additionally apply layering, we have the complication that the pool size reduces throughout the layer. Thus, the upper bound on the means becomes
\begin{align}
    \mathbb E\left(\left|\mathcal S_{\le\hat m}\right|\middle|N\right)&\le\sum_{t=1}^{t_{\max}}\sum_{\substack{\vec q\in\mathcal V^t\\\colon\norm{\vec q}_1\le N}}\mathbb E\left(\left|\mathcal S_{\le\tilde m}\right|\middle|N-\norm{\vec q}_1\right)\mathbb P\left(\vec q\right)=\bigO{\left|\mathcal S_0\right|+\frac{N^{\max\mathcal V}}{M}},\\
    \mathbb E\left(\hat m\middle|N\right)&\le\sum_{t=1}^{t_{\max}}\sum_{\substack{\vec q\in\mathcal V^t\\\colon\norm{\vec q}_1\le N}}\mathbb E\left(\tilde m\middle|N-\norm{\vec q}_1\right)\mathbb P\left(\vec q\right)=\bigO{\frac{N}{M}},
\end{align}
where $\mathbb P\left(\vec q\right)$ is the frequency at which the algorithm had 
already placed $\dim\vec q$ gates in a layer each acting on $q_i$ qubits prior 
to a given iteration. $\mathbb P\left(\vec q\right)$ will decay with $\dim\vec 
q$ due to early termination or use of larger ansatz elements.

\section{Generalized Commutativity}
\label{appendix: generalized commutativity}

We can generalize the notion of commutativity used in the body of this article as follows:
\begin{definition}[Generalized Commutativity]
    $\nset_{\textrm{G}}\left(\pool, A\right)$ and its complement define valid notions of generalized non-commutation and generalized commutation over $\pool$, respectively, if:
    \begin{equation}\label{eq: general commutativity condition}
        \nset_{\textrm{O}}\left(\pool,A\right)\subseteq\nset_{\textrm{G}}\left(\pool,A\right)\quad\forall A\in\pool.
    \end{equation}
\end{definition}

In particular, note that Property \ref{prop: non-blocking} holds for generalized commutativity. This is because \cref{eq: general commutativity condition} ensures that the ordering of ansatz elements found by \cref{alg: SubpoolExploration} does not matter.

Therefore, if the shallowness of circuits is of less importance, then supersets of the support non-commuting set can be used:
\begin{equation}
    \nset_{\textrm{O}}\left(\pool,A\right)\subseteq\nset_{\textrm{G}}\left(\pool,A\right)\subseteq\nset_{\textrm{S}}\left(\pool,A\right)\quad\forall A\in\pool.
\end{equation}

Alternatively, in physical systems, one can imagine cross-talk from simultaneous qubit operations being problematic, and so a more generous notion of support, including some padding, may be beneficial:
\begin{equation}
    \nset_{\textrm{O}}\left(\pool,A\right)\subseteq\nset_{\textrm{S}}\left(\pool,A\right)\subseteq\nset_{\textrm{G}}\left(\pool,A\right)\quad\forall A\in\pool.
\end{equation}

\section{Noise models}
\label{appendix: noise models}

In this appendix, we document the details of our noise models.

\subsection{Noise channels acting on a single qubit}
\label{appendix: single qubit noise model}

This subsection describes how noise acts on a single qubit.

\textit{Amplitude damping} noise acts on the density matrix of a single-qubit as
\begin{align}
    \mathcal \damp(\gamma)
    \left[\left(
    \begin{matrix}
        \rho_{00}&\rho_{01}\\
        \rho_{10}&\rho_{11}
    \end{matrix}
    \right)\right]
    \coloneqq
    \left(
    \begin{matrix}
        \rho_{00}+\gamma\rho_{11}&\sqrt{1-\gamma}\rho_{01}\\
        \sqrt{1-\gamma}\rho_{01}&\left(1-\gamma\right)\rho_{11}
    \end{matrix}
    \right).
\end{align}
Here, the decay constant $\gamma$ is given by
\begin{align}
 \gamma(\omega_1,\dt) \coloneqq1-e^{-\omega_1\dt},
\end{align}
where $\dt$ is the duration for which the qubit is exposed to amplitude damping, and $\omega_1\coloneqq T_1^{-1}$ is determined by the $T_1$ time.
Below, we denote amplitude damping of the $r$th qubit of a density matrix $\rho$ as $\damp(\gamma(\omega_1, \dt), r)[\rho]$.

\textit{Dephasing} noise on a single qubit $r$ of a density matrix $\rho$ is modeled using the following noise channel:
\begin{align}
\dephase(p_z,r)\left[\rho\right]\coloneqq (1-p_z)\rho + p_z \pauliz(r)[\rho].
\end{align}
Here, $\pauliz(r)$ denotes the channel induced by the Pauli-$Z$ gate acting on the $r$th qubit. The phase-flip probability is given by
\begin{align}
 p_z(\omega_z, \dt)=\frac{1}{2}\left(1-e^{-\omega_z\dt}\right),
\end{align}
where $\dt$ is again the duration for which the qubit is dephasing.
The decay constant $\omega_z$, on the other hand, is set by the $T_1$ and $T_2^*$ times:
\begin{equation}
    \omega_z=\frac{1}{T_2^*}-\frac{1}{2T_1}=\omega_2^*-\frac{1}{2}\omega_1.
\end{equation}

\textit{Depolarizing} noise on a single qubit $r$ is modeled using the channel:
\begin{equation}\label{eq: asymmetric depolarizing}
    \mathcal 
\depolarize(p,r)\left[\rho\right]\coloneqq
    \left(1-p\right)\rho
    +\frac{p}{3}\sum_{\pauli\in\{\paulix, \pauliy, \pauliz\}} \pauli(r)[\rho].
\end{equation}
Here, $p \in [0,1]$ is the polarization probability, while $\pauli(r)$ are the $\paulix, \pauliy, \pauliz$ channels induced by the corresponding Pauli gates, acting on qubit $r$.

\subsection{Noise channels acting on an ansatz-element layer}
\label{appendix: layer noise model}

Next, we explain how the noise channels acting on a single qubit are used in the noisy simulation of an ansatz circuit $\Lambda_t$.
To begin with, we decompose the ansatz circuit $\Lambda_t$ into $l=1,...,L_t$ layers of support commuting ansatz-element layers $\left\{\layer_l\right\}$ as $\Lambda_t = \layer^{\text{o}}_{L_t} \circ\cdots\circ \layer^{\text{o}}_l \circ\cdots\circ \layer^{\text{o}}_1$, \cref{eq: LayerDecomposition}.
For amplitude damping and dephasing noise, each ansatz-element layer $\layer^{\text{o}}_l$ is transpiled into columns of native gates that can be implemented in parallel.
The native gate with the longest execution time of each native-gate column sets the column execution time.
The sum of the column execution times then gives the execution time $\dt_l$ of the ansatz-element layer $\layer^{\text{o}}_l$.
After each ansatz-element layer $\layer^{\text{o}}_l$ amplitude damping is implemented by applying an amplitude-damping channel to every qubit $r=1,..., N$ in an amplitude-damping layer:
\begin{align}
 \damplayer(\omega_1, \dt_l) 
   = \bigotimes_{r=1}^{N}  \damp(\gamma(\omega_1, \dt_l), r).
\end{align}
This results in the amplitude-damped ansatz circuit
\begin{align}\label{eq: damping ansatz}
 \Lambda_t(\omega_1) = \damplayer(\omega_1, 
\dt_{L_t})\circ\layer^{\text{o}}_{L_t}
                     \circ \cdots \circ
                     \damplayer(\omega_1, \dt_1)\circ\layer^{\text{o}}_1.
\end{align}

Similarly, after each ansatz-element layer $\layer^{\text{o}}_l$, dephasing is implemented by applying a dephasing channel to every qubit $r=1,..., N$ in a dephasing layer:
\begin{align}
 \dephaselayer(\omega_z, \dt_l)
    = \bigotimes_{r=1}^{N} \dephase(p_z(\omega_z, \dt_l), r).
\end{align}
This results in the dephased ansatz circuit
\begin{align}\label{eq: dephasing ansatz}
 \Lambda_t(\omega_z) = \dephaselayer(\omega_z, 
\dt_{L_t})\circ\layer^{\text{o}}_{L_t} 
                     \circ \cdots \circ
                     \dephaselayer(\omega_z, \dt_1)\circ\layer^{\text{o}}_1.
\end{align}

Finally, for depolarizing noise, we first apply the whole ansatz-element layer.
We then apply a depolarizing channel to each qubit, the exact number of times that qubit was a target qubit of a CNOT gate in the preceding layer.
More specifically, defining $M_{r,l}$ to be the number of times the $r$th qubit was the target qubit of a CNOT gate in the $l$th layer, we model the effect of depolarization after the layer $\layer^{\text{o}}_l$ by:
\begin{align}
 \depolarizelayer(p,l)
    = \bigotimes_{r=1}^{N} \left(\prod_{i=1}^{M_{l,r}} \depolarize(p,r)\right).
\end{align}
This results in the depolarized ansatz circuit 
\begin{align}\label{eq:depolarizing ansatz}
 \Lambda_t(p) = \depolarizelayer(p, L_t)\circ\layer^{\text{o}}_{L_t} \circ 
\cdots
                           \circ\depolarizelayer(p, 1)\circ\layer^{\text{o}}_1.
\end{align}

\section{Derivation of noise susceptibility relations}
\label{appendix: noise susceptibility}

In this section, we derive the noise susceptibility expressions of Eqs.~\eqref{eq: NoiseSusceptibilityExplicit}.
We start from the definition of noise susceptibility in \cref{eq: NoiseSusceptibilityDef} and differentiate the amplitude-damped ansatz circuit $\Lambda_t(\omega_1)$ \cref{eq: damping ansatz}, the dephased ansatz circuit $\Lambda_t(\omega_z)$ \cref{eq: dephasing ansatz}, and the depolarized ansatz circuit $\Lambda_t(p)$ \cref{eq:depolarizing ansatz}, with respect to the noise parameter $\omega_1$, $\omega_z$, and $p$, respectively.
This results in the following expressions
\begin{subequations}
\label{eq:CircuitDiffs} 
\begin{align}
 \left.\frac{\partial\Lambda_t(\omega_1)}{\partial \omega_1}\right|_{\omega_1=0}
& = \sum_{l=1}^{L_t}\sum_{r=1}^{N} 
        \layer^{\text{o}}_{L_t} \circ ...\circ
\left(\left.\frac{\partial\damp(\gamma,r)}{\partial\gamma}\right|_{\gamma=0}    
\times \tau_l \right)        
        \circ \layer^{\text{o}}_{l} \circ ...\circ \layer^{\text{o}}_{1},\\
 \left.\frac{\partial\Lambda_t(\omega_z)}{\partial \omega_z}\right|_{\omega_z=0}
&  = \sum_{l=1}^{L_t}\sum_{r=1}^{N} 
        \layer^{\text{o}}_{L_t} \circ ...\circ 
        \left(\left.\frac{\partial\dephase(p_z,r)}{\partial p_z}\right|_{p_z=0}
        \times \frac{\tau_l}{2} \right)        
        \circ \layer^{\text{o}}_{l} \circ ...\circ \layer^{\text{o}}_{1},\\
 \left.\frac{\partial\Lambda_t(p)}{\partial p}\right|_{p=0} 
 &  = \sum_{l=1}^{L_t}\sum_{r=1}^{N} 
        M_{l,r}\layer^{\text{o}}_{L_t} \circ ...\circ 
        \left(\left.\frac{\partial\depolarize(p,r)}{\partial p}\right|_{p=0}
        \right)
        \circ \layer^{\text{o}}_{l} \circ ...\circ \layer^{\text{o}}_{1}.
\end{align}
\end{subequations}
Next, we compute the derivatives of the individual channels, given as
\begin{subequations}
 \label{eq:OneQubitChannelDiffs}
\begin{align}
\left.\frac{\partial\damp(\gamma,r)}{\partial\gamma}\right|_{\gamma=0}
& = d\damp(r) - \mathds{1},\\
\left.\frac{\partial\dephase(p_z,r)}{\partial p_z}\right|_{p_z=0}
& = \pauliz(r) - \mathds{1}, \\
\left.\frac{\partial\depolarize(p,r)}{\partial p}\right|_{p=0}
& = \frac{1}{3}\left(\paulix(r) + \pauliy(r) + \pauliz(r) \right)
 - \mathds{1}.
\end{align}
\end{subequations}
Here we use the linear, but the no longer positive nor trace-preserving map
\begin{align}\label{eq: damping derivative map}
    d\damp(r) = \damp\left(\frac{3}{4},r\right)  + \frac{1}{4}\residual(r),
\end{align}
with the residual linear map
\begin{align}
    \residual
    \left[\left(
    \begin{matrix}
        \rho_{00}&\rho_{01}\\
        \rho_{10}&\rho_{11}
    \end{matrix}
    \right)\right]
    \coloneqq
    \left(
    \begin{matrix}
        \rho_{11} & 0          \\
        0         & -\rho_{11}
    \end{matrix}
    \right),
\end{align}
which admits to the following representation
\begin{align}
    \residual[\rho] = K_1 \rho K_1^{\dagger} - K_2 \rho K_2^{\dagger},
    \quad \text{ with }
    K_1  \coloneqq
    \left(
    \begin{matrix}
        0 & 1\\
        0 & 0
    \end{matrix}
    \right),\text{ and }
    K_2 \coloneqq
    \left(
    \begin{matrix}
        0 & 0          \\
        0 & 1
    \end{matrix}
    \right).
\end{align}
Combining the channel derivatives, Eqs.~\eqref{eq:OneQubitChannelDiffs}, back into the ansatz-circuit derivatives, Eqs.~\eqref{eq:CircuitDiffs}, and substituting these back into the noise susceptibility definition, \cref{eq: NoiseSusceptibilityDef}, we obtain
\begin{align}
 \chi_t^\damp
 &= \sum_{l=1}^{L_t}\sum_{r=1}^{N} 
       \tau_l \left( 
\trace \left[H  \layer^{\text{o}}_{L_t} \circ...\circ d\damp(r) \circ 
\layer^{\text{o}}_{l}\circ...\circ\layer^{\text{o}}_{1} [\rho_0] \right] 
- \trace \left[H  \layer^{\text{o}}_{L_t} \circ...\circ \mathds{1} \circ 
\layer^{\text{o}}_{l}\circ...\circ\layer^{\text{o}}_{1} [\rho_0] \right]
        \right),\\
 \chi_t^\dephase
&= \sum_{l=1}^{L_t}\sum_{r=1}^{N} 
       \frac{\tau_l}{2} \left( 
\trace \left[H  \layer^{\text{o}}_{L_t} \circ...\circ \pauliz(r) \circ 
\layer^{\text{o}}_{l}\circ...\circ\layer^{\text{o}}_{1} [\rho_0] \right] 
- \trace \left[H  \layer^{\text{o}}_{L_t} \circ...\circ \mathds{1} \circ 
\layer^{\text{o}}_{l}\circ...\circ\layer^{\text{o}}_{1} [\rho_0] \right]
        \right),\\
 \chi_t^\depolarize
&=\sum_{l=1}^{L_t}\sum_{r=1}^{N} M_{l,r} \!\!\!\!\!\!
 \sum_{\pauli\in\{\paulix,\pauliy,\pauliz\}} \!\!\frac{1}{3}
        \left(
        \trace\left[ H \layer^{\text{o}}_{L_t} \circ ...\circ \pauli(r)
        \circ \layer^{\text{o}}_{l} \circ ...\circ \layer^{\text{o}}_{1} 
[\rho_0] \right]
        -
        \trace\left[ H \layer^{\text{o}}_{L_t} \circ ...\circ \mathds{1}
        \circ \layer^{\text{o}}_{l} \circ ...\circ \layer^{\text{o}}_{1} 
[\rho_0] \right]
        \right).
\end{align}
To summarize these expressions in a compact form, we define energy expectation values, where the original ansatz circuit [$\Lambda_t$ in \cref{eq: LayerDecomposition}] is perturbed by a linear map $\mathcal{M}$ acting on qubit $r$ after layer $\layer^{\text{o}}_l$ as
\begin{align}
\label{eq: NoiseSusceptibilityExpectationValues}
\bound(\mathcal{M},r,l) 
= \trace\left[ H 
\layer^{\text{o}}_{L_t}\!\circ\!...\!\circ\!\layer^{\text{o}}_{l+1}
                             \!\circ\!
                  \mathcal{M}(r)\!\circ\!\layer^{\text{o}}_{l}
                  \!\circ\!...\!\circ\!
                                         \layer^{\text{o}}_{1} [\rho_0] \right].
\end{align}
As $\bound(\mathds{1},r,l)$ is identical to the noiseless energy bound $\bound_t(0)$, we can further define energy fluctuations as
\begin{align}
 \label{eq:EnergyFluctuations}
 \delta\bound(\mathcal{M},r,l) = \bound(\mathcal{M},r,l) - \bound_t(0).
\end{align}
Averaging these fluctuations over the layer and qubit indices, $l$ and $r$, respectively, further allows us to define an average noise-induced energy fluctuation for amplitude damping, dephasing, and depolarizing noise as
\begin{subequations}
\label{eq: AverageEnergyFluctuations}
\begin{align}
 d\bound(\Lambda_t,\damp)
 &= \frac{1}{L_t N}\sum_{l=1}^{L_t}\sum_{r=1}^{N} \tau_l 
\delta\bound(d\damp,r,l),\\
 d\bound(\Lambda_t,\dephase)
 &= \frac{1}{L_t N}\sum_{l=1}^{L_t}\sum_{r=1}^{N} 
 \frac{\tau_l}{2} \delta\bound(\pauliz,r,l),\\
 d\bound(\Lambda_t,\depolarize)
 &=\frac{1}{N_{II}}\sum_{l=1}^{L_t}\sum_{r=1}^{N} \frac{M_{l,r}}{3}
  \!\!\!\sum_{\pauli\in\{\paulix,\pauliy,\pauliz\}} \!\!\!\!\!\!\!
  \delta\bound(\pauli,r,l).
\end{align}
\end{subequations}
These expressions allow us to cast the noise susceptibility of amplitude damping, dephasing, and depolarizing noise in a compact form, given by Eqs.~\ref{eq: NoiseSusceptibilityExplicit}, as summarized in the body of this article.

\end{document}